\documentclass[journal]{IEEEtran}

\usepackage{amsmath}
\usepackage{amssymb}
\usepackage{amsthm}
\usepackage{bm}
\usepackage{graphicx}
\usepackage[caption=false,font=footnotesize]{subfig}
\usepackage{placeins}
\usepackage[hidelinks]{hyperref}

\theoremstyle{plain}
\newtheorem{theorem}{Theorem}
\newtheorem{lemma}{Lemma}
\newtheorem{proposition}{Proposition}
\newtheorem{corollary}{Corollary}

\theoremstyle{definition}

\begin{document}
\bstctlcite{IEEEtranBSTCTL:NoDash}

\title{Finite-Modal Realization and Operator-Norm Convergence of a Source-to-Observation Electromagnetic Scattering Green Operator}

\author{
    Zhukang Wang,~\IEEEmembership{Graduate Student Member,~IEEE,}
    Da Li,~\IEEEmembership{Member,~IEEE,}
    Ruifeng Li,
    Jinyan Ma,
    Jiarun Hu,
    Jiahui Wang,
    Anqi Xia,
    Tengjiao Wang,~\IEEEmembership{Member,~IEEE,}
    Said Mikki,
    and~Er-Ping Li,~\IEEEmembership{Fellow,~IEEE}

    \thanks{
        This work has been submitted to the IEEE for possible publication. Copyright may be transferred without notice, after which this version may no longer be accessible.

        This work was supported in part by Zhejiang Laboratory Foundation of China under Grants LR25F010001 and LHZSD25F010001 and in part by the Information Technology Center and the State Key Laboratory of CAD\&CG, Zhejiang University. (\emph{Zhukang Wang and Ruifeng Li contributed equally to this work.}) (\emph{Corresponding author: Da Li.})

        Zhukang Wang, Da Li, Ruifeng Li, Jinyan Ma, Jiahui Wang, Anqi Xia, and Er-Ping Li are with Zhejiang Key Lab of Intelligent Electromagnetics Control and Advanced Electronic Integration, and also with the College of Information Science and Electronic Engineering, Zhejiang University, Hangzhou 310027, China (e-mail: li-da@zju.edu.cn, li\_ruifeng@zju.edu.cn, jinyanma@zju.edu.cn).

        Jiarun Hu and Said Mikki are with the Zhejiang University/University of Illinois at Urbana--Champaign (ZJU-UIUC) Institute, Zhejiang University, Haining 314400, China (e-mail: mikkisaid@intl.zju.edu.cn).

        Tengjiao Wang is with the Department of Wireless Network Research, Huawei Technologies Co., Ltd., Shanghai 201206, China (e-mail: wangtengjiao6@huawei.com).
    }}

\maketitle

\begin{abstract}
    Source-to-observation operators provide reusable environment-level descriptions for multi-query electromagnetic (EM) prediction and communication-mode analysis. However, in practical multiple-scattering models, these operators are represented with finitely many angular modes, and agreement for selected excitations or between successive truncation orders does not establish uniform accuracy of the full map or reliability of its singular channels. To close this gap, we formulate the environment-induced response as a scattering Green operator on fixed continuous source and observation spaces and derive an exact trace-space factorization that reconstructs the Maxwell scattered field. For fixed, pairwise-disjoint enclosing trace spheres and a well-posed collective problem, nested vector spherical wave function (VSWF) realizations converge in operator norm. A structural bound separates external modal tails from collective-resolvent sensitivity, and operator-norm convergence guarantees uniform convergence of the singular values. We further construct a finite metric core that preserves the nonzero singular values of each finite-order operator and reconstructs matched orthonormal source--field channels without introducing external-support discretization degrees of freedom (DoF) into the spectral problem. Full-wave benchmarks verify the finite-order implementation. A controlled near-resonant two-sphere study shows that adjacent-order agreement can precede resolution of the dominant high-order collective direction. It further shows that only part of the internal amplification appears in externally accessible gains and that resonance promotes a distinct high-order channel pair above an otherwise preserved low-order family. The resulting framework provides a convergent, metric-consistent finite-modal representation of multiple-scattering source-to-observation operators and their accessible channels.
\end{abstract}


\section{Introduction}
\label{sec:introduction}

\IEEEPARstart{M}{any} electromagnetic (EM) scattering problems require reusable maps from distributed sources to prescribed observation spaces, rather than a field solution for each excitation. In nanophotonics, such source-to-observation maps provide a unified description of environment-dependent wave interactions, underpinning optical characterization and calculations of the local density of states (LDoS)~\cite{bohren1983,mishchenko2002,mishchenko2004,yannopapas2007green,egel2017}. Beyond nanophotonics, the same operator framework is central to electromagnetic information theory (EMIT)~\cite{zhu2024eit,wang2024eit,li2026eitai}, where it supports communication-mode analysis~\cite{miller2000communicating,miller2019waves}, quantification of spatial degrees of freedom (DoF)~\cite{wan2026dof}, channel state information (CSI) acquisition~\cite{li2023emit,li2024neuroute}, and channel capacity evaluation~\cite{jensen2008continuous,mikki2023shannon,zhang2026nonparaxial}. With the homogeneous-background contribution known, the key challenge is to determine the environment-induced scattering operator that captures the deterministic source-to-observation coupling introduced by structured environments.

For finite scatterer assemblies, transition-matrix (T-matrix) methods provide a natural internal realization from which such a scattering operator can be constructed~\cite{waterman1965,mishchenko1996}. In this formulation, the incident and scattered fields are represented by regular and outgoing wave expansions, respectively. The T-matrix associated with each individual scatterer maps the incident coefficients to the scattered ones, and analytical translation formulas provide the interbody coupling. This modular structure has been used for Green-tensor construction~\cite{yannopapas2007green}, EM channel modeling~\cite{li2023emit}, and EM wave manipulation~\cite{ma2025joint}. Individual T-matrices can be extracted to controlled accuracy at a fixed retained order using high-order integral-equation or hybrid solvers~\cite{ganesh2010,gimbutas2013}. However, the resulting finite source-to-observation realization also truncates the interbody translations together with the source and observation interfaces. Therefore, accuracy of the finite T-matrices does not by itself establish that the assembled map converges to the physical Maxwell scattering operator as the truncation order increases.

The remaining question is which notion of convergence is appropriate for the complete source-to-observation map. Existing literature offers several complementary forms of guidance on modal truncation. Rigorous asymptotic rates have been established for two-dimensional scalar multiple scattering by circular cylinders under prescribed incident fields~\cite{fitzpatrick2021}. For Maxwell problems, practical order-selection rules have been formulated for particular observables, including near-field radiative transfer between spheres and phase-function accuracy for individual particles~\cite{sasihithlu2011convergence,zhang2022vswftruncation}. Operator-based domain-decomposition work has also examined convergence of reduced eigencurrent solutions for prescribed illuminations and proposed an empirical a priori truncation criterion~\cite{lancellotti2010}. Separately, tunnelling escape explains the weak exterior coupling of sufficiently high angular channels~\cite{miller2025tunnelling}. Taken together, these results inform truncation-order selection for prescribed incident fields, discrete solution currents, or specific observables, but they do not bound the complete map uniformly between fixed continuous spaces. Completeness of the underlying wave-expansion basis ensures convergence for each fixed field, yet it neither gives uniform convergence over all admissible sources nor controls the order-dependent multiple-scattering inverse, or collective resolvent. Near a collective resonance, this inverse can amplify unresolved modal content, so apparent stability for selected excitations need not bound the worst-case map error. The appropriate criterion is therefore operator-norm convergence to the physical Maxwell scattering operator, which requires the observation-space error to vanish uniformly over all admissible unit-norm sources.

This distinction is especially consequential for singular-value analysis. For a scattering map equipped with prescribed source and observation norms, its singular values quantify the gains of orthogonal source-to-observation channels and characterize spatial DoF~\cite{miller2019waves,miller2012modeconverters,molesky2022toperator}. However, the singular values of a finite coefficient matrix are not intrinsically physical: they change with the normalization of the modal coefficients, even when the represented electromagnetic fields do not. Therefore, a meaningful finite-order spectrum must inherit the same measures of source strength and observed-field magnitude as the continuous scattering map. This provides a valid spectrum at each fixed order, but not convergence across orders. To ensure that the resulting channel gains approach those of the exact Maxwell map, the finite-order operators must converge in norm to one common continuous operator.

To address the issues above, we formulate the source-to-observation problem on fixed continuous spaces and restrict angular truncation to the internal realization of multiple scattering. This approach leads to the following contributions.

\begin{enumerate}
    \item First, we define a continuous EM scattering operator independent of angular truncation. It maps distributed electric and equivalent magnetic currents to observation spaces endowed with prescribed metrics. Its exact trace-space factorization reconstructs the physical scattered field under the stated uniqueness assumptions and provides a common target for all finite realizations.
    \item Building on this target, we construct nested finite VSWF realizations by truncating only the internal regular and outgoing trace variables while retaining the external spaces. For fixed, pairwise-disjoint auxiliary spheres under collective well-posedness, we prove operator-norm convergence of the complete map and derive a structural error bound. The bound separates source and observation modal tails from collective-interaction perturbations. This convergence also guarantees consistency of the associated singular spectra. To the best of our knowledge, this is the first operator-norm convergence theorem for finite-modal realizations of a three-dimensional Maxwell multiple-scattering source-to-observation map between fixed continuous spaces.
    \item For metric-consistent spectral analysis, we compress the prescribed metrics into a finite core over the retained internal modes. The core preserves every nonzero singular value without carrying external-support discretization DoF into the spectral problem, while its singular vectors reconstruct orthonormal accessible scattering channels through explicit surface-current and field formulas.
    \item Full-wave comparisons for spherical and nonspherical clusters assess the accuracy and efficiency of the finite realization. Complementary two-sphere experiments examine truncation near a collective resonance and show how internal amplification, source accessibility, and observation visibility jointly govern convergence and the externally accessible channel structure.
\end{enumerate}

\emph{Organization.} Section~\ref{sec:abstract_green_operator} formulates the continuous scattering operator and establishes its Maxwell interpretation. Section~\ref{sec:vswf_modal_realization} develops the finite VSWF realization, proves operator-norm and singular-value convergence, and constructs the metric core. Section~\ref{sec:numerical_validation_strategy} presents full-wave validation, convergence diagnostics, and accessible-channel analysis. Section~\ref{sec:conclusion} concludes the paper.

\emph{Notation and conventions.} We adopt the $\exp(-\mathrm{j}\omega t)$ time convention, where $\mathrm j^2=-1$. Bold symbols denote vectors, matrices, or vector-valued fields; calligraphic symbols denote operators or spaces; and double overbars denote dyadics. Superscripts $\mathsf T$, $*$, and $\dagger$ denote transpose, conjugation, and the Hilbert adjoint, respectively; for arrays, $\dagger$ is the Hermitian transpose. Inner products $\langle\cdot,\cdot\rangle$ are linear in the second argument, and unqualified norms are induced by the stated inner products unless specified otherwise. For an operator $\mathcal A:\mathcal X\to \mathcal Y$, $\|\mathcal A\|_{\mathcal X\to \mathcal Y}$ denotes the induced operator norm, with subscripts omitted when the spaces are clear. Singular values $\sigma_n$ are listed in nonincreasing order with multiplicity; finite singular-value lists are zero-padded when compared with infinite ones. Operator, matrix, and dyadic identities are $\mathbb I$, $\mathbf I$, and $\overline{\overline{\mathbf I}}$. The symbols $\delta(\cdot)$ and $\delta_{ab}$ denote the Dirac distribution and Kronecker delta, respectively.

\section{Hilbert-Space Source-to-Observation Scattering Green Operator}
\label{sec:abstract_green_operator}

\subsection{Single-Body Hilbert Spaces and Interface Operators}
\label{subsec:abstract_single_body_response}
\label{subsec:fixed_sphere_trace_setting}

We first fix the homogeneous reference problem. Its parameters satisfy $\omega,\varepsilon,\mu>0$, $k=\omega\sqrt{\mu\varepsilon}$, and $\eta=\sqrt{\mu/\varepsilon}$. For fields written in the observation variable $\mathbf r$, define
\begin{equation}
    \mathcal L_{k,\mathbf r}\mathbf F(\mathbf r)
    :=
    \nabla_{\mathbf r}\times\nabla_{\mathbf r}\times\mathbf F(\mathbf r)
    -
    k^2\mathbf F(\mathbf r).
    \label{eq:maxwell_operator_lk}
\end{equation}
For single-argument fields, we write $\mathcal L_k$ when the differentiation variable is clear. The outgoing free-space dyadic Green's function $\overline{\overline{\mathbf G}}^0$ is the inverse kernel of $\mathcal L_{k,\mathbf r}$:
\begin{equation}
    \mathcal L_{k,\mathbf r}
    \overline{\overline{\mathbf G}}^0(\mathbf r,\mathbf r')
    =
    \overline{\overline{\mathbf I}}\delta(\mathbf r-\mathbf r'),
    \label{eq:free_space_dyadic_inverse}
\end{equation}
where the equality is distributional and the outgoing radiation condition is implied. Explicitly,
\begin{equation}
    \overline{\overline{\mathbf G}}^0(\mathbf r,\mathbf r')
    =
    \left(
    \overline{\overline{\mathbf I}}
    +
    \frac{1}{k^2}\nabla_{\mathbf r}\nabla_{\mathbf r}
    \right)
    \frac{\exp\!\left(\mathrm{j}k|\mathbf r-\mathbf r'|\right)}
    {4\pi|\mathbf r-\mathbf r'|}.
    \label{eq:free_space_dyadic_green_definition}
\end{equation}

Impressed electric current $\mathbf J$ and equivalent magnetic current $\mathbf M$ enter Maxwell's equations as
\begin{align}
    \nabla\times\mathbf E
     & =
    \mathrm{j}\omega\mu\mathbf H
    -
    \mathbf M,
    \label{eq:maxwell_with_equiv_sources_e} \\
    \nabla\times\mathbf H
     & =
    \mathbf J
    -
    \mathrm{j}\omega\varepsilon\mathbf E.
    \label{eq:maxwell_with_equiv_sources_h}
\end{align}
Eliminating $\mathbf H$ gives
\begin{equation}
    \mathcal L_k\mathbf E
    =
    \mathrm{j}\omega\mu\mathbf J
    -
    \nabla\times\mathbf M.
    \label{eq:maxwell_operator_with_source}
\end{equation}
Here derivatives are distributional for nonsmooth sources. For a Lipschitz domain $\Omega\subset\mathbb R^3$, set $H(\operatorname{curl}^2,\Omega):=\{\mathbf E\in H(\operatorname{curl},\Omega):\nabla\times\mathbf E\in H(\operatorname{curl},\Omega)\}$, with the corresponding local regularity understood on unbounded exterior domains. Define the source-free Maxwell space
\begin{equation}
    \mathcal M_k(\Omega)
    :=
    \{\mathbf E\in H(\operatorname{curl}^2,\Omega):
    \mathcal L_k\mathbf E=\mathbf 0\}.
    \label{eq:source_free_space}
\end{equation}

We now attach a Hilbert state to one scatterer on an enclosing sphere. Let $D=\{\Omega_D,\Theta_D\}$ denote the single-body datum, where $\Omega_D\subset\mathbb R^3$ is the material support and $\Theta_D$ collects the linear material parameters and boundary or transmission conditions. Fix an enclosing ball $B_D=B_R(\mathbf c)$ with $\overline{\Omega}_D\subset B_D$ and boundary sphere $\Gamma_D=\partial B_D$. With $\mathbf n_D$ the outward unit normal, define the Cauchy trace of Maxwell fields possessing the standard tangential traces \cite{nedelec2001,buffa2002traces}
\begin{equation}
    \gamma_D^{\mathrm C}\mathbf E
    :=
    \left.
    \left(
    \mathbf n_D\times\mathbf E,
    \mathbf n_D\times k^{-1}\nabla\times\mathbf E
    \right)
    \right|_{\Gamma_D}.
    \label{eq:maxwell_cauchy_trace}
\end{equation}
We use this Cauchy datum as the interface state of a regular interior field or a radiating exterior field. This trace pair belongs to the product Maxwell trace space $H^{-1/2}(\operatorname{div},\Gamma_D)^2$, equipped with a fixed Hilbert inner product inducing the standard product Maxwell trace topology. The regular and outgoing state spaces are the corresponding trace ranges
\begin{align}
    \mathcal H^{\mathrm{reg}}(D)
     & :=
    \gamma_D^{\mathrm C}\mathcal M_k(B_D),
    \label{eq:local_incident_field_space} \\
    \mathcal H^{\mathrm{out}}(D)
     & :=
    \gamma_D^{\mathrm C}
    \left\{
    \mathbf E\in\mathcal M_k(\mathbb R^3\setminus\overline{B}_D)
    :
    \mathbf E \,\,\text{outgoing}
    \right\},
    \label{eq:local_scattered_field_space}
\end{align}
Here outgoing means the Silver--M\"uller condition \cite{chew1995,colton2013}. Maxwell Calder\'on theory identifies the two trace spaces in \eqref{eq:local_incident_field_space}--\eqref{eq:local_scattered_field_space} with the closed ranges of the bounded interior and radiating exterior projectors \cite{nedelec2001,buffa2002traces,kristensson2020}; hence they are Hilbert spaces. The Maxwell representation formula makes the complete Cauchy trace injective on the regular interior class, while the exterior representation formula together with radiating uniqueness gives the corresponding result for outgoing fields. Thus each trace determines a unique regular or outgoing field representative. Technical details are given in Supplementary Note S1.

For a regular incident field $\mathbf E^{\mathrm{inc}}$, let $\mathbf E^{\mathrm{sca}}$ be the outgoing field for which $\mathbf E^{\mathrm{inc}}+\mathbf E^{\mathrm{sca}}$ satisfies the material and boundary or transmission conditions encoded by $\Theta_D$. The corresponding trace-space transition operator is
\begin{equation}
    \mathcal T_D
    : \mathcal H^{\mathrm{reg}}(D)
    \to
    \mathcal H^{\mathrm{out}}(D),
    \quad
    \mathcal T_D
    (\gamma_D^{\mathrm C}\mathbf E^{\mathrm{inc}})
    =
    \gamma_D^{\mathrm C}\mathbf E^{\mathrm{sca}}.
    \label{eq:abstract_t_operator}
\end{equation}
We assume that the isolated-body problem is well posed in the chosen trace norms, so that $\mathcal T_D$ is bounded.

To connect the trace response to sources and measurements, fix external supports $\Xi_{\mathrm{src}}$ and $\Xi_{\mathrm{obs}}$, each either a bounded volume region or a compact Lipschitz surface support, both at positive distance from $\overline{B}_D$. They carry the corresponding volume or surface measures $\mathrm d\mu_{\mathrm{src}}$ and $\mathrm d\mu_{\mathrm{obs}}$ and finite component dimensions $d_X$ and $d_Y$. For example, a full vector field on a volume uses $\mathrm d\mu=\mathrm dV$ and $d=3$, whereas its tangential restriction to a Lipschitz surface uses $\mathrm d\mu=\mathrm dS$ and $d=2$. Thus a volume electric--magnetic source pair has $d_X=6$, a tangential Huygens surface has $d_X=4$, and electric-field observation has $d_Y=3$ or $2$ for full or tangential components. The source and observation Hilbert spaces are
\begin{equation}
    \mathcal X
    :=
    L^2(\Xi_{\mathrm{src}};\mathbb C^{d_X}),
    \quad
    \mathcal Y
    :=
    L^2(\Xi_{\mathrm{obs}};\mathbb C^{d_Y}),
    \label{eq:source_observation_hilbert_spaces}
\end{equation}
where $\mathbf q=(\mathbf J^{\mathsf T},\mathbf M^{\mathsf T})^{\mathsf T}\in\mathcal X$ is an electric/equivalent-magnetic current pair and $\mathbf E\in\mathcal Y$ is the observed field component. We use the weighted inner products
\begin{align}
    \langle \mathbf q_1,\mathbf q_2\rangle_{\mathcal X}
     & =
    \int_{\Xi_{\mathrm{src}}}
    \mathbf q_1^\dagger(\mathbf r')
    \mathbf W_X(\mathbf r')
    \mathbf q_2(\mathbf r')\,\mathrm d\mu_{\mathrm{src}}(\mathbf r'),
    \label{eq:hilbert_source_metric} \\
    \langle \mathbf E_1,\mathbf E_2\rangle_{\mathcal Y}
     & =
    \int_{\Xi_{\mathrm{obs}}}
    \mathbf E_1^\dagger(\mathbf r)
    \mathbf W_Y(\mathbf r)
    \mathbf E_2(\mathbf r)\,\mathrm d\mu_{\mathrm{obs}}(\mathbf r).
    \label{eq:hilbert_observation_metric}
\end{align}
Here $\mathbf W_X(\mathbf r')\in\mathbb C^{d_X\times d_X}$ and $\mathbf W_Y(\mathbf r)\in\mathbb C^{d_Y\times d_Y}$ are bounded, uniformly positive-definite Hermitian weights.

The source space enters the trace problem through free-space radiation. Combining the inverse kernel \eqref{eq:free_space_dyadic_inverse}--\eqref{eq:free_space_dyadic_green_definition} with the driven equation \eqref{eq:maxwell_operator_with_source} defines the free-space Green operator $\mathcal G^0$, and $\mathbf E^0:=\mathcal G^0\mathbf q$ is the corresponding radiation field:
\begin{equation}
    \begin{aligned}
        \mathbf E^0(\mathbf r)
         & =
        \mathrm{j}\omega\mu
        \int_{\Xi_{\mathrm{src}}}
        \overline{\overline{\mathbf G}}^0(\mathbf r,\mathbf r')
        \cdot\mathbf J(\mathbf r')\,\mathrm d\mu_{\mathrm{src}}(\mathbf r') \\
         & \quad
        -
        \nabla_{\mathbf r}\times
        \int_{\Xi_{\mathrm{src}}}
        \overline{\overline{\mathbf G}}^0(\mathbf r,\mathbf r')
        \cdot\mathbf M(\mathbf r')\,\mathrm d\mu_{\mathrm{src}}(\mathbf r').
    \end{aligned}
    \label{eq:background_radiation_integral}
\end{equation}
Since $\Xi_{\mathrm{src}}$ is separated from $B_D$, $\mathbf E^0$ is source-free in $B_D$, hence $\mathbf E^0\in\mathcal M_k(B_D)$. Set $\mathbf E^{\mathrm{inc}}:=\mathbf E^0=\mathcal G^0\mathbf q$ and define the source and observation maps by
\begin{align}
    \mathcal S_D
     & : \mathcal X
    \to
    \mathcal H^{\mathrm{reg}}(D),
    \quad
    \mathcal S_D\mathbf q
    =
    \gamma_D^{\mathrm C}\mathcal G^0 \mathbf q ,
    \label{eq:single_body_source_operator} \\
    \mathcal O_D
     & : \mathcal H^{\mathrm{out}}(D)
    \to
    \mathcal Y,
    \quad
    \mathcal O_D
    \left(
    \gamma_D^{\mathrm C}\mathbf E^{\mathrm{sca}}
    \right)
    =
    \left.\mathbf E^{\mathrm{sca}}\right|_{\Xi_{\mathrm{obs}}} ,
    \label{eq:single_body_observation_operator}
\end{align}
Here restriction to $\Xi_{\mathrm{obs}}$ means the selected component in $\mathcal Y$. Strict spatial separation makes the free-space kernels smooth on $\Gamma_D\times\Xi_{\mathrm{src}}$ and on $\Xi_{\mathrm{obs}}$ away from $\Gamma_D$. Hence $\mathcal S_D$ maps bounded sources into a smoother trace class that is compactly embedded in $\mathcal H^{\mathrm{reg}}(D)$. The same smoothing makes $\mathcal O_D$ compact from outgoing traces to $\mathcal Y$. The separated-support regularity and compact-embedding details for both maps are given in Supplementary Note S2. The single-body scattering Green operator is consequently the compact map
\begin{equation}
    \mathcal G_D^{\mathrm{sca}}
    =
    \mathcal O_D
    \mathcal T_D
    \mathcal S_D
    :
    \mathcal X
    \to
    \mathcal Y.
    \label{eq:single_body_scattering_green_operator}
\end{equation}

\subsection{Multiple-Scattering Foldy--Lax Assembly}
\label{subsec:abstract_green_factorization}

\begin{figure*}[t]
    \centering
    \includegraphics{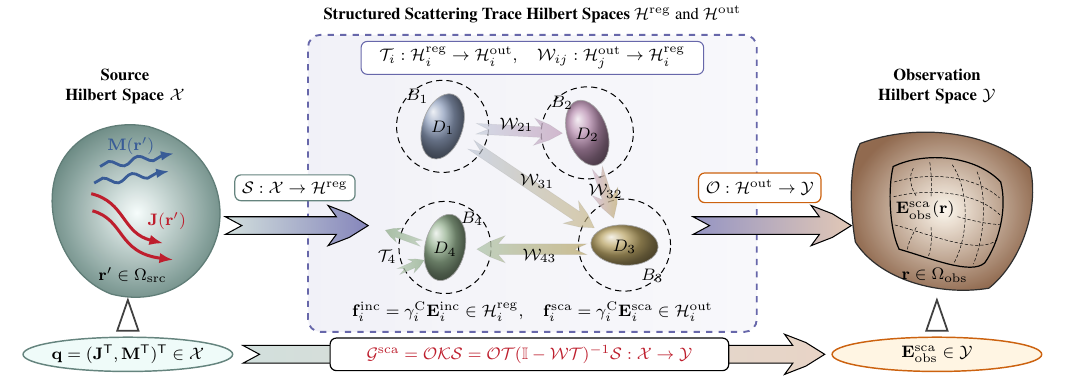}
    \caption{Source-to-observation scattering Green operator factorization. The source map $\mathcal S$ produces local regular Cauchy traces from compact electric and equivalent magnetic currents. The local operators $\mathcal T_i$ and off-diagonal continuations $\mathcal W_{ij}$ form the collective trace response. The observation map $\mathcal O$ reconstructs and superposes the solved outgoing fields on the observation support. The black dashed circles denote the enclosing-sphere boundaries $\Gamma_i=\partial B_i$, on which the Cauchy-trace states $\mathbf f_i^{\mathrm{inc}}$ and $\mathbf f_i^{\mathrm{sca}}$ are defined.}
    \label{fig:abstract_green_operator_diagram}
\end{figure*}

In multiple scattering, each body retains its own trace space, while interbody continuation maps an outgoing trace on one enclosing sphere to a regular incident trace on another. Consider $N$ compact scatterers with data $D_i=\{\Omega_i,\Theta_i\}$. Choose enclosing balls $B_i=B_{R_i}(\mathbf c_i)$ such that $\overline{\Omega}_i\subset B_i$ and the closed balls are pairwise disjoint. Set $\Gamma_i=\partial B_i$, and let $\gamma_i^{\mathrm C}$ denote the trace on $\Gamma_i$. We retain $\mathcal X$ and $\mathcal Y$, now requiring both supports to have positive distance from every $\overline{B}_i$. With $\mathcal H_i^{\mathrm{reg}}:=\mathcal H^{\mathrm{reg}}(D_i)$ and $\mathcal H_i^{\mathrm{out}}:=\mathcal H^{\mathrm{out}}(D_i)$, define the product Hilbert spaces with direct-sum inner products:
\begin{equation}
    \mathcal H^{\mathrm{reg}}
    =
    \bigoplus_{i=1}^{N}\mathcal H_i^{\mathrm{reg}},
    \quad
    \mathcal H^{\mathrm{out}}
    =
    \bigoplus_{i=1}^{N}\mathcal H_i^{\mathrm{out}}.
    \label{eq:abstract_product_spaces}
\end{equation}
The isolated responses form the bounded block-diagonal operator
\begin{equation}
    \mathcal T
    =
    \bigoplus_{i=1}^{N}
    \mathcal T_i
    :
    \mathcal H^{\mathrm{reg}}
    \to
    \mathcal H^{\mathrm{out}}.
    \label{eq:abstract_block_t_operator}
\end{equation}

For $i\neq j$ and $\gamma_j^{\mathrm C}\mathbf E_j\in\mathcal H_j^{\mathrm{out}}$, the represented outgoing field $\mathbf E_j$ is source-free in $B_i$. Define
\begin{equation}
    \mathcal W_{ij}
    : \mathcal H_j^{\mathrm{out}}
    \to
    \mathcal H_i^{\mathrm{reg}},
    \quad
    \mathcal W_{ij}
    (\gamma_j^{\mathrm C}\mathbf E_j)
    :=
    \gamma_i^{\mathrm C}
    \mathbf E_j.
    \label{eq:abstract_pair_translation_action}
\end{equation}
Set $\mathcal W_{ii}=0$, because self-interaction is already contained in $\mathcal T_i$. The blocks assemble into the interbody coupling operator
\begin{equation}
    \mathcal W
    =
    \left[
        \mathcal W_{ij}
        \right]_{i,j=1}^{N}
    :
    \mathcal H^{\mathrm{out}}
    \to
    \mathcal H^{\mathrm{reg}}.
    \label{eq:abstract_block_translation_action}
\end{equation}
Because the trace spheres are separated for $i\neq j$, the outgoing representative from $\Gamma_j$ is smooth in a neighborhood of $\Gamma_i$. Therefore $\mathcal W_{ij}$ maps bounded sets into a smoother trace class whose embedding in $\mathcal H_i^{\mathrm{reg}}$ is compact. Hence $\mathcal W_{ij}$ is compact, and the finite block operator $\mathcal W$ is compact. The detailed compactness argument is given in Supplementary Note S2.

We next describe how the external source acts on the multiple-scattering system. Its free-space radiation induces regular Cauchy traces on the enclosing spheres, which define the local and stacked source maps
\begin{align}
    \mathcal S_i
     & : \mathcal X
    \to
    \mathcal H_i^{\mathrm{reg}},
    \quad
    \mathcal S_i\mathbf q
    =
    \gamma_i^{\mathrm C}\mathcal G^0\mathbf q,
    \label{eq:abstract_source_projection_local} \\
    \mathcal S
     & : \mathcal X
    \to
    \mathcal H^{\mathrm{reg}},
    \quad
    \mathcal S\mathbf q
    =
    \left(
    \mathcal S_i\mathbf q
    \right)_{i=1}^{N}.
    \label{eq:abstract_source_projection_operator}
\end{align}
Each $\mathcal S_i$ is compact by the single-body argument, so the finite direct sum $\mathcal S$ is compact.

With these pieces fixed, for each body $i$ define the direct-source trace $\mathbf f_i^{\mathrm{exc}}:=\mathcal S_i\mathbf q$, the collective incident trace $\mathbf f_i^{\mathrm{inc}}:=\gamma_i^{\mathrm C}\mathbf E_i^{\mathrm{inc}}$, and the outgoing scattered trace $\mathbf f_i^{\mathrm{sca}}:=\gamma_i^{\mathrm C}\mathbf E_i^{\mathrm{sca}}$. Stack them as $\mathbf f^{\mathrm{exc}}$, $\mathbf f^{\mathrm{inc}}$, and $\mathbf f^{\mathrm{sca}}$. Then $\mathbf f^{\mathrm{exc}},\mathbf f^{\mathrm{inc}}\in\mathcal H^{\mathrm{reg}}$ and $\mathbf f^{\mathrm{sca}}\in\mathcal H^{\mathrm{out}}$. The local response and Foldy--Lax balance are, respectively \cite{foldy1945,lax1951},
\begin{align}
    \mathbf f^{\mathrm{sca}}
     & =
    \mathcal T\mathbf f^{\mathrm{inc}},
    \label{eq:abstract_block_local_response} \\
    \mathbf f^{\mathrm{inc}}
     & =
    \mathbf f^{\mathrm{exc}}
    +
    \mathcal W\mathbf f^{\mathrm{sca}}
    =
    \mathcal S\mathbf q
    +
    \mathcal W\mathbf f^{\mathrm{sca}}.
    \label{eq:abstract_component_foldy_lax}
\end{align}
Eliminating $\mathbf f^{\mathrm{sca}}$ from \eqref{eq:abstract_block_local_response} and \eqref{eq:abstract_component_foldy_lax} gives
\begin{equation}
    \left(\mathbb I-\mathcal W\mathcal T\right)
    \mathbf f^{\mathrm{inc}}
    =
    \mathcal S\mathbf q,
    \label{eq:abstract_incident_system}
\end{equation}
where $\mathbb I$ is the identity on $\mathcal H^{\mathrm{reg}}$. Let $\mathcal A:=\mathcal W\mathcal T$. Since $\mathcal W$ is compact and $\mathcal T$ is bounded, $\mathcal A$ is compact, so $\mathbb I-\mathcal A$ is Fredholm of index zero \cite{kato1995}. Assume the homogeneous radiating multiple-scattering Maxwell problem is unique at the frequency considered \cite{nedelec2001,colton2013}. A kernel element reconstructs a homogeneous radiating multiple-scattering solution: collective uniqueness makes the superposed field zero, bodywise outgoing uniqueness then makes every scattered field and trace vanish, and the homogeneous Foldy--Lax balance makes every incident trace vanish. Thus $\ker(\mathbb I-\mathcal A)=\{\mathbf0\}$, and the Fredholm alternative and the bounded inverse theorem make $\mathbb I-\mathcal A$ boundedly invertible. The detailed field-reconstruction and kernel arguments are given in Supplementary Note S3. Define the collective resolvent and cluster transition by
\begin{equation}
    \mathcal R
    :=
    \left(\mathbb I-\mathcal W\mathcal T\right)^{-1},
    \quad
    \mathcal K
    :=
    \mathcal T\mathcal R
    =
    \mathcal T\left(\mathbb I-\mathcal W\mathcal T\right)^{-1}.
    \label{eq:abstract_collective_resolvent_transition}
\end{equation}
Here $\mathcal R:\mathcal H^{\mathrm{reg}}\to\mathcal H^{\mathrm{reg}}$ and $\mathcal K:\mathcal H^{\mathrm{reg}}\to\mathcal H^{\mathrm{out}}$ are bounded, and the solved scattered trace is $\mathbf f^{\mathrm{sca}}=\mathcal K\mathcal S\mathbf q$.

Finally, the observation map continues and superposes the local outgoing fields on $\Xi_{\mathrm{obs}}$:
\begin{equation}
    \mathcal O
    \colon \mathcal H^{\mathrm{out}}
    \to
    \mathcal Y,
    \quad
    \mathcal O\mathbf f^{\mathrm{sca}}
    =
    \sum_{i=1}^{N}
    \left.\mathbf E_i^{\mathrm{sca}}\right|_{\Xi_{\mathrm{obs}}}
    =:
    \mathbf E_{\mathrm{obs}}^{\mathrm{sca}}.
    \label{eq:abstract_observation_operator}
\end{equation}
The observation--sphere separations make $\mathcal O$ compact by the same smooth-continuation and compact-embedding argument used for $\mathcal O_D$. Applying $\mathcal O$ to the solved scattered trace gives
\begin{equation}
    \mathcal G^{\mathrm{sca}}
    :=
    \mathcal O\mathcal K
    \mathcal S
    :
    \mathcal X
    \to
    \mathcal Y,
    \quad
    \mathcal G^{\mathrm{sca}}\mathbf q
    =
    \mathbf E_{\mathrm{obs}}^{\mathrm{sca}}.
    \label{eq:abstract_scattering_green_operator}
\end{equation}
Because $\mathcal K$ is bounded and $\mathcal S$ and $\mathcal O$ are compact, $\mathcal G^{\mathrm{sca}}$ is compact. Equivalently, $\mathcal G^{\mathrm{sca}}=\mathcal O\mathcal T(\mathbb I-\mathcal W\mathcal T)^{-1}\mathcal S$, which separates source accessibility, collective scattering, and observation visibility. Fig.~\ref{fig:abstract_green_operator_diagram} summarizes the resulting flow through the external and trace spaces.

The following proposition identifies this Hilbert-space factorization with the physical solution of the Maxwell multiple-scattering problem.

\begin{proposition}[Exactness of the trace assembly]
    \label{prop:exact_trace_factorization_equivalence}
    Under the isolated-body and collective uniqueness assumptions above, for every $\mathbf q\in\mathcal X$, the solutions of \eqref{eq:abstract_block_local_response}--\eqref{eq:abstract_component_foldy_lax} are exactly the Cauchy traces induced by the physical multiple-scattering solution. Consequently, \eqref{eq:abstract_scattering_green_operator} returns the observed component of the physical scattered field on $\Xi_{\mathrm{obs}}$.
\end{proposition}

\begin{proof}
    The unique trace solution reconstructs regular incident fields and radiating bodywise scattered fields. The Foldy--Lax balance identifies each incident field with the free-space excitation plus the fields scattered by the other bodies, while the local transition relations enforce the prescribed material or boundary conditions. Their superposition is therefore the unique physical scattering solution; bodywise outgoing uniqueness identifies its scattered components, after which the Foldy--Lax balance identifies the incident traces. A complete proof is given in Supplementary Note S3.
\end{proof}

If the free-space source-to-observation map $\mathcal G^0:\mathcal X\to\mathcal Y$ is bounded, the corresponding total-field operator is simply
\begin{equation}
    \mathcal G^{\mathrm{tot}}
    =
    \mathcal G^0
    +
    \mathcal G^{\mathrm{sca}}.
    \label{eq:green_operator_additive_split}
\end{equation}

\section{Finite VSWF Realization of the Scattering Green Operator}
\label{sec:vswf_modal_realization}

Section~\ref{sec:abstract_green_operator} leaves the untruncated Green operator in an infinite-dimensional trace form. A numerical realization should approximate only this internal representation, without changing the prescribed source and observation spaces or their metrics. VSWFs provide nested modal coordinates on the enclosing spheres, making truncation explicit and independent of external sampling. They therefore yield a sequence of finite source-to-observation realizations that can be compared with the same continuous Maxwell operator and evaluated under the same external metrics.

\subsection{Local VSWF Coordinates}
\label{subsec:vswf}

At each local center, regular and outgoing VSWFs represent incident and scattered traces, respectively. With $\mathbf r$ measured from the local center, we use the solenoidal convention of \cite{doicu2006}. The real functions $\widehat P_l^{|m|}$ are normalized by $\int_{-1}^{1}|\widehat P_l^{|m|}(\mu)|^2\,\mathrm d\mu=1$, and appear through the angular factor $\widehat P_l^{|m|}(\cos\theta)\mathrm{e}^{\mathrm{j}m\phi}$. The radial branches are $z_l^{\mathrm{reg}}=j_l$ and $z_l^{\mathrm{out}}=h_l^{(1)}$, where $l\geq1$ and $-l\leq m\leq l$. For $s\in\{\mathrm{reg},\mathrm{out}\}$,
\begin{align}
    \mathbf U_{lm}^{s}(\mathbf r)
     & =
    \frac{1}{\sqrt{2l(l+1)}}\nabla\times
    \left[\mathbf r\,z_l^{s}(kr)\widehat P_l^{|m|}(\cos\theta)\mathrm{e}^{\mathrm{j}m\phi}\right],
    \label{eq:vswf_u} \\
    \mathbf V_{lm}^{s}(\mathbf r)
     & =
    \frac{1}{k}\nabla\times\mathbf U_{lm}^{s}(\mathbf r).
    \label{eq:vswf_v}
\end{align}
The longitudinal family is absent for source-free homogeneous fields. Write $\bm{\Phi}_{1lm}^{s}=\mathbf U_{lm}^{s}$ and $\bm{\Phi}_{2lm}^{s}=\mathbf V_{lm}^{s}$, and use $\alpha=(\tau,l,m)$ for modal labels:
\begin{align}
    \Lambda
     & =
    \{(\tau,l,m):\tau=1,2,\ l\geq1,\ -l\leq m\leq l\},
    \label{eq:physical_modal_set} \\
    \Lambda_L
     & =
    \{\alpha=(\tau,l,m)\in\Lambda:l\leq L\}.
    \label{eq:truncated_modal_set}
\end{align}
Set $\bm{\Phi}_{\alpha}^{s}=\bm{\Phi}_{\tau lm}^{s}$ for $\alpha=(\tau,l,m)$, and define
\begin{equation}
    \overline\alpha=(\tau,l,-m),
    \quad
    \alpha^\star=(\varsigma(\tau),l,m),
    \label{eq:modal_label_involutions}
\end{equation}
where $\varsigma(1)=2$ and $\varsigma(2)=1$. The involutions commute and preserve $\Lambda_L$. Infinite series are limits of partial sums over $\Lambda_L$. Finite vectors use the polarization-major map
\begin{align}
    u=\xi_L(\tau,l,m)
     & =
    (\tau-1)L(L+2)+l(l+1)+m,
    \label{eq:vswf_single_index} \\
    U_L
     & =
    2L(L+2),
    \label{eq:vswf_total_modes}
\end{align}
for $(\tau,l,m)\in\Lambda_L$, so $u=1,\ldots,U_L$. For finite-modal label $u$, we write $\bm{\Phi}_u^{s}=\bm{\Phi}_{\xi_L^{-1}(u)}^{s}$. If $\xi_L^{-1}(u)=\alpha$, then $\overline u=\xi_L(\overline\alpha)$ and $u^\star=\xi_L(\alpha^\star)$. The truncated expansions are
\begin{equation}
    \mathbf E_L^{\mathrm{inc}}(\mathbf r)
    =
    \sum_{u=1}^{U_L}a_u^{\mathrm{inc}}\bm{\Phi}_u^{\mathrm{reg}}(\mathbf r),
    \quad
    \mathbf E_L^{\mathrm{sca}}(\mathbf r)
    =
    \sum_{u=1}^{U_L}a_u^{\mathrm{sca}}\bm{\Phi}_u^{\mathrm{out}}(\mathbf r).
    \label{eq:vswf_modal_expansion}
\end{equation}
The regular branch is finite at the center, whereas the outgoing branch is singular there and satisfies the radiation condition.

\subsection{VSWF Trace Projections}
\label{subsec:modal_internal_operators}

The coordinate approximation is made on traces, not on arbitrary volume fields. For body $i$, branch $s\in\{\mathrm{reg},\mathrm{out}\}$, and $\alpha\in\Lambda$, trace the local basis field by setting
$\bm\phi_{i,\alpha}^{s}:=\gamma_i^{\mathrm C}[\bm\Phi_\alpha^{s}(\,\cdot-\mathbf c_i)]$ and define the retained synthesis map
\begin{equation}
    P_{i,L}^{s}:\mathbb C^{U_L}\to\mathcal H_i^s,
    \quad
    P_{i,L}^{s}\mathbf a
    :=
    \sum_{u=1}^{U_L}
    a_u\bm\phi_{i,\xi_L^{-1}(u)}^{s}.
    \label{eq:local_modal_trace_synthesis}
\end{equation}
Thus $P_{i,L}^s$ maps raw, trace-unnormalized VSWF amplitudes to Cauchy traces.

We now make explicit the fixed trace inner product referred to in Section~\ref{subsec:fixed_sphere_trace_setting}. On each sphere $\Gamma_i$, it is the VSH-diagonal product inner product, normalized by $1/(2\eta)$, which induces the standard $H^{-1/2}(\operatorname{div},\Gamma_i)^2$ topology \cite{kristensson2020,buffa2002traces}. The VSWF traces form a complete orthogonal family in this metric \cite{kristensson2020,chew1995,doicu2006}. Equipping $\mathbb C^{U_L}$ with the Euclidean inner product, define
\begin{equation}
    \mathbf G_{i,L}^{s}
    :=
    (P_{i,L}^{s})^\dagger P_{i,L}^{s}
    \in \mathbb{C}^{U_L\times U_L}.
    \label{eq:local_modal_trace_gram}
\end{equation}
For $u=\xi_L(\tau,l,m)$, set $x_i:=kR_i$, $d_l^s(x):=[xz_l^s(x)]'/x$, and $w_{l,i}^{\pm}:=(1+l(l+1)/x_i^2)^{\pm1/2}$. The Cauchy-trace formulas in this chosen metric give, for $u'=1,\ldots,U_L$,
\begin{equation}
    [\mathbf G_{i,L}^{s}]_{u,u'}
    =
    \delta_{u u'}
    \frac{\pi R_i^2}{2\eta}
    \left[
        w_{l,i}^{+}|z_l^s(x_i)|^2
        +
        w_{l,i}^{-}|d_l^s(x_i)|^2
        \right].
    \label{eq:trace_gram_entries}
\end{equation}
Equation~\eqref{eq:trace_gram_entries} is the Gram matrix obtained by restricting the fixed infinite-dimensional trace metric to the retained raw VSWF traces. Each diagonal weight is positive and depends on the mode through $l$ and on $R_i$, $k$, and $s$, but not on the ambient truncation order $L$. Hence $\mathbf G_{i,L}^s$ is Hermitian positive definite, and the trace analysis map is
\begin{equation}
    Q_{i,L}^{s}
    :=
    (\mathbf G_{i,L}^{s})^{-1}(P_{i,L}^{s})^\dagger
    :
    \mathcal H_i^s \to \mathbb C^{U_L}.
    \label{eq:local_modal_trace_analysis}
\end{equation}
It satisfies $Q_{i,L}^sP_{i,L}^s=\mathbf I$, while $P_{i,L}^sQ_{i,L}^s$ is the orthogonal projection onto $\operatorname{ran}P_{i,L}^s$. If a trace $\mathbf f\in\mathcal H_i^s$ is expanded as
$\mathbf f=\sum_{\alpha\in\Lambda}a_\alpha\bm\phi_{i,\alpha}^s$, then
\begin{equation}
    [Q_{i,L}^s\mathbf f]_u
    =a_{\xi_L^{-1}(u)},
    \quad u=1,\ldots,U_L.
    \label{eq:raw_modal_coefficient_extraction}
\end{equation}

Taking direct sums gives
\begin{equation}
    P_L^s
    =\bigoplus_{i=1}^{N}P_{i,L}^s,
    \quad
    Q_L^s
    =\bigoplus_{i=1}^{N}Q_{i,L}^s,
    \quad
    \Pi_L^s
    =P_L^sQ_L^s .
    \label{eq:stacked_modal_projection_operators}
\end{equation}
Thus $Q_L^sP_L^s=\mathbf I$, and $\Pi_L^s:\mathcal H^s\to\mathcal H^s$ is the orthogonal projection onto $\operatorname{ran}P_L^s$. The sets $\Lambda_L$ make these ranges nested, and VSWF completeness makes their union dense, so
\begin{equation}
    \|\Pi_L^s\|=1,
    \quad
    \Pi_L^s
    \xrightarrow[L\to\infty]{\mathrm{strong}}
    \mathbb I .
    \label{eq:modal_projection_stability}
\end{equation}
The VSH metric, the trace calculation leading to \eqref{eq:trace_gram_entries}, and the completeness and projection arguments underlying \eqref{eq:modal_projection_stability} are detailed in Supplementary Note S4.

\subsection{Finite-Modal Scattering Green Operator}
\label{subsec:finite_modal_scattering_operator}

We now replace each trace-space factor by retained coordinates. The projected local T-matrix is
\begin{equation}
    \mathbf T_i^{(L)}
    =
    Q_{i,L}^{\mathrm{out}}
    \mathcal T_i
    P_{i,L}^{\mathrm{reg}}
    \in \mathbb C^{U_L\times U_L}.
    \label{eq:vswf_t_operator_specialization}
\end{equation}
Its columns are obtained from retained regular-VSWF excitations and outgoing extraction or exterior-field fitting. The isolated-body solve may be analytic, numerical, or externally supplied \cite{ganesh2010,harrington1968,taflove2005,jin2014}. Extraction, fitting, or discretization errors in a supplied T-matrix are separate numerical errors.

All local VSWF frames are taken parallel to the global frame. If a body-fixed frame is used, the local T-matrix is first rotated into this convention. For translation from body $j$ to $i$, with $\mathbf d=\mathbf c_i-\mathbf c_j$, the addition theorem gives \cite{cruzan1962,egel2017}
\begin{equation}
    \bm{\Phi}_{\alpha'}^{\mathrm{out}}(\mathbf r+\mathbf d)
    =
    \lim_{L\to\infty}
    \sum_{\alpha\in\Lambda_L}
    A_{\alpha',\alpha}^{\mathrm{out}}(\mathbf d)
    \bm{\Phi}_{\alpha}^{\mathrm{reg}}(\mathbf r),
    \quad
    \alpha'\in\Lambda.
    \label{eq:vswf_addition_theorem}
\end{equation}
Within its convergence domain, the block from body $j$ to $i$ is
\begin{equation}
    \left[\mathbf W_{ij}^{(L)}\right]_{u,u'}
    =
    A_{\xi_L^{-1}(u'),\xi_L^{-1}(u)}^{\mathrm{out}}(\mathbf c_i-\mathbf c_j),
    \quad
    i\neq j,
    \label{eq:pair_translation_matrix_entries}
\end{equation}
and $\mathbf W_{ii}^{(L)}=\mathbf 0$, equivalently $\mathbf W_{ij}^{(L)}=Q_{i,L}^{\mathrm{reg}}\mathcal W_{ij}P_{j,L}^{\mathrm{out}}$. Here $u'$ is outgoing at $\mathbf c_j$ and $u$ regular at $\mathbf c_i$. The condition $R_i+R_j<|\mathbf c_i-\mathbf c_j|$ suffices for convergence.

Accordingly, the assembled internal matrices are
\begin{equation}
    \mathbf T_L
    =
    \bigoplus_{i=1}^{N}\mathbf T_i^{(L)},
    \quad
    [\mathbf W_L]_{ij}
    =
    \begin{cases}
        \mathbf 0,            & i=j,     \\
        \mathbf W_{ij}^{(L)}, & i\neq j,
    \end{cases}.
    \label{eq:block_diagonal_t_matrix}
\end{equation}
With $N_L:=NU_L$, both $\mathbf T_L$ and $\mathbf W_L$ belong to $\mathbb C^{N_L\times N_L}$. Let $\mathbf a^{\mathrm{exc}},\mathbf a^{\mathrm{inc}},\mathbf a^{\mathrm{sca}}\in\mathbb C^{N_L}$ denote the stacked external-excitation, collective-incident, and scattered coefficients. Projecting the abstract balance onto the retained modal subspaces gives
\begin{equation}
    \left(\mathbf I-\mathbf W_L\mathbf T_L\right)\mathbf a^{\mathrm{inc}}
    =
    \mathbf a^{\mathrm{exc}},
    \quad
    \mathbf a^{\mathrm{sca}}=\mathbf T_L\mathbf a^{\mathrm{inc}}.
    \label{eq:global_foldy_lax_system}
\end{equation}
Assume the retained system is nonsingular at the order considered. Then
\begin{equation}
    \mathbf R_L
    :=
    \left(\mathbf I-\mathbf W_L\mathbf T_L\right)^{-1},
    \quad
    \mathbf K_L
    :=
    \mathbf T_L\mathbf R_L.
    \label{eq:finite_internal_resolvent_transition}
\end{equation}
Thus $\mathbf a^{\mathrm{sca}}=\mathbf K_L\mathbf a^{\mathrm{exc}}$; for a fixed scene and frequency, $\mathbf K_L$ is independent of the source. Its lifted transition is
\begin{equation}
    \widehat{\mathcal K}_L
    :=
    P_L^{\mathrm{out}}\mathbf K_LQ_L^{\mathrm{reg}}.
    \label{eq:lifted_finite_transition}
\end{equation}

Next, we consider the external excitation field generated by $\mathbf q\in\mathcal X$. Since $\Xi_{\mathrm{src}}$ is separated from $B_i$, this field has the regular expansion about $\mathbf c_i$:
\begin{equation}
    \mathbf E^0(\mathbf c_i+\mathbf r)
    =
    \lim_{L\to\infty}
    \sum_{\alpha\in\Lambda_L}
    a_{i,\alpha}^{0}
    \bm{\Phi}_{\alpha}^{\mathrm{reg}}(\mathbf r),
    \quad
    |\mathbf r|<R_i.
    \label{eq:local_external_excitation_expansion}
\end{equation}
The retained source map $\mathcal S_L:=Q_L^{\mathrm{reg}}\mathcal S$ records the finite coefficient vector of this direct excitation. The following lemma gives these coefficients as source pairings.

\begin{lemma}[Source-to-modal excitation]
    \label{lem:continuous_source_modal_excitation}
    Let $\mathbf q\in\mathcal X$ have components $\mathbf J$ and $\mathbf M$. Assume $R_i<\operatorname{dist}(\mathbf c_i,\Xi_{\mathrm{src}})$ for each body. The exact background coefficients in \eqref{eq:local_external_excitation_expansion} are, for $\alpha\in\Lambda$,
    \begin{equation}
        \begin{aligned}
            a_{i,\alpha}^{0}
             & =
            -\frac{\omega\mu k}{\pi}
            \int_{\Xi_{\mathrm{src}}}
            \bm{\Phi}_{\overline\alpha}^{\mathrm{out}}(\mathbf r'-\mathbf c_i)
            \cdot\mathbf J(\mathbf r')\,\mathrm d\mu_{\mathrm{src}}(\mathbf r') \\
             & \quad
            -
            \frac{\mathrm{j}k^2}{\pi}
            \int_{\Xi_{\mathrm{src}}}
            \bm{\Phi}_{\overline{\alpha^\star}}^{\mathrm{out}}(\mathbf r'-\mathbf c_i)
            \cdot\mathbf M(\mathbf r')\,\mathrm d\mu_{\mathrm{src}}(\mathbf r').
        \end{aligned}
        \label{eq:continuous_source_modal_projection}
    \end{equation}
    Retaining $l\leq L$ gives $\mathcal S_L \mathbf q=\mathbf a^{\mathrm{exc}}\in\mathbb C^{N_L}$, where
    $a_{i,u}^{\mathrm{exc}}:=a_{i,\xi_L^{-1}(u)}^{0}$ for $u=1,\ldots,U_L$.
\end{lemma}

\begin{proof}
    See Appendix~\ref{app:continuous_source_projection}.
\end{proof}

After $\mathbf K_L$ maps the direct excitation to scattered coefficients, the observation map synthesizes their outgoing field. For $\mathbf r\in\Xi_{\mathrm{obs}}$, the selected observation component of $\mathcal O_L:=\mathcal O P_L^{\mathrm{out}}$ is
\begin{equation}
    \left(\mathcal O_L\mathbf a^{\mathrm{sca}}\right)(\mathbf r)
    =
    \sum_{i=1}^{N}\sum_{u=1}^{U_L}
    a_{i,u}^{\mathrm{sca}}\bm\Phi_u^{\mathrm{out}}(\mathbf r-\mathbf c_i)
    =
    \mathbf E_{\mathrm{obs},L}^{\mathrm{sca}}(\mathbf r).
    \label{eq:global_scattered_field_reconstruction}
\end{equation}
Combining source analysis, modal propagation, and observation synthesis gives the finite-modal scattering operator
\begin{equation}
    \mathcal G_L^{\mathrm{sca}}
    :=
    \mathcal O_L\mathbf K_L\mathcal S_L
    =
    \mathcal O\widehat{\mathcal K}_L\mathcal S
    :\mathcal X\to\mathcal Y,
    \quad
    \mathcal G_L^{\mathrm{sca}}\mathbf q
    =
    \mathbf E_{\mathrm{obs},L}^{\mathrm{sca}}.
    \label{eq:finite_scattering_green_operator}
\end{equation}
Whenever the bounded free-space map in Section~\ref{sec:abstract_green_operator} is available, the total field is $\mathbf E_{\mathrm{obs},L}^{\mathrm{tot}}=\mathbf E_{\mathrm{obs}}^0+\mathbf E_{\mathrm{obs},L}^{\mathrm{sca}}$. Only the scattered contribution is modal-truncated.

\subsection{Finite-Modal Operator Convergence}
\label{subsec:modal_green_convergence}

The finite matrices must approximate the same Maxwell scattering map defined in Section~\ref{sec:abstract_green_operator}, uniformly over the chosen source Hilbert space. By Proposition~\ref{prop:exact_trace_factorization_equivalence}, $\mathcal G^{\mathrm{sca}}$ returns the physical scattered field under the stated uniqueness assumptions. The following theorem establishes operator-norm convergence of the retained VSWF realization and separates the effects of modal tails and collective resolvent sensitivity.

\begin{theorem}[Asymptotic convergence of the modal scattering Green operator]
    \label{thm:modal_green_operator_convergence}
    Under the separated fixed-sphere setting and the collective well-posedness condition above, define
    \begin{equation}
        \mathcal A=\mathcal W\mathcal T,
        \quad
        \mathcal A_L
        =
        \Pi_L^{\mathrm{reg}}\mathcal W
        \Pi_L^{\mathrm{out}}\mathcal T
        \Pi_L^{\mathrm{reg}},
        \quad
        \delta_L=\|\mathcal A_L-\mathcal A\|.
        \label{eq:lifted_collective_coupling}
    \end{equation}
    Then $\delta_L\to0$. Hence, for all sufficiently large $L$, $\|\mathcal R\|\delta_L<1$, $\mathbf I-\mathbf W_L\mathbf T_L$ is invertible, and the lifted resolvent $\mathcal R_L^\sharp=(\mathbb I-\mathcal A_L)^{-1}$ satisfies
    \begin{equation}
        \|\mathcal R_L^\sharp\|
        \leq
        \frac{\|\mathcal R\|}{1-\|\mathcal R\|\delta_L},
        \quad
        \|\mathcal R_L^\sharp-\mathcal R\|
        \leq
        \frac{\|\mathcal R\|^2\delta_L}
        {1-\|\mathcal R\|\delta_L}.
        \label{eq:collective_resolvent_error_bound}
    \end{equation}
    Moreover, $\widehat{\mathcal K}_L\to\mathcal K$ strongly and is uniformly bounded for all sufficiently large $L$. With
    \begin{equation}
        \epsilon_{S,L}
        =\|(\mathbb I-\Pi_L^{\mathrm{reg}})\mathcal S\|,
        \quad
        \epsilon_{O,L}
        =\|\mathcal O(\mathbb I-\Pi_L^{\mathrm{out}})\|,
        \label{eq:source_observation_modal_tails}
    \end{equation}
    both tails tend to zero. Define the structural finite-$L$ majorant
    \begin{equation}
        \begin{aligned}
            E_L
            :={} &
            \epsilon_{O,L}\|\mathcal T\|\|\mathcal R_L^\sharp\|\|\mathcal S\| \\
                 & +\|\mathcal O\|\|\mathcal T\|
            \|\mathcal R_L^\sharp-\mathcal R\|\|\mathcal S\|                  \\
                 & +\|\mathcal O\|\|\mathcal T\|\|\mathcal R\|\epsilon_{S,L}.
        \end{aligned}
        \label{eq:green_operator_convergence_bound}
    \end{equation}
    Then
    $\|\mathcal G_L^{\mathrm{sca}}-\mathcal G^{\mathrm{sca}}\|_{\mathcal X\to\mathcal Y}\leq E_L\to0$.
\end{theorem}

\begin{proof}
    See Appendix~\ref{app:modal_green_convergence_proof}.
\end{proof}

The additive decomposition in \eqref{eq:green_operator_additive_split} is inherited by the truncated realization. Whenever $\mathcal G^0:\mathcal X\to\mathcal Y$ is bounded as specified in Section~\ref{sec:abstract_green_operator}, set $\mathcal G_L^{\mathrm{tot}}:=\mathcal G^0+\mathcal G_L^{\mathrm{sca}}$. Since $\mathcal G^0$ is unchanged by modal truncation,
\begin{equation}
    \|\mathcal G_L^{\mathrm{tot}}-\mathcal G^{\mathrm{tot}}\|_{\mathcal X\to\mathcal Y}
    =
    \|\mathcal G_L^{\mathrm{sca}}-\mathcal G^{\mathrm{sca}}\|_{\mathcal X\to\mathcal Y}
    \leq
    E_L
    \to
    0.
    \label{eq:total_green_operator_convergence_bound}
\end{equation}
Moreover, since the scattering maps are compact, the same norm convergence uniformly stabilizes their singular values. Weyl's inequality gives \cite{simon2005}
\begin{equation}
    \sup_{n\geq1}\left|
    \sigma_n(\mathcal G_L^{\mathrm{sca}})
    -
    \sigma_n(\mathcal G^{\mathrm{sca}})
    \right|
    \leq
    \|\mathcal G_L^{\mathrm{sca}}-\mathcal G^{\mathrm{sca}}\|_{\mathcal X\to\mathcal Y}
    \to
    0.
    \label{eq:singular_value_stability_from_norm_convergence}
\end{equation}

The bound separates three convergence mechanisms. The first and third terms are the observation and source modal tails; their sizes reflect the high-order trace content sampled or excited by the external supports. Through the resolvent estimate in \eqref{eq:collective_resolvent_error_bound}, the middle term is controlled by $\delta_L$, the error in the full projected product $\mathcal W\mathcal T$, and is amplified by the collective resolvent. It therefore includes incident-space truncation, high-order outgoing response generated by $\mathcal T$, and subsequent translation by $\mathcal W$, rather than translation error alone. A sensitive resolvent can magnify these unresolved components, but compactness by itself does not predict the order at which the asymptotic regime begins.

\subsection{Finite-Modal Operator Metrics}
\label{subsec:finite_operator_metrics}

Theorem~\ref{thm:modal_green_operator_convergence} is an infinite-dimensional consistency result, whereas the numerical diagnostics require norms of retained finite operators. One route discretizes both external supports and forms a large weighted transfer matrix or generalized eigenvalue problem \cite{miller2000communicating,miller2019waves,golub2013matrix}. Instead, \eqref{eq:finite_scattering_green_operator} compresses the prescribed external metrics onto the retained modal coordinates. The following theorem states this finite realization.

\begin{theorem}[Metric realization of the finite scattering operator]
    \label{thm:finite_operator_metric_realization}
    Assume $\mathbf I-\mathbf W_L\mathbf T_L$ is nonsingular at order $L$. Define the source and observation Gramians by
    \begin{equation}
        \mathbf C_{S,L}
        =
        \mathcal S_L\mathcal S_L^\dagger,
        \quad
        \mathbf C_{O,L}
        =
        \mathcal O_L^\dagger\mathcal O_L.
        \label{eq:finite_source_observation_grams}
    \end{equation}
    Both Gramians are $N_L\times N_L$ Hermitian positive semidefinite matrices and depend on the chosen source and observation Hilbert metrics. All square roots below denote the Hermitian positive-semidefinite square roots. Set
    \begin{equation}
        \mathbf B_L
        :=
        \mathbf C_{O,L}^{1/2}\mathbf K_L\mathbf C_{S,L}^{1/2}.
        \label{eq:finite_metric_core}
    \end{equation}
    Then $\mathcal G_L^{\mathrm{sca}}$ and $\mathbf B_L$ have the same nonzero singular values:
    \begin{equation}
        \sigma_n(\mathcal G_L^{\mathrm{sca}})
        =
        \sigma_n(\mathbf B_L).
        \label{eq:finite_metric_singular_equivalence}
    \end{equation}
    In particular,
    \begin{equation}
        \|\mathcal G_L^{\mathrm{sca}}\|_{\mathcal X\to\mathcal Y}
        =
        \sigma_1(\mathbf B_L).
        \label{eq:finite_operator_norm_metric_core}
    \end{equation}
\end{theorem}

\begin{proof}
    Apply the finite polar decomposition \cite{simon2005} to the two outer factors in $\mathcal G_L^{\mathrm{sca}}=\mathcal O_L\mathbf K_L\mathcal S_L$. Since $\mathbf C_{S,L}=\mathcal S_L\mathcal S_L^\dagger$ and $\mathbf C_{O,L}=\mathcal O_L^\dagger\mathcal O_L$, there are partial isometries $\mathcal V_{S,L}$ and $\mathcal V_{O,L}$ such that $\mathcal S_L=\mathbf C_{S,L}^{1/2}\mathcal V_{S,L}$ and $\mathcal O_L=\mathcal V_{O,L}\mathbf C_{O,L}^{1/2}$. Therefore $\mathcal G_L^{\mathrm{sca}}=\mathcal V_{O,L}\mathbf B_L\mathcal V_{S,L}$. The square-root factors make $\mathbf B_L$ act only between the initial and final spaces of these partial isometries, so the nonzero singular values are unchanged by the two isometric embeddings. The norm identity follows from the standard equality between the induced Hilbert-space operator norm and the largest singular value.
\end{proof}

\begin{corollary}[Accessible scattering channels]
    \label{cor:accessible_scattering_channels}
    For fixed $L$, the singular value decomposition of $\mathbf B_L$ is $\mathbf B_L=\mathbf U\mathbf\Sigma\mathbf V^\dagger$, where $\mathbf U=[\mathbf u_1, \ldots, \mathbf u_{N_L}]$, $\mathbf V=[\mathbf v_1, \ldots, \mathbf v_{N_L}]$, and $\mathbf\Sigma=\operatorname{diag}(\sigma_1, \ldots, \sigma_{N_L})$. For any $1\leq n\leq N_L$ with $\sigma_n>0$, define
    \begin{equation}
        \mathbf q_n
        :=
        \mathcal S_L^\dagger
        \bigl(\mathbf C_{S,L}^{1/2}\bigr)^+
        \mathbf v_n,
        \quad
        \widehat{\mathbf E}_{n}^{\mathrm{sca}}
        :=
        \mathcal O_L
        \bigl(\mathbf C_{O,L}^{1/2}\bigr)^+
        \mathbf u_n,
        \label{eq:accessible_scattering_channel_pair}
    \end{equation}
    where $+$ denotes the Moore--Penrose inverse. The reconstructed source and observation families are orthonormal, i.e., $\langle\mathbf q_n,\mathbf q_{n^\prime}\rangle_{\mathcal X}=\langle\widehat{\mathbf E}_{n}^{\mathrm{sca}},\widehat{\mathbf E}_{n^\prime}^{\mathrm{sca}}\rangle_{\mathcal Y}=\delta_{n n^\prime}$. Each matched pair further satisfies the channel action
    \begin{equation}
        \mathcal G_L^{\mathrm{sca}}\mathbf q_n
        =
        \sigma_n\widehat{\mathbf E}_{n}^{\mathrm{sca}}.
        \label{eq:accessible_scattering_channel_action}
    \end{equation}
    These pairs are the source-accessible, observation-visible scattering channels under the prescribed metrics.
\end{corollary}

\begin{proof}
    For $\sigma_n>0$, the SVD gives $\mathbf B_L\mathbf v_n=\sigma_n\mathbf u_n$ and $\mathbf B_L^\dagger\mathbf u_n=\sigma_n\mathbf v_n$. By \eqref{eq:finite_metric_core}, $\mathbf v_n\in\operatorname{ran}\mathbf C_{S,L}^{1/2}$ and $\mathbf u_n\in\operatorname{ran}\mathbf C_{O,L}^{1/2}$. Let $\mathcal V_{S,L}$ and $\mathcal V_{O,L}$ be the partial isometries used above. Since the square-root pseudoinverse products project orthogonally onto these ranges, \eqref{eq:accessible_scattering_channel_pair} becomes $\mathbf q_n=\mathcal V_{S,L}^\dagger\mathbf v_n$ and $\widehat{\mathbf E}_n^{\mathrm{sca}}=\mathcal V_{O,L}\mathbf u_n$. The columns of $\mathbf U$ and $\mathbf V$ are orthonormal, and the partial isometries preserve inner products on their support spaces; hence both reconstructed families are orthonormal. Finally, $\mathcal V_{S,L}\mathcal V_{S,L}^\dagger\mathbf v_n=\mathbf v_n$, so substituting this identity and $\mathbf B_L\mathbf v_n=\sigma_n\mathbf u_n$ into $\mathcal G_L^{\mathrm{sca}}=\mathcal V_{O,L}\mathbf B_L\mathcal V_{S,L}$ proves \eqref{eq:accessible_scattering_channel_action}.
\end{proof}

Thus retained scattering-operator norms can be evaluated by an $N_L\times N_L$ singular-value computation after the source and observation metrics are included. For all operator-level results below, we use the tangential surface-to-surface instance of \eqref{eq:hilbert_source_metric}--\eqref{eq:hilbert_observation_metric}. The source surface $\Sigma_{\mathrm{src}}$ carries $\mathbf q=(\mathbf J_t^{\mathsf T},\mathbf M_t^{\mathsf T})^{\mathsf T}$, and the observation surface $\Sigma_{\mathrm{obs}}$ carries the tangential electric field. Thus $\mathrm d\mu_{\mathrm{src}}=\mathrm d\mu_{\mathrm{obs}}=\mathrm dS$, $d_X=4$, and $d_Y=2$. In local tangential coordinates, the weights are
\begin{equation}
    \mathbf W_{X,\Sigma}
    =
    \frac12
    \begin{bmatrix}
        \eta\mathbf I_t & \mathbf 0            \\
        \mathbf 0       & \eta^{-1}\mathbf I_t
    \end{bmatrix},
    \quad
    \mathbf W_{Y,\Sigma}
    =
    \frac{1}{2\eta}\mathbf I_t,
    \label{eq:surface_metric_weights}
\end{equation}
where $\mathbf I_t$ acts on the two tangential components. For $i=1,\ldots,N$ and $u=1,\ldots,U_L$, define the finite cluster index
\begin{equation}
    p
    =
    \zeta_L(i,u)
    =
    (i-1)U_L+u,
    \quad
    p=1,\ldots,N_L.
    \label{eq:finite_cluster_index}
\end{equation}
Set $\bm\Psi_p(\mathbf r):=\bm\Phi_u^{\mathrm{out}}(\mathbf r-\mathbf c_i)$ and $\bm\Psi_{p,t}:=\mathbf P_t\bm\Psi_p$, where $\mathbf P_t(\mathbf r):=\mathbf I-\hat{\mathbf n}(\mathbf r)\hat{\mathbf n}^{\mathsf T}(\mathbf r)$. The relabelings $\overline p:=\zeta_L(i,\overline u)$ and $p^\star:=\zeta_L(i,u^\star)$ inherit \eqref{eq:modal_label_involutions}. The source-side pairings follow by specializing Lemma~\ref{lem:continuous_source_modal_excitation} to $\Xi_{\mathrm{src}}=\Sigma_{\mathrm{src}}$, while the observation side follows from \eqref{eq:global_scattered_field_reconstruction}. With the Hilbert metrics \eqref{eq:hilbert_source_metric}--\eqref{eq:hilbert_observation_metric} specialized by \eqref{eq:surface_metric_weights}, evaluating \eqref{eq:finite_source_observation_grams} and using $\omega\mu=k\eta$ gives
\begin{align}
    [\mathbf C_{S,L}]_{p,q}
     & =
    \frac{2\eta k^4}{\pi^2}
    \int_{\Sigma_{\mathrm{src}}}
    (
    \bm\Psi_{\overline p,t}^{\mathsf T}
    \bm\Psi_{\overline q,t}^{*}
    +
    \bm\Psi_{\overline{p^\star},t}^{\mathsf T}
    \bm\Psi_{\overline{q^\star},t}^{*}
    )\,\mathrm dS,
    \label{eq:surface_source_gram_entries} \\
    [\mathbf C_{O,L}]_{p,q}
     & =
    \frac{1}{2\eta}
    \int_{\Sigma_{\mathrm{obs}}}
    \bm\Psi_{p,t}^{\dagger}
    \bm\Psi_{q,t}\,\mathrm dS.
    \label{eq:surface_observation_gram_entries}
\end{align}
Consequently, for fixed $(\Sigma_{\mathrm{src}},\Sigma_{\mathrm{obs}})$ and retained core $\mathbf K_L$, \eqref{eq:finite_operator_norm_metric_core} gives the norm of $\mathcal G_L^{\mathrm{sca}}$ without selecting a particular source distribution or retaining source- or observation-support quadrature points as DoF of the spectral problem.

For reconstruction of the channels in Corollary~\ref{cor:accessible_scattering_channels}, set
\begin{equation}
    \mathbf h_{n,S}
    :=
    \bigl(\mathbf C_{S,L}^{1/2}\bigr)^+\mathbf v_n,
    \quad
    \mathbf h_{n,O}
    :=
    \bigl(\mathbf C_{O,L}^{1/2}\bigr)^+\mathbf u_n.
    \label{eq:accessible_channel_unwhitened_coordinates}
\end{equation}
Comparing the source pairings in \eqref{eq:continuous_source_modal_projection} with the Hilbert-adjoint identity under \eqref{eq:surface_metric_weights} gives $\mathcal S_L^\dagger\mathbf h_{n,S}$, while \eqref{eq:global_scattered_field_reconstruction} gives $\mathcal O_L\mathbf h_{n,O}$. Consequently, with $\mathbf q_{n,t}=(\mathbf J_{n,t}^{\mathsf T},\mathbf M_{n,t}^{\mathsf T})^{\mathsf T}$, \eqref{eq:accessible_scattering_channel_pair} becomes
\begin{align}
    \mathbf J_{n,t}(\mathbf r')
     & =
    -\frac{2k^2}{\pi}
    \sum_{p=1}^{N_L}
    [\mathbf h_{n,S}]_p
    \bm\Psi_{\overline p,t}^{*}(\mathbf r'),
    \label{eq:accessible_channel_electric_current} \\
    \mathbf M_{n,t}(\mathbf r')
     & =
    \frac{2\mathrm j\eta k^2}{\pi}
    \sum_{p=1}^{N_L}
    [\mathbf h_{n,S}]_p
    \bm\Psi_{\overline{p^\star},t}^{*}(\mathbf r'),
    \label{eq:accessible_channel_magnetic_current} \\
    \widehat{\mathbf E}_{n,t}^{\mathrm{sca}}(\mathbf r)
     & =
    \sum_{p=1}^{N_L}
    [\mathbf h_{n,O}]_p
    \bm\Psi_{p,t}(\mathbf r).
    \label{eq:accessible_channel_observation_mode}
\end{align}
The interface-Gram and channel-recovery calculations leading to \eqref{eq:surface_source_gram_entries}--\eqref{eq:surface_observation_gram_entries} and \eqref{eq:accessible_channel_electric_current}--\eqref{eq:accessible_channel_observation_mode} are given in Supplementary Note S5. These surface-metric formulas underlie the finite-order operator diagnostics in Section~\ref{subsec:numerical_operator_convergence} and the explicit accessible-channel reconstructions in Section~\ref{subsec:accessible_scattering_channels}.

\section{Full-Wave Validation and Operator Analysis}
\label{sec:numerical_validation_strategy}

The preceding theory establishes asymptotic convergence, but it does not quantify implementation error at practical orders, convergence rates in sensitive configurations, or the physical form of the resulting scattering channels. We therefore use three complementary numerical assessments. Full-wave comparisons test the source-to-field pipeline for representative excitations and mixed clusters. Operator-norm experiments probe worst-case truncation behavior under fixed Hilbert metrics, keeping implementation error distinct from finite-reference modal convergence. Finally, accessible-channel analysis connects the metric core to optimal source--field pairs.

\subsection{Full-Wave Implementation Tests}
\label{subsec:full_wave_validation}

The present full-wave comparisons assess the practical finite-modal implementation, not operator-norm convergence: they evaluate the source-to-field response for prescribed excitations and therefore do not control the worst-case error over the source unit ball. We accordingly fix $L=3$, which sufficiently resolves the selected electrically small, nonresonant bodies, and vary both the mixed-scatterer composition and the excitation type to test the implementation across diverse configurations. PEC spheres use analytic diagonal Mie T-matrices \cite{doicu2006,bohren1983,mishchenko2002}. The locally centered T-matrices of a cube, an ellipsoid, and a torus are extracted with SCUFF-EM \cite{reid2015scuff} in vacuum at $1\,\mathrm{GHz}$. Table~\ref{tab:nonspherical_tmatrix_data} summarizes their material and surface-discretization settings, extraction times, and independent plane-wave replay errors. The surface meshes, extraction protocol, convention conversion, and complete diagnostics are given in Supplementary Note S6.

\begin{table}[!t]
    \caption{Nonspherical local T-matrix extraction and replay data}
    \label{tab:nonspherical_tmatrix_data}
    \centering
    \footnotesize
    \begin{tabular}{c c c c c c}
        \hline
        Body      & $\varepsilon_r$   & Triangles & Time (s) & E replay & H replay \\
        \hline
        Cube      & $4$               & 726       & 39.43    & $1.63\%$ & $0.61\%$ \\
        Ellipsoid & $10$              & 908       & 64.33    & $2.46\%$ & $0.97\%$ \\
        Torus     & $6+0.5\mathrm{j}$ & 318       & 27.90    & $3.44\%$ & $2.42\%$ \\
        \hline
    \end{tabular}
\end{table}

\begin{figure}[!t]
    \centering
    \includegraphics{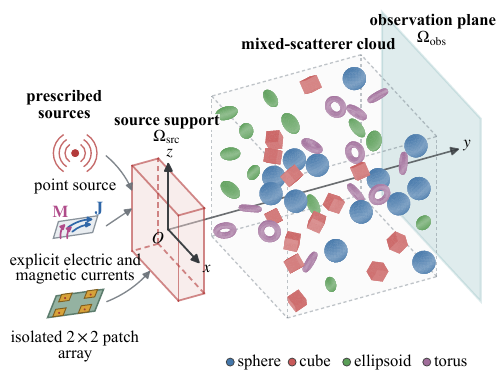}
    \caption{Full-wave implementation-test configuration at $1\,\mathrm{GHz}$.}
    \label{fig:fullwave_validation_scene}
\end{figure}

Fig.~\ref{fig:fullwave_validation_scene} shows the wavelength-normalized configuration. The tests use one deterministic family of nested $5\leq N\leq40$ mixed clusters, comprising PEC spheres with $ka=1$ and the three nonspherical bodies in Table~\ref{tab:nonspherical_tmatrix_data}; every pair of enclosing spheres has a surface gap of at least $0.2\lambda$. COMSOL Multiphysics 6.4 supplies the full-wave reference. The computational domain, PML, mesh, solver, field sampling, nested-scene construction, phasor and constitutive-parameter conversions, and timing protocol are specified in Supplementary Note S7.

Three prescribed excitations probe distinct source interfaces: a distributional electric point current, a compactly supported Gaussian electric--magnetic surface-current pair, and the frozen Huygens currents of an isolated, equal-phase $2\times2$ patch array. The point current is used only as an implementation check and lies outside the $L^2$ source space $\mathcal X$. For the third case, one empty-environment COMSOL solve fixes $\mathbf J_H=\hat{\mathbf n}\times\mathbf H^0$ and $\mathbf M_H=-\hat{\mathbf n}\times\mathbf E^0$ on the source box; the same trace data are then retained for every loaded scene. Thus this case tests scattering from a prescribed Huygens source, not self-consistent antenna loading or port back-action. The exact source definitions and the two solvers' respective source discretizations are given in Supplementary Note S7.

\begin{figure}[!t]
    \centering
    \subfloat[]{\includegraphics{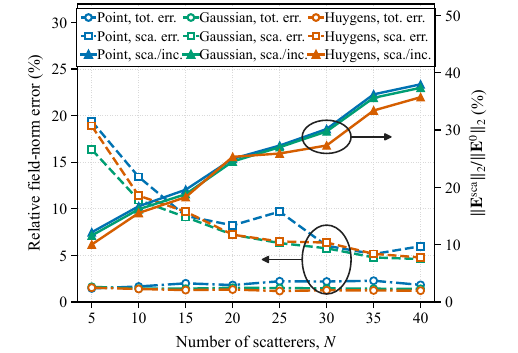}\label{fig:fullwave_electric_scaling_a}}
    \par
    \subfloat[]{\includegraphics{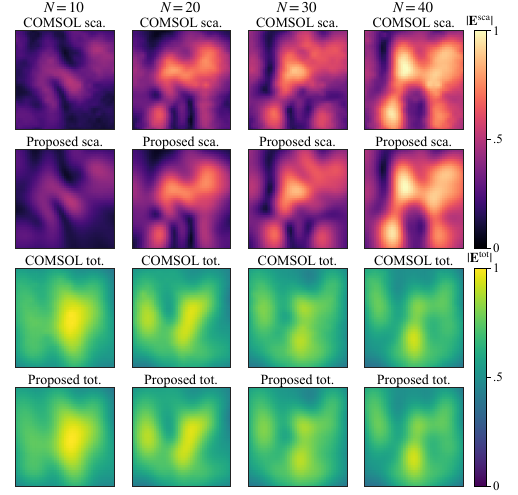}\label{fig:fullwave_electric_scaling_b}}
    \caption{Electric-field implementation tests. (a) Total- and scattered-field relative discrete field-norm errors and $\|\mathbf E^{\mathrm{sca}}\|_2/\|\mathbf E^0\|_2$ for the three excitations. (b) Frozen-Huygens-current field-amplitude maps. The scattered-field maps and the total-field maps are normalized by their respective group maxima.}
    \label{fig:fullwave_electric_scaling}
\end{figure}

Fig.~\ref{fig:fullwave_electric_scaling} quantifies implementation agreement for the three sources over the fixed nested cases; the curves are deterministic comparisons, not ensemble averages. The total-electric-field error remains below $2.3\%$, and its largest excess above the corresponding empty-scene baseline is $0.11$ percentage points. As $\|\mathbf E^{\mathrm{sca}}\|_2/\|\mathbf E^0\|_2$ grows from $10.0\%$--$12.1\%$ at $N=5$ to $35.7\%$--$38.0\%$ at $N=40$, the scattered-field error falls from $16.4\%$--$19.4\%$ to $4.6\%$--$6.0\%$. These are source-to-field implementation comparisons interpreted against an empty-scene interface baseline, not operator-norm or isolated modal-truncation errors. The Huygens maps in Fig.~\ref{fig:fullwave_electric_scaling}\subref{fig:fullwave_electric_scaling_b} show the same spatial agreement as the scattering contribution grows; the exact discrete metrics and phasor mapping are given in Supplementary Note S7.

\begin{table}[!t]
    \captionsetup{position=top}
    \caption{Single-run runtime comparison and Huygens timing breakdown in seconds}
    \label{tab:runtime_comparison_breakdown}
    \centering
    \footnotesize
    \subfloat[Online per-scene wall time: proposed/COMSOL {[speedup]}]{%
    \label{tab:end_to_end_runtime}
    \setlength{\tabcolsep}{5pt}
    \begin{tabular*}{0.9\columnwidth}{@{\extracolsep{\fill}} c c c c}
        \hline
        $N$ & Point & Gaussian & Huygens \\
        \hline
        10 & 0.54/95.0 [$175\times$] & 1.72/88.2 [$51\times$] & 13.41/422.5 [$32\times$] \\
        20 & 1.11/170.6 [$154\times$] & 3.52/160.3 [$46\times$] & 25.47/521.0 [$20\times$] \\
        30 & 1.98/272.0 [$137\times$] & 5.59/263.0 [$47\times$] & 37.95/593.9 [$16\times$] \\
        40 & 3.98/338.0 [$85\times$] & 8.90/344.6 [$39\times$] & 55.96/724.7 [$13\times$] \\
        \hline
    \end{tabular*}
    }
    \par
    \subfloat[Proposed Huygens timing breakdown]{%
        \label{tab:huygens_timing_breakdown}
        \setlength{\tabcolsep}{5pt}
        \begin{tabular*}{0.9\columnwidth}{@{\extracolsep{\fill}} c r r r r r r}
            \hline
            $N$ & \multicolumn{1}{c}{$t_{\mathrm{set}}$} & \multicolumn{1}{c}{$t_{\mathcal S}$} & \multicolumn{1}{c}{$t_{\mathrm{sol}}$} & \multicolumn{1}{c}{$t_{\mathcal O}$} & \multicolumn{1}{c}{$t_{\mathrm{tot}}$} & \multicolumn{1}{c}{$t_{\mathcal S}/t_{\mathrm{tot}}$} \\
            \hline
            10 & 0.145 & 11.110 & 0.273 & 1.881 & 13.410 & $82.8\%$ \\
            20 & 0.175 & 22.512 & 0.780 & 1.999 & 25.467 & $88.4\%$ \\
            30 & 0.204 & 33.961 & 1.740 & 2.046 & 37.950 & $89.5\%$ \\
            40 & 0.255 & 50.023 & 3.378 & 2.299 & 55.956 & $89.4\%$ \\
            \hline
        \end{tabular*}
    }
\end{table}

Table~\ref{tab:runtime_comparison_breakdown} reports single-run wall-clock measurements on an Intel Core i7-14700 workstation with $32\,\mathrm{GB}$ of memory; the precise timer boundaries and offline exclusions are given in Supplementary Note S7. With precomputed local T-matrices, the speedups over $N=10$--$40$ are $85\times$--$175\times$, $39\times$--$51\times$, and $13\times$--$32\times$ for the point, Gaussian, and Huygens sources, respectively. At $N=40$, the Huygens online times are $55.96\,\mathrm s$ and $724.7\,\mathrm s$. Direct CPU pairing of the 3456 electric--magnetic surface-current samples with $U_L=30$ VSWFs at every scatterer is the bottleneck: source projection accounts for $82.8\%$--$89.5\%$ of the proposed runtime, whereas the collective solve takes $0.27$--$3.38\,\mathrm s$. All modal solves reach relative residuals below $7.9\times10^{-10}$.

\begin{figure}[!t]
    \centering
    \includegraphics{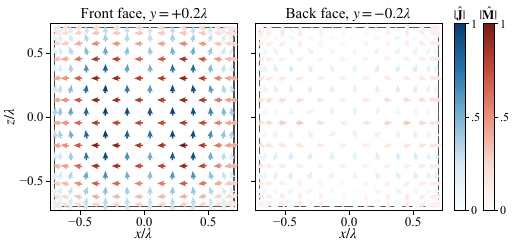}
    \caption{Equivalent Huygens currents extracted from the isolated $2\times2$ patch array. Arrows show the instantaneous tangential directions at a common reference phase. Colors show the separately normalized phasor magnitudes of $\mathbf J_H$ and $\mathbf M_H$ across the front and back faces.}
    \label{fig:huygens_front_back_currents}
\end{figure}

Fig.~\ref{fig:huygens_front_back_currents} displays the directions and normalized magnitudes of $\mathbf J_H$ and $\mathbf M_H$ on the two broadside faces $y=\pm0.2\lambda$ of the closed source support. The nonzero $\mathbf M_H$ motivates retaining equivalent magnetic currents in the source map. The frozen Huygens pair reproduces the isolated radiator's prescribed exterior incident field without transferring its internal conduction or volume currents to each scene.

\subsection{Modal Operator Convergence and Resolvent Sensitivity}
\label{subsec:numerical_operator_convergence}

The structural majorant $E_L$ in Theorem~\ref{thm:modal_green_operator_convergence} separates source tails, collective stability, and observation tails. Its first and third terms contain the observation and source projection tails, respectively, while its middle term measures the collective-resolvent perturbation. We diagnose these factors through finite high-order references rather than evaluating the majorant itself. A finite order $L_\star$ is accepted as a reference only after the retained operator and the diagnostics of interest are stable, to the accuracy required for the comparison, under further increases of $L_\star$; it then serves as a numerically resolved surrogate for the limiting system.

Let $L_\star$ be such a finite reference. Assuming $\|\mathcal G_{L_\star}^{\mathrm{sca}}\|_{\mathcal X\to\mathcal Y}>0$, define the relative worst-case discrepancy and its adjacent-order counterpart by
\begin{equation}
    e_{L\mid L_\star}
    :=
    \frac{
    \|\mathcal G_{L_\star}^{\mathrm{sca}}
    -
    \mathcal G_L^{\mathrm{sca}}\|_{\mathcal X\to\mathcal Y}
    }{
    \|\mathcal G_{L_\star}^{\mathrm{sca}}\|_{\mathcal X\to\mathcal Y}
    },
    \quad
    \widehat e_L
    :=
    e_{L\mid L+1}.
    \label{eq:numerical_finite_reference_relative_error}
\end{equation}
Thus $e_{L\mid L_\star}$ measures the largest relative change visible between orders $L$ and $L_\star$ over all unit-norm prescribed sources, whereas $\widehat e_L$ tests only the next truncation order and can miss content omitted by both adjacent systems.

The finite-reference source and observation tails are
\begin{align}
    \kappa_{S,L\mid L_\star}
     & :=
    \frac{
    \|(\Pi_{L_\star}^{\mathrm{reg}}-\Pi_L^{\mathrm{reg}})
    \mathcal S\|_{\mathcal X\to\mathcal H^{\mathrm{reg}}}
    }{
    \|\Pi_{L_\star}^{\mathrm{reg}}
    \mathcal S\|_{\mathcal X\to\mathcal H^{\mathrm{reg}}}
    },
    \label{eq:finite_reference_source_tail_definition} \\
    \kappa_{O,L\mid L_\star}
     & :=
    \frac{
    \|\mathcal O
    (\Pi_{L_\star}^{\mathrm{out}}-\Pi_L^{\mathrm{out}})
    \|_{\mathcal H^{\mathrm{out}}\to\mathcal Y}
    }{
    \|\mathcal O
    \Pi_{L_\star}^{\mathrm{out}}
    \|_{\mathcal H^{\mathrm{out}}\to\mathcal Y}
    }.
    \label{eq:finite_reference_observation_tail_definition}
\end{align}
They quantify the fractions of the reference-resolved source-accessible regular trace and observation-visible outgoing trace, respectively, that are carried by orders above $L$ and are therefore omitted at order $L$.

To monitor internal collective amplification, let $\mathbf G_{r,L}=\bigoplus_i\mathbf G_{i,L}^{\mathrm{reg}}$ be the fixed regular-trace Gramian from \eqref{eq:trace_gram_entries} and define
\begin{equation}
    \rho_L
    :=
    \|\widetilde{\mathbf R}_L\|_2
    =
    \|\mathbf G_{r,L}^{1/2}
    \mathbf R_L
    \mathbf G_{r,L}^{-1/2}\|_2.
    \label{eq:numerical_trace_normalized_resolvent}
\end{equation}
This is the largest amplification of a retained incident trace measured in the fixed trace metric, equivalently the reciprocal smallest singular value of the metric-normalized collective system.

These finite-reference quantities connect directly to the three terms of $E_L$. For fixed $L$, compactness of the source and observation maps gives, as $L_\star\to\infty$,
\begin{equation}
    \kappa_{S,L\mid L_\star}
    \to
    \frac{\epsilon_{S,L}}{\|\mathcal S\|},
    \quad
    \kappa_{O,L\mid L_\star}
    \to
    \frac{\epsilon_{O,L}}{\|\mathcal O\|},
    \label{eq:finite_reference_tail_limits}
\end{equation}
whenever the corresponding limiting denominator is nonzero. Moreover, with $\delta_L=\|\mathcal A_L-\mathcal A\|$ as in \eqref{eq:lifted_collective_coupling},
\begin{align}
    \|\mathcal R_L^\sharp\|
     & =
    \max\{1,\rho_L\},
    \label{eq:numerical_lifted_resolvent_norm} \\
    \|\mathcal R_L^\sharp-\mathcal R\|
     & \leq
    \max\{1,\rho_L\}\,\delta_L\|\mathcal R\|.
    \label{eq:numerical_resolvent_perturbation_diagnostic}
\end{align}
Hence $\kappa_{O,L\mid L_\star}$ diagnoses the observation-tail factor in the first term of $E_L$, $\rho_L$ gives its resolvent factor exactly and quantifies the amplification of $\delta_L$ underlying the middle term, and $\kappa_{S,L\mid L_\star}$ diagnoses the source-tail factor in the third term. Likewise, $e_{L\mid L_\star}$ tends to $\|\mathcal G_L^{\mathrm{sca}}-\mathcal G^{\mathrm{sca}}\|/\|\mathcal G^{\mathrm{sca}}\|$ when the limiting operator is nonzero, and therefore checks the combined end-to-end effect of all three mechanisms. These diagnostics do not evaluate $E_L$ itself, but they are finite-reference realizations of precisely the quantities that govern it. Detailed derivations are given in Supplementary Note S8.

\begin{figure}[!t]
    \centering
    \subfloat[]{\includegraphics{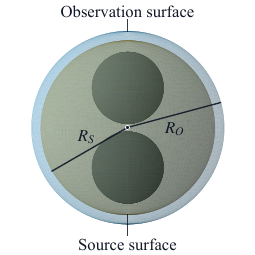}\label{fig:operator_convergence_setup_a}}\hfill
    \subfloat[]{\includegraphics{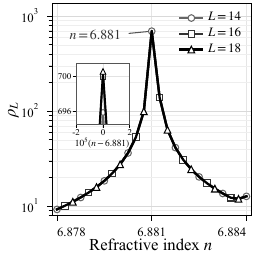}\label{fig:operator_convergence_setup_b}}
    \caption{Two-sphere configuration and sampled resolvent peak. (a) Surface-radius convention; the source and observation surfaces coincide in the fixed-surface test. (b) Refractive-index scan of $\rho_L$ for $L=14,16,18$; the inset magnifies the separation among the three curves at $n=6.881$.}
    \label{fig:operator_convergence_setup}
\end{figure}

The following finite-reference test examines how collective resolvent amplification affects practical truncation. We use two identical nonmagnetic lossless dielectric spheres with $\mu_r=1$, $\varepsilon_r=n^2$, and analytic Mie T-matrices, eliminating local extraction error. Taking the trace-sphere radius $R$ as the length unit, each scatterer has radius $a=0.97R$, with $kR=1.575$ and $ka=1.528$. Their trace-sphere centers are $\mathbf c_{1,2}=(0,0,\mp1.08R)$, giving $d_C=2.16R$ and a smallest concentric enclosing-sphere radius $R_C=R+d_C/2=2.08R$. For the source and observation radii in Fig.~\ref{fig:operator_convergence_setup}\subref{fig:operator_convergence_setup_a}, define $d_S:=k(R_S-R_C)$ and $d_O:=k(R_O-R_C)$. The fixed surfaces coincide at $R_S=R_O=2.334R$, or $d_S=d_O=0.4$.

\begin{figure}[!t]
    \centering
    \includegraphics{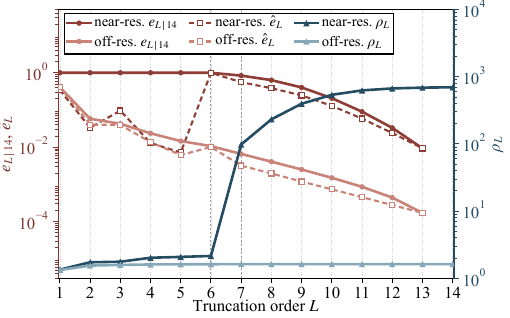}
    \caption{Fixed-surface truncation comparison for the near-resonant ($n=6.881$) and off-resonant ($n=7.000$) cases. The curves show $e_{L\mid14}$, $\widehat e_L$, and $\rho_L$; gray guides mark the $L=6\to7$ transition.}
    \label{fig:operator_convergence_comparison}
\end{figure}

With the geometry and trace metric fixed, we scan $n\in[6.878,6.884]$ at 25 equally spaced points for $L=14,16,18$. All three curves in Fig.~\ref{fig:operator_convergence_setup}\subref{fig:operator_convergence_setup_b} attain a common sampled maximum at $n=6.881$. We take this value as the near-resonant case and $n=7.000$ as the off-resonant control, with all other parameters unchanged. At the selected peak, $\rho_{14}=6.959\times10^2$, $\rho_{16}=7.000\times10^2$, and $\rho_{18}=7.006\times10^2$, a spread below $0.7\%$. With $L=18$ as a higher-order check, $e_{14\mid18}=5.262\times10^{-3}$ and $e_{16\mid18}=5.839\times10^{-4}$. We therefore use $L_\star=14$ as the finite reference.

To identify the angular content responsible for the resolvent peak, we next examine the truncation dependence of the retained collective system. Under the same fixed surface metrics, Fig.~\ref{fig:operator_convergence_comparison} reveals different behavior in the two material states. In the off-resonant control, $\rho_L$ remains of order unity, increasing from $1.33$ to $1.63$, while $e_{L\mid14}$ and $\widehat e_L$ closely track each other. In the present near-resonant case, $e_{L\mid14}$ remains close to unity through $L=6$, although several adjacent-order differences are small. Consecutive low-order systems omit essentially the same collective direction, so their mutual agreement does not imply agreement with the higher-order response. Across $L=6\to7$, $\rho_L$ jumps from $2.15$ to $97.9$, and $\widehat e_6=0.9997$. Since this step admits precisely the $l=7$ VSWF subspace, we examine the dominant input singular subspace of $\widetilde{\mathbf R}_L$. This subspace is two-dimensional for both $L=7$ and $14$, and the corresponding basis-invariant trace-normalized squared-coefficient fractions at $l=7$ are $99.4\%$ and $99.3\%$, respectively. The singular-subspace construction and separation assessment are detailed in Supplementary Note S8. The newly admitted $l=7$ subspace therefore carries the dominant near-singular input subspace. Once this subspace is retained, $\rho_L$ approaches its reference amplification and both error measures decrease together. Thus an adjacent-order indicator can suggest premature convergence until the dominant resolvent subspace is resolved.

\begin{figure}[!t]
    \centering
    \subfloat[]{\includegraphics[width=0.515\columnwidth]{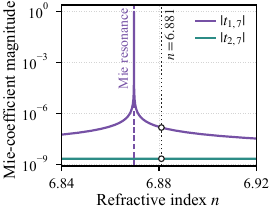}\label{fig:collective_resonance_diagnostic_mie}}\hfill
    \subfloat[]{\includegraphics[width=0.475\columnwidth]{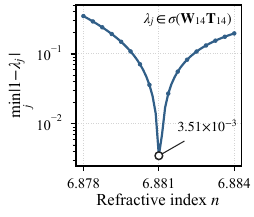}\label{fig:collective_resonance_diagnostic_wt}}
    \caption{Single-body and collective diagnostics for the $l=7$ sensitivity identified by the truncation scan. (a) Isolated-sphere $l=7$ Mie coefficients. (b) Minimum distance of the eigenvalues of $\mathbf W_{14}\mathbf T_{14}$ to unity.}
    \label{fig:collective_resonance_diagnostic}
\end{figure}

The localization at $l=7$ motivates a direct test of whether the observed sensitivity originates from an isolated-sphere Mie pole of the same angular order. Since the isolated-sphere T-matrix is diagonal in angular order and polarization, we examine its two $l=7$ coefficients in Fig.~\ref{fig:collective_resonance_diagnostic}\subref{fig:collective_resonance_diagnostic_mie}. A bracketed search for a zero of the imaginary part of the lossless Mie denominator locates the unity-magnitude resonance at $n=6.870$. At $n=6.881$, however, the corresponding coefficients are only $|t_{1,7}|=1.49\times10^{-7}$ and $|t_{2,7}|=2.15\times10^{-9}$, so the sampled peak is not a direct continuation of the isolated-sphere $l=7$ pole. The collective diagnostic in Fig.~\ref{fig:collective_resonance_diagnostic}\subref{fig:collective_resonance_diagnostic_wt} instead yields $\min_j|1-\lambda_j|=3.51\times10^{-3}$, indicating that an eigenvalue of $\mathbf W_{14}\mathbf T_{14}$ lies close to unity. Together, these results rule out direct inheritance from the same-order isolated-sphere pole and locate the observed near singularity in the coupled interaction resolvent.

To isolate external filtering, Fig.~\ref{fig:operator_clearance_kappa} varies the relevant clearance $d_S$ and $d_O$ while holding the two-sphere trace geometry and surface-metric fixed. All six finite-reference tails decrease with $L$ and their corresponding clearance. Under the present metrics, $\kappa_{S,L\mid14}<\kappa_{O,L\mid14}$ at every matched $d$. At $L=1$, increasing $d$ from $0.4$ to $2.0$ reduces the source and observation tails by about $7\%$ and $11\%$, respectively. At $L=13$, the same increase reduces both tails by a factor of about $5.6\times10^3$.

\begin{figure}[!t]
    \centering
    \includegraphics{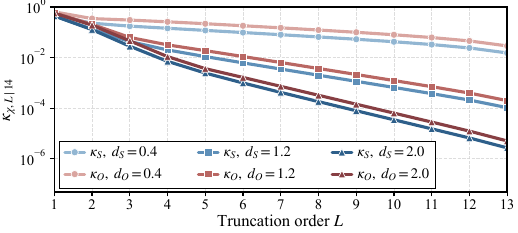}
    \caption{Source and observation finite-reference tail ratios at $L_\star=14$ for $d_S,d_O\in\{0.4,1.2,2.0\}$.}
    \label{fig:operator_clearance_kappa}
\end{figure}

\begin{figure}[!t]
    \centering
    \includegraphics{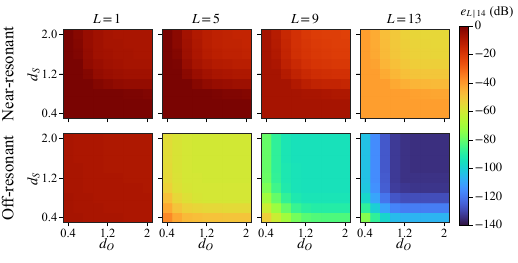}
    \caption{Joint maps of $e_{L\mid14}^{\mathrm{dB}}:=20\log_{10}e_{L\mid14}$ over $d_S,d_O\in[0.4,2.0]$ at $L=1,5,9,13$. Values are computed and displayed on a $9\times9$ clearance grid. The near- and off-resonant states occupy the upper and lower rows, respectively.}
    \label{fig:operator_clearance_heatmaps}
\end{figure}

Across the $9\times9$ grid in Fig.~\ref{fig:operator_clearance_heatmaps}, larger clearances generally lower $e_{L\mid14}$ in both material states, but with different sensitivities. In the present near-resonant case, moving either surface from $d=0.4$ to $2.0$ while holding the other at $d=0.4$ attenuates the numerator $\|\mathcal G_L^{\rm sca}-\mathcal G_{14}^{\rm sca}\|$ and reference norm $\|\mathcal G_{14}^{\rm sca}\|$ almost proportionally. At $L=13$, either one-sided displacement therefore reduces $e_{L\mid14}$ by at most $0.29$~dB, whereas moving both surfaces to $d_S=d_O=2.0$ reduces it by $14.2$~dB. The weak one-sided sensitivity is thus conditional on the opposite surface remaining at $d=0.4$. In the off-resonant control, the reference norm varies by less than $2\%$, so external filtering is expressed more directly in the relative error. Increasing only $d_S$ or $d_O$ to $2.0$, while holding the other at $0.4$, reduces $e_{L\mid14}$ by $27.7$ and $31.0$~dB at $L=13$, respectively. Increasing both reduces it by $58.7$~dB. This contrast identifies a coupled two-sided bottleneck in the near-resonant case: appreciable improvement occurs only after both external tails are suppressed.

\subsection{Accessible Scattering Channels}
\label{subsec:accessible_scattering_channels}

Having assessed finite-order convergence, Fig.~\ref{fig:accessible_scattering_channels} illustrates the 10 largest singular values of $\mathcal G_{14}^{\mathrm{sca}}$ for the same $d_S=d_O=0.4$ geometry. Each singular value is the scattered-field gain attained by a unit-norm source in the corresponding accessible channel. Consistent with the axial symmetry, each spectrum begins with a $\pm m$ pair: $m=\pm2$ near resonance and $m=\pm1$ in the off-resonant control. The corresponding leading singular values are $(\sigma_1,\sigma_2)=(29.068,29.068)$ and $(1.406,1.406)$, respectively. Near resonance, $\sigma_3=1.455$ and $\sigma_{10}=0.696$, giving $\sigma_2/\sigma_3=19.976$, while the leading pair accounts for $99.170\%$ of the total squared singular-value sum over all retained channels. By contrast, the off-resonant case has $\sigma_3=1.376$, $\sigma_{10}=0.586$, and $\sigma_2/\sigma_3=1.022$, with no comparable separation. The cross-state ratio $\sigma_1^{\mathrm{near}}/\sigma_1^{\mathrm{off}}=20.672$ is far below $\rho_{14}^{\mathrm{near}}/\rho_{14}^{\mathrm{off}}=427.987$. This disparity shows that internal amplification is only partly transferred to $\mathcal G_{14}^{\mathrm{sca}}$, whose gain also depends on $\mathbf T_{14}$, source accessibility, and observation visibility.

\begin{figure}[!t]
    \centering
    \subfloat[]{\includegraphics{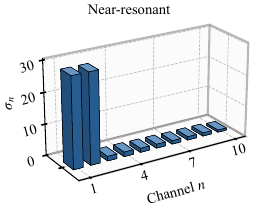}\label{fig:accessible_scattering_channels_a}}\hfill
    \subfloat[]{\includegraphics{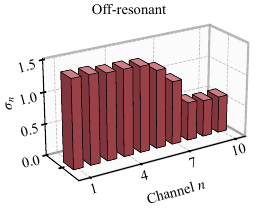}\label{fig:accessible_scattering_channels_b}}
    \caption{Leading 10 accessible-channel singular values at $L=14$ and $d_S=d_O=0.4$. (a) Near-resonant state with an $m=\pm2$ leading pair. (b) Off-resonant state with an $m=\pm1$ leading pair.}
    \label{fig:accessible_scattering_channels}
\end{figure}

Corollary~\ref{cor:accessible_scattering_channels}, together with the explicit surface formulas \eqref{eq:accessible_channel_electric_current}--\eqref{eq:accessible_channel_observation_mode}, recovers the source and observation members of each scattering channel from the finite metric core. The observation member is unit normalized, so the actual scattered field of the $n$th singular source is $\sigma_n\widehat{\mathbf E}_n^{\mathrm{sca}}$. Because each leading $\pm m$ pair spans a degenerate two-dimensional subspace, its orthonormal basis is nonunique. We phase-fix the two $m$-resolved singular sources $\mathbf q_1,\mathbf q_2$ and visualize the unit-norm standing-wave representative $\mathbf q_{\mathrm{dom}}=(\mathbf q_1+\mathbf q_2)/\sqrt{2}$. Fig.~\ref{fig:accessible_optimal_source_field_vectors} shows its electric and magnetic surface currents in the first row and its actual scattered field in the second. Although both displayed sources have unit $\mathcal X$-norm, the near- and off-resonant output norms are $29.068$ and $1.406$, respectively. The common output scale makes this enhancement directly visible. An independent forward solution of the retained Foldy--Lax system using the reconstructed source reproduces the two outputs with relative $\mathcal Y$-norm errors of $9.370\times10^{-10}$ and $9.150\times10^{-12}$, respectively. The stable finite-core realization, degeneracy treatment, and forward check are detailed in Supplementary Note S9.

\begin{figure}[!t]
    \centering
    \includegraphics[width=\columnwidth]{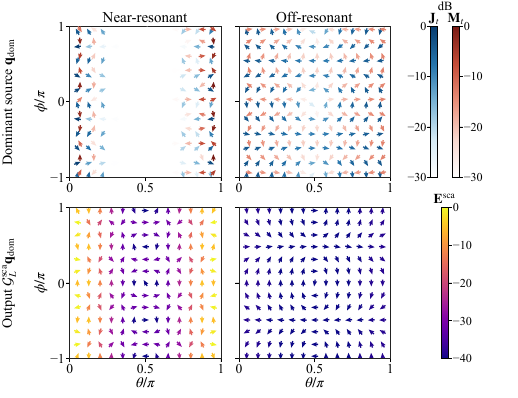}
    \caption{One nonunique, phase-fixed representative of each leading $\pm m$ pair at $L=14$ and $d_S=d_O=0.4$; the rows show equivalent source currents and singular-value-weighted scattered fields.}
    \label{fig:accessible_optimal_source_field_vectors}
\end{figure}

Fig.~\ref{fig:accessible_optimal_source_field_vectors} shows the external source and field patterns but not the internal angular orders carrying the enhancement. For each member of the leading pair, let $\mathbf f_n^{\mathrm{exc}}:=P_{14}^{\mathrm{reg}}\mathcal S_{14}\mathbf q_n$ be its directly excited regular trace and $\mathbf f_n^{\mathrm{sca}}:=P_{14}^{\mathrm{out}}\mathbf K_{14}\mathcal S_{14}\mathbf q_n$ the resulting scattered outgoing trace. Set $\Delta\Pi_l^{s}:=\Pi_l^{s}-\Pi_{l-1}^{s}$ for $s\in\{\mathrm{reg},\mathrm{out}\}$, with $\Pi_0^{s}=0$. Using $\|\cdot\|_{\mathrm{reg}}$ and $\|\cdot\|_{\mathrm{out}}$ for the fixed norms of the two trace spaces, we define the pair-averaged angular fractions by
\begin{align}
    w_l^{\mathrm{exc}}
     & :=
    \frac{\sum_{n=1}^{2}\|\Delta\Pi_l^{\mathrm{reg}}\mathbf f_n^{\mathrm{exc}}\|_{\mathrm{reg}}^2}
    {\sum_{n=1}^{2}\|\mathbf f_n^{\mathrm{exc}}\|_{\mathrm{reg}}^2},
    \label{eq:accessible_channel_trace_weight_exc} \\
    w_l^{\mathrm{sca}}
     & :=
    \frac{\sum_{n=1}^{2}\|\Delta\Pi_l^{\mathrm{out}}\mathbf f_n^{\mathrm{sca}}\|_{\mathrm{out}}^2}
    {\sum_{n=1}^{2}\|\mathbf f_n^{\mathrm{sca}}\|_{\mathrm{out}}^2}.
    \label{eq:accessible_channel_trace_weight_sca}
\end{align}
These fractions sum to one and are invariant under a unitary change of basis within the leading two-dimensional channel subspace. They compare the marginal trace content before and after the coherent multiple-scattering map; they are not an incoherent $l$-to-$l$ power balance. Their coordinate realization and invariance are derived in Supplementary Note S9.

\begin{figure}[!t]
    \centering
    \subfloat[]{\includegraphics{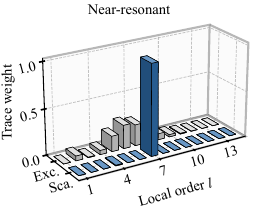}\label{fig:accessible_channel_trace_weights_a}}\hfill
    \subfloat[]{\includegraphics{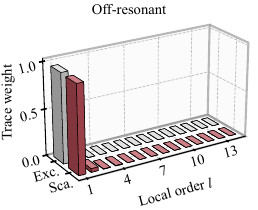}\label{fig:accessible_channel_trace_weights_b}}
    \caption{Angular trace-metric fractions of the leading two-dimensional channel subspace at $L=14$ and $d_S=d_O=0.4$. (a) Near-resonant state. (b) Off-resonant state. \emph{Exc.} denotes the directly excited regular trace, and \emph{Sca.} denotes the corresponding scattered outgoing trace.}
    \label{fig:accessible_channel_trace_weights}
\end{figure}

Fig.~\ref{fig:accessible_channel_trace_weights} reveals the angular reorganization behind the spectral enhancement. In the near-resonant pair, $79.6\%$ of the excitation regular-trace weight lies at $l=5$--$8$, while $99.25\%$ of the scattered outgoing-trace weight is concentrated at $l=7$. In the off-resonant pair, the corresponding excitation and scattered weights are instead $96.35\%$ and $95.31\%$ concentrated at $l=1$. In the $\mathcal X$ and $\mathcal Y$ metrics, the near-resonant $n=3,4$ source and observation subspaces align closely with their off-resonant $n=1,2$ counterparts, with minimum principal cosines of $0.9941$ and $0.9967$, respectively. The collective peak therefore promotes a distinct high-order pair above an otherwise preserved low-order channel family rather than merely amplifying the off-resonant optimum.

Together, these numerical studies separate implementation fidelity, truncation convergence, and channel strength. Full-wave field comparisons validate representative excitations, operator diagnostics expose worst-case truncation error and distinguish external tail suppression from internal resolvent amplification, and the channel analysis identifies the corresponding unit-source responses. In the two-sphere example, near resonance creates a coupled two-sided convergence bottleneck and promotes a distinct high-order accessible channel.

\section{Conclusion}
\label{sec:conclusion}

This work formulates a source-to-observation scattering Green operator on fixed continuous Hilbert spaces and establishes its Maxwell interpretation through a trace-space factorization under isolated-body well-posedness and collective uniqueness. For fixed, pairwise-disjoint auxiliary trace spheres enclosing the scatterers and bounded local transition operators, the finite VSWF realizations of the scattering Green operator converge in operator norm. The resulting structural bound separates external modal tails from collective resolvent sensitivity. A finite metric core preserves the retained operator's nonzero singular values without carrying external-support quadrature DoF into the singular-value problem and reconstructs the corresponding accessible source--field pairs. Mixed-cluster full-wave tests quantify implementation fidelity, while finite-reference near-resonant tests expose the distinct roles of internal amplification, source accessibility, observation visibility, and channel reordering.

\appendices
\setcounter{equation}{0}
\renewcommand{\theequation}{A\arabic{equation}}
\renewcommand{\theHequation}{A\arabic{equation}}

\section{Proof of the Source-to-Modal Excitation Map}
\label{app:continuous_source_projection}

Set $\mathbf r_i=\mathbf r-\mathbf c_i$ and $\mathbf E_i^0(\mathbf r_i)=\mathbf E^0(\mathbf c_i+\mathbf r_i)$. Then \eqref{eq:background_radiation_integral} becomes
\begin{equation}
    \begin{aligned}
        \mathbf E_i^0(\mathbf r_i)
         & =
        \mathrm{j}\omega\mu
        \int_{\Xi_{\mathrm{src}}}
        \overline{\overline{\mathbf G}}^0(\mathbf c_i+\mathbf r_i,\mathbf r')
        \cdot\mathbf J(\mathbf r')\,\mathrm d\mu_{\mathrm{src}}(\mathbf r') \\
         & \quad
        -
        \nabla_i\times
        \int_{\Xi_{\mathrm{src}}}
        \overline{\overline{\mathbf G}}^0(\mathbf c_i+\mathbf r_i,\mathbf r')
        \cdot\mathbf M(\mathbf r')\,\mathrm d\mu_{\mathrm{src}}(\mathbf r').
    \end{aligned}
    \label{eq:appendix_source_field_local}
\end{equation}
Here $\nabla_i$ acts on $\mathbf r_i$. The separation hypothesis gives some $\rho_i<1$ such that
\begin{equation}
    \frac{|\mathbf r_i|}{|\mathbf r'-\mathbf c_i|}
    \leq
    \rho_i
    <
    1
    \label{eq:appendix_separated_ratio}
\end{equation}
for $|\mathbf r_i|\leq R_i$ and $\mathbf r'$ in the source support. On this separated product set the free-space dyadic is smooth and admits the VSWF expansion in the convention of Section~\ref{subsec:vswf} \cite{chew1995,doicu2006}
\begin{equation}
    \overline{\overline{\mathbf G}}^0(\mathbf c_i+\mathbf r_i,\mathbf r')
    =
    \frac{\mathrm{j}k}{\pi}
    \lim_{L\to\infty}
    \sum_{\alpha\in\Lambda_L}
    \bm{\Phi}_{\alpha}^{\mathrm{reg}}(\mathbf r_i)
    \otimes
    \bm{\Phi}_{\overline\alpha}^{\mathrm{out}}(\mathbf r'-\mathbf c_i).
    \label{eq:appendix_dyadic_green_vswf_expansion}
\end{equation}
The partial sums and one field-variable curl converge absolutely and uniformly on such separated compact sets. For the magnetic-current contribution, the second Cauchy component follows from $\nabla\times\nabla\times\overline{\overline{\mathbf G}}^0=k^2\overline{\overline{\mathbf G}}^0$ on the separated source-free product set. Thus the termwise integration and curl operations for every $\mathbf q\in\mathcal X$ are valid. Point sources, or other distributions not represented by the chosen Hilbert support, use the same formulas as distributional pairings outside the theorem.

Substitution into the electric-current term yields
\begin{equation}
    \begin{aligned}
        \mathbf E_i^{J}(\mathbf r_i)
         & =
        -\frac{\omega\mu k}{\pi}
        \lim_{L\to\infty}
        \sum_{\alpha\in\Lambda_L}
        \bm{\Phi}_{\alpha}^{\mathrm{reg}}(\mathbf r_i) \\
         & \quad{}\times
        \int_{\Xi_{\mathrm{src}}}
        \bm{\Phi}_{\overline\alpha}^{\mathrm{out}}(\mathbf r'-\mathbf c_i)
        \cdot\mathbf J(\mathbf r')\,\mathrm d\mu_{\mathrm{src}}(\mathbf r').
    \end{aligned}
    \label{eq:appendix_electric_source_expansion}
\end{equation}
For the magnetic term, $\nabla_i\times\bm{\Phi}_{\alpha}^{\mathrm{reg}}=k\bm{\Phi}_{\alpha^\star}^{\mathrm{reg}}$ and relabelling give
\begin{equation}
    \begin{aligned}
        \mathbf E_i^{M}(\mathbf r_i)
         & =
        -\frac{\mathrm{j}k^2}{\pi}
        \lim_{L\to\infty}
        \sum_{\alpha\in\Lambda_L}
        \bm{\Phi}_{\alpha}^{\mathrm{reg}}(\mathbf r_i) \\
         & \quad{}\times
        \int_{\Xi_{\mathrm{src}}}
        \bm{\Phi}_{\overline{\alpha^\star}}^{\mathrm{out}}(\mathbf r'-\mathbf c_i)
        \cdot\mathbf M(\mathbf r')\,\mathrm d\mu_{\mathrm{src}}(\mathbf r').
    \end{aligned}
    \label{eq:appendix_magnetic_source_expansion}
\end{equation}
Adding \eqref{eq:appendix_electric_source_expansion} and \eqref{eq:appendix_magnetic_source_expansion}, then comparing with the unique regular expansion \eqref{eq:local_external_excitation_expansion}, gives the coefficients $a_{i,\alpha}^{0}$ in \eqref{eq:continuous_source_modal_projection}. Taking Cauchy traces termwise on $\Gamma_i$ identifies their retained entries with $Q_{i,L}^{\mathrm{reg}}\mathcal S_i\mathbf q$. Stacking the bodies gives $\mathcal S_L\mathbf q=\mathbf a^{\mathrm{exc}}$, completing the proof of Lemma~\ref{lem:continuous_source_modal_excitation}.

\setcounter{equation}{0}
\renewcommand{\theequation}{B\arabic{equation}}
\renewcommand{\theHequation}{B\arabic{equation}}

\section{Proof of the Modal Green-Operator Convergence}
\label{app:modal_green_convergence_proof}

For readability, write $r$ and $o$ for $\mathrm{reg}$ and $\mathrm{out}$. We work in the trace setting of Section~\ref{subsec:fixed_sphere_trace_setting}. By the positive-separation compactness in Section~\ref{subsec:abstract_green_factorization}, $\mathcal S$, $\mathcal O$, and $\mathcal W$ are compact. Since $\mathcal T$ is bounded, $\mathcal A:=\mathcal W\mathcal T$ is compact. For compact $\mathcal C$ and strongly convergent orthogonal projections $\Pi_L$, both $\|(\mathbb I-\Pi_L)\mathcal C\|$ and $\|\mathcal C(\mathbb I-\Pi_L)\|$ vanish \cite{kato1995,simon2005}. Using
\begin{equation}
    \begin{aligned}
        \mathcal A-\mathcal A_L
        ={} &
        (\mathbb I-\Pi_L^r)\mathcal A
        +\Pi_L^r\mathcal W(\mathbb I-\Pi_L^o)\mathcal T \\
            & +\Pi_L^r\mathcal W\Pi_L^o\mathcal T
        (\mathbb I-\Pi_L^r),
    \end{aligned}
    \label{eq:appendix_b_coupling_error_split}
\end{equation}
the first two terms vanish. For the last, write
$\mathcal W\Pi_L^o\mathcal T(\mathbb I-\Pi_L^r)
    =\mathcal A(\mathbb I-\Pi_L^r)
    -\mathcal W(\mathbb I-\Pi_L^o)\mathcal T(\mathbb I-\Pi_L^r)$.
Both terms vanish in norm. Hence $\delta_L\to0$.

If $\|\mathcal R\|\delta_L<1$, then
\begin{equation}
    \mathbb I-\mathcal A_L
    =
    (\mathbb I-\mathcal A)
    \left[\mathbb I-\mathcal R(\mathcal A_L-\mathcal A)\right]
    \label{eq:appendix_b_resolvent_factorization}
\end{equation}
is invertible by the Neumann lemma \cite{kato1995}, giving the first bound in \eqref{eq:collective_resolvent_error_bound}. The resolvent identity
\begin{equation}
    \mathcal R_L^\sharp-\mathcal R
    =
    \mathcal R_L^\sharp(\mathcal A_L-\mathcal A)\mathcal R
    \label{eq:appendix_b_resolvent_identity}
\end{equation}
gives the second. The range of $\mathcal A_L$ lies in $\operatorname{ran}\Pi_L^r$, and $\mathcal A_L$ vanishes on its orthogonal complement. The analysis--synthesis identities $Q_L^sP_L^s=\mathbf I$ and $P_L^sQ_L^s=\Pi_L^s$ give
\begin{equation}
    Q_L^r\mathcal A_LP_L^r
    =
    Q_L^r\mathcal W P_L^oQ_L^o\mathcal T P_L^r
    =
    \mathbf W_L\mathbf T_L.
    \label{eq:appendix_b_retained_interaction_coordinates}
\end{equation}
Thus the retained coordinate matrix is exactly $\mathbf W_L\mathbf T_L$, and the full-space inverse is
\begin{equation}
    \mathcal R_L^\sharp
    =
    \mathbb I-\Pi_L^r
    +P_L^{\mathrm{reg}}\mathbf R_LQ_L^{\mathrm{reg}},
    \quad
    P_L^{\mathrm{reg}}\mathbf R_LQ_L^{\mathrm{reg}}
    =\mathcal R_L^\sharp\Pi_L^r.
    \label{eq:appendix_b_finite_full_resolvent_relation}
\end{equation}
Equation~\eqref{eq:appendix_b_finite_full_resolvent_relation} gives $\widehat{\mathcal K}_L=\Pi_L^o\mathcal T\mathcal R_L^\sharp\Pi_L^r$. Norm convergence of $\mathcal R_L^\sharp$, strong convergence of the projections, and boundedness of $\mathcal T$ imply $\widehat{\mathcal K}_L\to\mathcal K$ strongly and $\sup_L\|\widehat{\mathcal K}_L\|<\infty$ for sufficiently large $L$.

The same compact-projection argument gives $\epsilon_{S,L},\epsilon_{O,L}\to0$. Subtracting $\mathcal O\mathcal T\mathcal R\mathcal S$ from $\mathcal O\Pi_L^o\mathcal T\mathcal R_L^\sharp\Pi_L^r\mathcal S$ yields the bound $E_L$ in \eqref{eq:green_operator_convergence_bound}, proving Theorem~\ref{thm:modal_green_operator_convergence}.

\bibliographystyle{IEEEtran}
\bibliography{ref}

\end{document}

\bstctlcite{IEEEtranBSTCTL:NoDash}

\title{Supplementary Material:\\ Finite-Modal Realization and Operator-Norm Convergence of a Source-to-Observation Electromagnetic Scattering Green Operator}

\author{
    Zhukang Wang,~\IEEEmembership{Graduate Student Member,~IEEE,}
    Da Li,~\IEEEmembership{Member,~IEEE,}
    Ruifeng Li,
    Jinyan Ma,
    Jiarun Hu,
    Jiahui Wang,
    Anqi Xia,
    Tengjiao Wang,~\IEEEmembership{Member,~IEEE,}
    Said Mikki,
    and~Er-Ping Li,~\IEEEmembership{Fellow,~IEEE}
}

\maketitle

\begin{table*}[!b]
    \caption{Guide to the supplementary notes and their roles in the main paper}
    \label{tab:supplementary_note_guide}
    \centering
    \footnotesize
    \setlength{\tabcolsep}{3.5pt}
    \renewcommand{\arraystretch}{1.04}
    \begin{tabularx}{\textwidth}{C{0.045\textwidth}L{0.205\textwidth}L{0.305\textwidth}Y}
        \toprule
        \textbf{No.}
         & \textbf{Title}
         & \textbf{Main content}
         & \textbf{Location(s) served in the main paper}                                                                                                                                         \\
        \midrule
        \hyperref[supp:trace_spaces]{\ref*{supp:trace_spaces}}
         & Maxwell Cauchy-Trace Spaces on an Enclosing Sphere
         & Ambient Maxwell trace topology, continuity of the Cauchy trace, closed Calder\'on ranges, and uniqueness of regular and outgoing representatives.
         & Section II-A, the two paragraphs defining and characterizing the regular and outgoing trace ranges after the Cauchy-trace definition.                                                 \\
        \specialrule{0.3pt}{0.1pt}{0.1pt}
        \hyperref[supp:separated_compact_maps]{\ref*{supp:separated_compact_maps}}
         & Compact Maxwell Maps Across Separated Supports
         & Separated-support smoothing and compact embeddings for the source, observation, and interbody continuation maps.
         & Section II-A, the paragraph following the source and observation maps; Section II-B, the paragraph following the interbody continuation operator.                                     \\
        \specialrule{0.3pt}{0.1pt}{0.1pt}
        \hyperref[supp:fredholm_exactness]{\ref*{supp:fredholm_exactness}}
         & Fredholm Solvability and Exactness of the Trace Assembly
         & Bodywise outgoing uniqueness, Fredholm invertibility of the collective trace system, and exact reconstruction of the physical multiple-scattering field.
         & Section II-B, the paragraph following the collective incident-trace equation, and the proof of Proposition 1.                                                                         \\
        \specialrule{0.3pt}{0.1pt}{0.1pt}
        \hyperref[supp:vswf_trace_metric]{\ref*{supp:vswf_trace_metric}}
         & Fixed VSH Trace Metric and VSWF Trace Projections
         & Fixed VSH Hilbert metric, VSWF trace Gram entries, modal projections, contractivity, and strong convergence.
         & Section III-B, from the fixed trace metric through the modal-projection stability statement.                                                                                          \\
        \specialrule{0.3pt}{0.1pt}{0.1pt}
        \hyperref[supp:surface_interface_grams]{\ref*{supp:surface_interface_grams}}
         & Surface Interface Gramians and Channel Recovery
         & Hilbert adjoints of the surface interfaces, source and observation Gramians, and explicit recovery of source and field channels.
         & Section III-E, the surface-interface Gram and channel-recovery formulas; used again for the reconstructions in Section IV-C.                                                          \\
        \specialrule{0.3pt}{0.1pt}{0.1pt}
        \hyperref[supp:nonspherical_tmatrix_extraction]{\ref*{supp:nonspherical_tmatrix_extraction}}
         & External Extraction of Nonspherical Local T-Matrices
         & Discrete geometries and materials, batched SCUFF-EM solves, outgoing-wave fitting, convention conversion, and replay diagnostics.
         & Section IV-A, the paragraph introducing the three nonspherical T-matrices, together with Table I.                                                                                     \\
        \specialrule{0.3pt}{0.1pt}{0.1pt}
        \hyperref[supp:comsol_fullwave_protocol]{\ref*{supp:comsol_fullwave_protocol}}
         & COMSOL Full-Wave Reference and Benchmark Protocol
         & Computational domain, PML, mesh, solver, phasor mapping, comparison norms, prescribed sources, nested scenes, and timing boundaries.
         & Section IV-A, the paragraphs specifying the COMSOL setup, the three sources, field-error interpretation, and runtime measurements, together with Table II.                            \\
        \specialrule{0.3pt}{0.1pt}{0.1pt}
        \hyperref[supp:finite_reference_diagnostics]{\ref*{supp:finite_reference_diagnostics}}
         & Finite-Reference Operator Diagnostics
         & Semianalytic external-sphere Gramians; nested-reference errors; source and observation tails; trace-normalized resolvent and angular-subspace diagnostics.
         & Section IV-B, the definitions and interpretation of the finite-reference errors, tail ratios, and resolvent amplification, plus the paragraph diagnosing the dominant $l=7$ subspace. \\
        \specialrule{0.3pt}{0.1pt}{0.1pt}
        \hyperref[supp:accessible_channel_diagnostics]{\ref*{supp:accessible_channel_diagnostics}}
         & Finite-Modal Realization and Diagnostics of Accessible Channels
         & QR realization of the metric core, doublet and azimuthal-sector diagnostics, basis-invariant angular fractions, and the independent forward check.
         & Section IV-C, the paragraphs reporting the channel spectra, reconstructed representative and forward check, and angular trace fractions.                                              \\
        \bottomrule
    \end{tabularx}

    \vspace{3em}
    \caption{Core notation used throughout the supplementary material}
    \label{tab:supplementary_core_notation}
    \centering
    \footnotesize
    \setlength{\tabcolsep}{3.5pt}
    \renewcommand{\arraystretch}{1.08}
    \begin{tabularx}{\textwidth}{C{0.145\textwidth}L{0.315\textwidth}C{0.145\textwidth}Y}
        \toprule
        \textbf{Symbol} & \textbf{Meaning}                                                                                  & \textbf{Symbol} & \textbf{Meaning} \\
        \midrule
        $B_i,\ \Gamma_i$
                        & Enclosing ball and trace sphere associated with body $i$.
                        & $\mathcal X,\ \mathcal Y$
                        & Prescribed source and observation Hilbert spaces.                                                                                      \\
        $\mathcal H^{\mathrm{reg}},\ \mathcal H^{\mathrm{out}}$
                        & Product spaces of regular and outgoing Maxwell Cauchy traces.
                        & $\mathcal S,\ \mathcal T,\ \mathcal W,\ \mathcal O$
                        & Source, local-transition, interbody-continuation, and observation maps.                                                                \\
        $\mathcal R,\ \mathcal K$
                        & Collective resolvent and cluster transition operator.
                        & $P_L^s,\ Q_L^s,\ \Pi_L^s$
                        & Modal synthesis, analysis, and orthogonal projection for $s\in\{\mathrm{reg},\mathrm{out}\}$.                                          \\
        $\mathbf T_L,\ \mathbf W_L,\ \mathbf R_L,\ \mathbf K_L$
                        & Retained coordinate matrices for local response, continuation, resolvent, and cluster transition.
                        & $\mathcal G^{\mathrm{sca}},\ \mathcal G_L^{\mathrm{sca}}$
                        & Exact and finite-modal scattering Green operators.                                                                                     \\
        \bottomrule
    \end{tabularx}
\end{table*}

\clearpage

\section{Maxwell Cauchy-Trace Spaces on an Enclosing Sphere}
\label{supp:trace_spaces}

This note supplies the trace-space details used in Section~II of the main paper. It concerns only source-free solutions of the homogeneous background Maxwell equation on either side of an artificial enclosing sphere. In particular, a field regular in the enclosing ball is a background incident-field representative; it is not the physical field inside the material support.

For the fixed background wavenumber $k>0$, let $B_D=B_R(\mathbf c)$ be an open ball with boundary $\Gamma_D=\partial B_D$ and outward unit normal $\mathbf n_D$. Set
\begin{equation}
    \mathbf X_D
    :=
    H^{-1/2}(\operatorname{div}_{\Gamma_D},\Gamma_D)
    \times
    H^{-1/2}(\operatorname{div}_{\Gamma_D},\Gamma_D).
    \label{eq:supp_ambient_trace_space}
\end{equation}
The two factors in \eqref{eq:supp_ambient_trace_space} carry their standard Maxwell graph norms, and $\mathbf X_D$ carries the associated product Hilbert structure. Any other fixed Hilbert inner product inducing the same product topology gives the same notions of closedness, boundedness, and compactness used in the main paper.

For a Lipschitz domain $\Omega$, the tangential trace $\mathbf u\mapsto\mathbf n\times\mathbf u$ extends continuously from $H(\operatorname{curl},\Omega)$ to $H^{-1/2}(\operatorname{div}_{\partial\Omega},\partial\Omega)$ \cite{buffa2002traces,nedelec2001}. Since $\mathbf E\in H(\operatorname{curl}^2,\Omega)$ means that both $\mathbf E$ and $\nabla\times\mathbf E$ belong to $H(\operatorname{curl},\Omega)$, both possess continuous tangential traces. On $\Gamma_D$, define the Cauchy trace
\begin{equation}
    \gamma_D^{\mathrm C}\mathbf E
    :=
    \left(
    \mathbf n_D\times\mathbf E,
    \mathbf n_D\times k^{-1}\nabla\times\mathbf E
    \right)\big|_{\Gamma_D}
    \label{eq:supp_cauchy_trace}
\end{equation}
The trace pair in \eqref{eq:supp_cauchy_trace} therefore belongs to $\mathbf X_D$, and $\gamma_D^{\mathrm C}$ is continuous on the corresponding $H(\operatorname{curl}^2)$ field class. The factor $k^{-1}$ changes only the physical normalization of the second component and does not change its trace regularity.

With $\mathcal M_k$ defined in the main paper, the trace ranges are
\begin{align}
    \mathcal H^{\mathrm{reg}}(D)
     & :=
    \gamma_D^{\mathrm C}\mathcal M_k(B_D),
    \label{eq:supp_regular_trace_range} \\
    \mathcal H^{\mathrm{out}}(D)
     & :=
    \gamma_D^{\mathrm C}
    \left\{
    \mathbf E\in\mathcal M_k(\mathbb R^3\setminus\overline B_D)
    :
    \mathbf E \,\,\text{outgoing}
    \right\}.
    \label{eq:supp_outgoing_trace_range}
\end{align}
The label $D$ is retained only to associate the auxiliary sphere with its enclosed body. The spaces in \eqref{eq:supp_regular_trace_range}--\eqref{eq:supp_outgoing_trace_range} themselves depend on the fixed background Maxwell operator, the wavenumber $k$, the sphere $\Gamma_D$, and the regular or outgoing branch, but not on the material support $\Omega_D$ or the parameters $\Theta_D$ inside the sphere.
Here the exterior instance of $\mathcal M_k$ has $H(\operatorname{curl}^2)$ regularity up to $\Gamma_D$ and the corresponding local regularity at infinity, as stipulated in the main paper. Outgoing means the Silver--M\"uller radiation condition \cite{chew1995,colton2013}. Closedness follows from the Calder\'on characterization below.

\begin{proposition}[Closed trace ranges and unique field representatives]
    \label{prop:supp_trace_range_characterization}
    The spaces $\mathcal H^{\mathrm{reg}}(D)$ and $\mathcal H^{\mathrm{out}}(D)$ in \eqref{eq:supp_regular_trace_range}--\eqref{eq:supp_outgoing_trace_range} are closed subspaces of $\mathbf X_D$ and hence are Hilbert spaces with the inherited inner products. Moreover, the Cauchy trace is injective on $\mathcal M_k(B_D)$ and on the outgoing subclass of $\mathcal M_k(\mathbb R^3\setminus\overline B_D)$. Consequently, every element of $\mathcal H^{\mathrm{reg}}(D)$ or $\mathcal H^{\mathrm{out}}(D)$ has a unique regular or outgoing field representative, respectively.
\end{proposition}

\begin{proof}
    Let $\gamma_{D,-}^{\mathrm C}$ and $\gamma_{D,+}^{\mathrm C}$ denote the one-sided Cauchy traces taken from $B_D$ and from $\mathbb R^3\setminus\overline B_D$, respectively, with both traces evaluated using the normal $\mathbf n_D$ fixed above. The main paper suppresses the side marker because the domain of each field makes the relevant trace unambiguous. Any difference between this convention and a standard Maxwell Calder\'on convention in the one-sided normal sign, field ordering, or the fixed factor $k^{-1}$ is represented by a fixed boundedly invertible coordinate map on $\mathbf X_D$. Conjugation by this map preserves boundedness and idempotence and maps the standard Cauchy-data ranges onto those used below.

    The Maxwell representation formula associates with a pair $\mathbf f\in\mathbf X_D$ an interior background field $\mathcal P_D^{\mathrm{reg}}\mathbf f$ and a radiating exterior background field $\mathcal P_D^{\mathrm{out}}\mathbf f$. The layer-potential mapping properties make these representation maps continuous into the corresponding local $H(\operatorname{curl}^2)$ solution spaces, and the tangential jump relations define bounded Calder\'on operators on $\mathbf X_D$ \cite{nedelec2001,buffa2003maxwellbem,claeys2012multitrace,kristensson2020}. With the signs of the representation potentials fixed by the common normal convention, set
    \begin{equation}
        \mathcal C_D^{\mathrm{reg}}
        :=
        \gamma_{D,-}^{\mathrm C}\mathcal P_D^{\mathrm{reg}},
        \quad
        \mathcal C_D^{\mathrm{out}}
        :=
        \gamma_{D,+}^{\mathrm C}\mathcal P_D^{\mathrm{out}}.
        \label{eq:supp_calderon_operators}
    \end{equation}
    The representation and jump formulas give
    \begin{align}
        (\mathcal C_D^{\mathrm{reg}})^2
         & =
        \mathcal C_D^{\mathrm{reg}},
         &
        \operatorname{ran}\mathcal C_D^{\mathrm{reg}}
         & =
        \mathcal H^{\mathrm{reg}}(D),
        \label{eq:supp_regular_calderon_range} \\
        (\mathcal C_D^{\mathrm{out}})^2
         & =
        \mathcal C_D^{\mathrm{out}},
         &
        \operatorname{ran}\mathcal C_D^{\mathrm{out}}
         & =
        \mathcal H^{\mathrm{out}}(D).
        \label{eq:supp_outgoing_calderon_range}
    \end{align}
    Indeed, applying a representation potential to arbitrary boundary data produces a field of the indicated class, so the corresponding Calder\'on range is contained in the appropriate trace range. Conversely, the representation formula reproduces every regular interior solution and every radiating exterior solution from its Cauchy data; its trace is therefore a fixed point of the corresponding Calder\'on operator. This proves the range identities in \eqref{eq:supp_regular_calderon_range} and \eqref{eq:supp_outgoing_calderon_range}.

    A bounded projector has closed range because
    \begin{equation}
        \operatorname{ran}\mathcal C_D^{s}
        =
        \ker(\mathbb I-\mathcal C_D^{s}),
        \label{eq:supp_projector_closed_range}
    \end{equation}
    and the kernel of a bounded operator is closed. Here $s\in\{\mathrm{reg},\mathrm{out}\}$ and $\mathbb I$ is the identity on $\mathbf X_D$. Equations~\eqref{eq:supp_regular_calderon_range} and \eqref{eq:supp_outgoing_calderon_range} therefore show that both spaces in \eqref{eq:supp_regular_trace_range}--\eqref{eq:supp_outgoing_trace_range} are closed in $\mathbf X_D$. They are consequently Hilbert spaces with the inherited inner products.

    It remains to prove uniqueness of the representatives. If two fields in $\mathcal M_k(B_D)$ have the same Cauchy trace, their difference $\mathbf U$ has zero Cauchy data. The interior representation formula gives
    \begin{equation}
        \mathbf U
        =
        \mathcal P_D^{\mathrm{reg}}
        \gamma_{D,-}^{\mathrm C}\mathbf U
        =
        \mathbf 0
        \quad\text{in }B_D.
        \label{eq:supp_interior_trace_injectivity}
    \end{equation}
    This is uniqueness for complete Maxwell Cauchy data and does not require excluding an interior eigenfrequency. Likewise, if two radiating exterior fields have the same Cauchy trace, their difference is radiating and has zero Cauchy data. The exterior representation formula, equivalently exterior radiating uniqueness together with the Maxwell Rellich lemma, gives
    \begin{equation}
        \mathbf U
        =
        \mathcal P_D^{\mathrm{out}}
        \gamma_{D,+}^{\mathrm C}\mathbf U
        =
        \mathbf 0
        \quad\text{in }\mathbb R^3\setminus\overline B_D
        \label{eq:supp_exterior_trace_injectivity}
    \end{equation}
    \cite{nedelec2001,colton2013}. Thus both restricted trace maps are injective. Their surjectivity onto the two spaces is part of the range definitions in \eqref{eq:supp_regular_trace_range}--\eqref{eq:supp_outgoing_trace_range}, which completes the proof.
\end{proof}

Proposition~\ref{prop:supp_trace_range_characterization} justifies treating the regular and outgoing Cauchy traces as Hilbert interface states in the main paper. It does not assert well-posedness of a material scattering problem: boundedness of the local transition operator remains the separate isolated-body assumption stated there.

\section{Compact Maxwell Maps Across Separated Supports}
\label{supp:separated_compact_maps}

This note supplies the separated-support regularity and compactness arguments used in Section~II of the main paper. It treats the continuous source, continuation, and observation maps before any modal coordinates are introduced.

We first record that the prescribed source and observation weights do not alter the relevant topology. Boundedness and uniform positive definiteness of $\mathbf W_X$ and $\mathbf W_Y$ give constants $0<c_X\leq C_X<\infty$ and $0<c_Y\leq C_Y<\infty$ such that
\begin{align}
    c_X\|\mathbf q\|_{L^2(\Xi_{\mathrm{src}})}^2
     & \leq
    \|\mathbf q\|_{\mathcal X}^2
    \leq
    C_X\|\mathbf q\|_{L^2(\Xi_{\mathrm{src}})}^2,
    \label{eq:supp_source_weight_equivalence} \\
    c_Y\|\mathbf E\|_{L^2(\Xi_{\mathrm{obs}})}^2
     & \leq
    \|\mathbf E\|_{\mathcal Y}^2
    \leq
    C_Y\|\mathbf E\|_{L^2(\Xi_{\mathrm{obs}})}^2.
    \label{eq:supp_observation_weight_equivalence}
\end{align}
Here and below, the unweighted norms use the volume or surface measure and the finite component dimension specified in the main paper. Consequently, boundedness and compactness in the weighted spaces are equivalent to those in the corresponding unweighted $L^2$ spaces.

\begin{lemma}[Smoothing across a positive separation]
    \label{lem:supp_separated_smoothing}
    Let $U\subset\mathbb R^3$ be a bounded Lipschitz domain, and let $m\geq0$ be an integer.
    \begin{enumerate}
        \item Let $\Xi$ be either a bounded finite-measure volume support or a compact Lipschitz surface support, and suppose that $\operatorname{dist}(\overline U,\Xi)>0$. The free-space radiation map formed from electric and equivalent magnetic currents as in the main paper satisfies
              \begin{equation}
                  \mathcal G^0:
                  L^2(\Xi;\mathbb C^{d_X})
                  \to
                  H^m(U;\mathbb C^3)
                  \quad\text{compactly}.
                  \label{eq:supp_source_smoothing}
              \end{equation}
        \item Let $\Sigma$ be a compact Lipschitz surface, set
              \begin{equation}
                  \mathbf X_\Sigma
                  :=
                  H^{-1/2}(\operatorname{div}_{\Sigma},\Sigma)^2,
                  \label{eq:supp_generic_trace_space}
              \end{equation}
              and suppose that $\operatorname{dist}(\overline U,\Sigma)>0$. The radiating Maxwell representation potential satisfies
              \begin{equation}
                  \mathcal P_{\Sigma}^{\mathrm{out}}:
                  \mathbf X_\Sigma
                  \to
                  H^m(U;\mathbb C^3)
                  \quad\text{compactly}.
                  \label{eq:supp_trace_potential_smoothing}
              \end{equation}
    \end{enumerate}
\end{lemma}

\begin{proof}
    On a product set separated from the diagonal, the free-space dyadic $\overline{\overline{\mathbf G}}^0$ and all of its derivatives in either variable are smooth and uniformly bounded. The same is true of the differentiated kernel $\nabla_{\mathbf r}\times\overline{\overline{\mathbf G}}^0$ multiplying the magnetic current. Since $\Xi$ has finite volume or surface measure, differentiation under the integral and the Cauchy--Schwarz inequality give
    \begin{equation}
        \|\mathcal G^0\mathbf q\|_{H^{m+1}(U)}
        \leq
        C_{m,U,\Xi}\|\mathbf q\|_{L^2(\Xi)}.
        \label{eq:supp_source_smoothing_estimate}
    \end{equation}
    For a surface source, its finite tangential-component representation is first inserted into the ambient electric and magnetic current vectors; this is a bounded pointwise operation and does not change the estimate. The compact embedding $H^{m+1}(U)\Subset H^m(U)$ proves \eqref{eq:supp_source_smoothing}.

    For the representation potential, choose a bounded Lipschitz domain $V$ such that $U\Subset V$ and $\operatorname{dist}(\overline V,\Sigma)>0$. The standard mapping properties of the Maxwell layer potentials give the continuous map
    \begin{equation}
        \mathcal P_{\Sigma}^{\mathrm{out}}:
        \mathbf X_\Sigma
        \to
        H(\operatorname{curl}^2,V)
        \label{eq:supp_trace_potential_base_mapping}
    \end{equation}
    \cite{nedelec2001,buffa2003maxwellbem,claeys2012multitrace}. Set $\mathbf u:=\mathcal P_{\Sigma}^{\mathrm{out}}\mathbf f$. Away from $\Sigma$, the potential satisfies
    \begin{equation}
        \nabla\times\nabla\times\mathbf u-k^2\mathbf u=\mathbf0
        \quad\text{in }V.
        \label{eq:supp_trace_potential_local_maxwell}
    \end{equation}
    Taking the distributional divergence of \eqref{eq:supp_trace_potential_local_maxwell} gives $-k^2\nabla\cdot\mathbf u=0$. Since $k>0$, $\nabla\cdot\mathbf u=0$ in $V$ and therefore $\mathbf u\in H(\operatorname{div}0,V)$, where $H(\operatorname{div}0,V):=\{\mathbf v\in H(\operatorname{div},V):\nabla\cdot\mathbf v=0\}$. Hence the vector identity $\nabla\times\nabla\times=\nabla\nabla\cdot-\Delta$ shows that every Cartesian component of $\mathbf u$ satisfies the scalar Helmholtz equation $(\Delta+k^2)u_a=0$ in $V$. Interior elliptic regularity on $U\Subset V$, followed by the continuity in \eqref{eq:supp_trace_potential_base_mapping}, gives
    \begin{equation}
        \|\mathcal P_{\Sigma}^{\mathrm{out}}\mathbf f\|_{H^{m+1}(U)}
        \leq
        C_{m,U,V}\|\mathbf u\|_{L^2(V)}
        \leq
        C_{m,U,V,\Sigma}\|\mathbf f\|_{\mathbf X_\Sigma}.
        \label{eq:supp_trace_potential_smoothing_estimate}
    \end{equation}
    Composing this estimate with $H^{m+1}(U)\Subset H^m(U)$ proves \eqref{eq:supp_trace_potential_smoothing}.
\end{proof}

For later use, let $\mathbf X_i$ denote the product Maxwell trace space on $\Gamma_i$. If $U_i$ is a neighborhood of $\Gamma_i$ and $m\geq3$, the ordinary Sobolev trace theorem applied to a smooth field and its curl gives
\begin{equation}
    \|\gamma_i^{\mathrm C}\mathbf E\|_{\mathbf X_i}
    \leq
    C_{i,m}\|\mathbf E\|_{H^m(U_i)}.
    \label{eq:supp_smooth_cauchy_trace_bound}
\end{equation}
Indeed, for any $s>1/2$, smooth tangential fields satisfy
\begin{equation}
    H^s_{\mathrm{tan}}(\Gamma_i)
    \hookrightarrow
    H^{-1/2}(\operatorname{div}_{\Gamma_i},\Gamma_i),
    \label{eq:supp_smooth_to_maxwell_trace_embedding}
\end{equation}
because $H^s(\Gamma_i)\hookrightarrow H^{-1/2}(\Gamma_i)$ and the surface divergence maps $H^s_{\mathrm{tan}}(\Gamma_i)$ continuously into $H^{s-1}(\Gamma_i)\hookrightarrow H^{-1/2}(\Gamma_i)$. These two embeddings control the two components of the Maxwell graph norm. Choosing an intermediate exponent $1/2<t<s$ also shows that this embedding is compact, by $H^s_{\mathrm{tan}}(\Gamma_i)\Subset H^t_{\mathrm{tan}}(\Gamma_i)$ followed by the continuous embedding in \eqref{eq:supp_smooth_to_maxwell_trace_embedding} with $t$ in place of $s$.

\begin{proposition}[Compact source, continuation, and observation maps]
    \label{prop:supp_compact_external_maps}
    Under the support-separation and pairwise-disjoint enclosing-ball assumptions of Section~II in the main paper, the maps
    \begin{align}
        \mathcal S_i
         & :\mathcal X\to\mathcal H_i^{\mathrm{reg}},                  \\
        \mathcal W_{ij}
         & :\mathcal H_j^{\mathrm{out}}\to\mathcal H_i^{\mathrm{reg}},
        \quad i\ne j,                                                  \\
        \mathcal O_i
         & :\mathcal H_i^{\mathrm{out}}\to\mathcal Y
    \end{align}
    are well defined and compact. Here $\mathcal O_i$ restricts the unique outgoing representative of its argument to the selected observation component, so that the cluster observation map is $\mathcal O(\mathbf f_i)_i=\sum_i\mathcal O_i\mathbf f_i$. Consequently, the stacked source map $\mathcal S$, the finite block continuation map $\mathcal W$, and the observation map $\mathcal O$ are compact.
\end{proposition}

\begin{proof}
    Fix $i$. Positive separation of $\Xi_{\mathrm{src}}$ from $\overline B_i$ permits a bounded Lipschitz neighborhood $U_i\supset\overline B_i$ whose closure remains separated from the source support. Lemma~\ref{lem:supp_separated_smoothing} shows that $\mathbf q\mapsto\mathbf E^0=\mathcal G^0\mathbf q$ is compact into $H^m(U_i)$ for any chosen $m\geq3$. Since the distributional sources are supported outside $U_i$, $\mathcal L_k\mathbf E^0=\mathbf0$ in $U_i$; hence its restriction to $B_i$ belongs to $\mathcal M_k(B_i)$. The continuous Cauchy trace in \eqref{eq:supp_smooth_cauchy_trace_bound} therefore gives a compact map
    \begin{equation}
        \mathcal S_i\mathbf q
        =
        \gamma_i^{\mathrm C}\mathcal G^0\mathbf q
        \in
        \mathcal H_i^{\mathrm{reg}}.
        \label{eq:supp_source_map_compactness}
    \end{equation}
    Compactness first holds in the ambient space $\mathbf X_i$. Its range lies in the closed subspace $\mathcal H_i^{\mathrm{reg}}$ by \hyperref[supp:trace_spaces]{Supplementary Note~\ref*{supp:trace_spaces}}, so it is also compact with that subspace as its codomain.

    Now let $i\ne j$. Pairwise disjointness of the closed balls permits a neighborhood $U_i\supset\overline B_i$ whose closure is disjoint from $\overline B_j$ and hence has positive distance from $\Gamma_j$. For $\mathbf f_j\in\mathcal H_j^{\mathrm{out}}$, \hyperref[supp:trace_spaces]{Supplementary Note~\ref*{supp:trace_spaces}} gives a unique outgoing representative, and the exterior representation formula identifies it on $U_i$ with $\mathcal P_{\Gamma_j}^{\mathrm{out}}\mathbf f_j$. Lemma~\ref{lem:supp_separated_smoothing} and \eqref{eq:supp_smooth_cauchy_trace_bound} then show that
    \begin{equation}
        \mathcal W_{ij}\mathbf f_j
        =
        \gamma_i^{\mathrm C}
        \mathcal P_{\Gamma_j}^{\mathrm{out}}\mathbf f_j
        \label{eq:supp_continuation_map_representation}
    \end{equation}
    is compact into $\mathbf X_i$. The represented field satisfies the homogeneous background Maxwell equation throughout $B_i$, so the range lies in the closed subspace $\mathcal H_i^{\mathrm{reg}}$. This proves both well-definedness and compactness of $\mathcal W_{ij}$.

    Finally, choose a bounded Lipschitz neighborhood $U_{\mathrm{obs}}$ of $\Xi_{\mathrm{obs}}$ whose closure is disjoint from every $\overline B_i$. For a volume observation support, restriction from $H^m(U_{\mathrm{obs}})$ to $L^2(\Xi_{\mathrm{obs}})$ is bounded. For a Lipschitz surface observation support, the Sobolev trace theorem gives the same conclusion for $m>1/2$; selecting full or tangential components is also bounded in $L^2$. Lemma~\ref{lem:supp_separated_smoothing}, followed by this restriction, proves that
    \begin{equation}
        \mathcal O_i\mathbf f_i
        =
        \left.
        \mathcal P_{\Gamma_i}^{\mathrm{out}}\mathbf f_i
        \right|_{\Xi_{\mathrm{obs}}}
        \label{eq:supp_observation_map_representation}
    \end{equation}
    is compact from $\mathcal H_i^{\mathrm{out}}$ to the unweighted observation space and hence, by \eqref{eq:supp_observation_weight_equivalence}, to $\mathcal Y$.

    A map into a finite Hilbert direct sum is compact when each component is compact, a finite block operator is compact when each block is compact, and a finite sum of compact maps is compact. These prove the final statement for $\mathcal S$, $\mathcal W$, and $\mathcal O$.
\end{proof}

Because the local transition operator $\mathcal T$ is bounded by the isolated-body assumption, Proposition~\ref{prop:supp_compact_external_maps} implies that $\mathcal W\mathcal T$ is compact on $\mathcal H^{\mathrm{reg}}$. It also proves directly that the single-body map $\mathcal O_D\mathcal T_D\mathcal S_D$ is compact. Once the collective inverse in Section~II of the main paper exists, the same proposition and boundedness of the collective transition yield compactness of $\mathcal G^{\mathrm{sca}}=\mathcal O\mathcal T(\mathbb I-\mathcal W\mathcal T)^{-1}\mathcal S$.

\section{Fredholm Solvability and Exactness of the Trace Assembly}
\label{supp:fredholm_exactness}

This note supplies the Fredholm and field-reconstruction details used twice in Section~II of the main paper: first to establish the collective inverse, and then to identify the trace assembly with the physical multiple-scattering solution. Only the continuous trace spaces and operators are used. Let $B_i$, $1\leq i\leq N$, be the finite family of enclosing balls, with pairwise disjoint closures, and set
\begin{equation}
    \Omega_B
    :=
    \mathbb R^3
    \setminus
    \bigcup_{i=1}^{N}\overline B_i.
    \label{eq:supp_common_ball_exterior}
\end{equation}
\hyperref[supp:trace_spaces]{Supplementary Note~\ref*{supp:trace_spaces}} gives unique regular and outgoing field representatives of the trace states, while \hyperref[supp:separated_compact_maps]{Supplementary Note~\ref*{supp:separated_compact_maps}} proves that
\begin{equation}
    \mathcal W:
    \mathcal H^{\mathrm{out}}
    \to
    \mathcal H^{\mathrm{reg}}
    \quad\text{is compact}.
    \label{eq:supp_compact_collective_continuation}
\end{equation}
We first isolate the uniqueness property of the bodywise outgoing representation. The result concerns only the artificial balls and the homogeneous background equation; it does not invoke any material field inside a scatterer.

\begin{lemma}[Bodywise outgoing uniqueness]
    \label{lem:supp_bodywise_outgoing_uniqueness}
    For each $i$, let
    $\mathbf E_i\in\mathcal M_k(\mathbb R^3\setminus\overline B_i)$ be outgoing. If
    \begin{equation}
        \sum_{i=1}^{N}\mathbf E_i
        =
        \mathbf0
        \quad\text{in }\Omega_B,
        \label{eq:supp_zero_bodywise_sum}
    \end{equation}
    then $\mathbf E_i=\mathbf0$ in
    $\mathbb R^3\setminus\overline B_i$ for every $i$.
\end{lemma}

\begin{proof}
    Fix $i$. Pairwise disjointness and finiteness of the closed balls permit a bounded open neighborhood $V_i\supset\overline B_i$ such that $V_i\cap\overline B_j=\varnothing$ for every $j\ne i$. Hence
    \begin{equation}
        \mathbf F_i
        :=
        -\sum_{j\ne i}\mathbf E_j
        \label{eq:supp_bodywise_extension_field}
    \end{equation}
    is a homogeneous background Maxwell field throughout $V_i$. On the nonempty open overlap
    $V_i\setminus\overline B_i\subset\Omega_B$, equation~\eqref{eq:supp_zero_bodywise_sum} gives
    $\mathbf F_i=\mathbf E_i$. The two fields therefore glue to a field
    $\widetilde{\mathbf E}_i\in H_{\mathrm{loc}}(\operatorname{curl}^2,\mathbb R^3)$ satisfying
    \begin{equation}
        \nabla\times\nabla\times\widetilde{\mathbf E}_i
        -k^2\widetilde{\mathbf E}_i
        =
        \mathbf0
        \quad\text{in }\mathbb R^3.
        \label{eq:supp_entire_maxwell_extension}
    \end{equation}
    Indeed, use $\mathbf F_i$ on $V_i$ and $\mathbf E_i$ on
    $\mathbb R^3\setminus\overline B_i$; equality on their open overlap makes the definition and its distributional curls consistent. The glued field is outgoing because it equals $\mathbf E_i$ outside the bounded set $V_i$. The Maxwell Rellich lemma and unique continuation then give
    $\widetilde{\mathbf E}_i=\mathbf0$ in $\mathbb R^3$ \cite{nedelec2001,colton2013}. Its restriction to
    $\mathbb R^3\setminus\overline B_i$ is $\mathbf E_i$, so $\mathbf E_i=\mathbf0$. Repeating the argument for every $i$ proves the claim.
\end{proof}

The next proposition turns physical uniqueness into invertibility of the collective trace equation. Isolated-body well-posedness has the precise role stated in the main paper: it makes each local transition operator $\mathcal T_i$ a bounded map and ensures that its output is the outgoing response satisfying the local material or boundary conditions.

\begin{proposition}[Bounded invertibility of the collective trace operator]
    \label{prop:supp_collective_trace_invertibility}
    Suppose the local transition operator
    $\mathcal T:\mathcal H^{\mathrm{reg}}\to\mathcal H^{\mathrm{out}}$ is bounded and the homogeneous radiating multiple-scattering problem is unique at the frequency considered. Then
    \begin{equation}
        \mathbb I-\mathcal W\mathcal T:
        \mathcal H^{\mathrm{reg}}
        \to
        \mathcal H^{\mathrm{reg}}
        \label{eq:supp_collective_trace_operator}
    \end{equation}
    is boundedly invertible. Writing $\mathcal B(U,V)$ for the bounded linear operators from $U$ to $V$, and $\mathcal B(U):=\mathcal B(U,U)$, one consequently has
    \begin{equation}
        \begin{aligned}
            \mathcal R
             & :=
            (\mathbb I-\mathcal W\mathcal T)^{-1}
            \in
            \mathcal B(\mathcal H^{\mathrm{reg}}), \\
            \mathcal K
             & :=
            \mathcal T\mathcal R
            \in
            \mathcal B(\mathcal H^{\mathrm{reg}},\mathcal H^{\mathrm{out}}).
        \end{aligned}
        \label{eq:supp_bounded_collective_resolvent}
    \end{equation}
\end{proposition}

\begin{proof}
    By \eqref{eq:supp_compact_collective_continuation} and boundedness of $\mathcal T$, the operator
    $\mathcal A:=\mathcal W\mathcal T$ is compact on
    $\mathcal H^{\mathrm{reg}}$. Thus $\mathbb I-\mathcal A$ is Fredholm of index zero \cite{kato1995}. It remains to show that its kernel is trivial. Let
    $\mathbf f^{\mathrm{inc}}\in\ker(\mathbb I-\mathcal W\mathcal T)$ and set
    \begin{equation}
        \mathbf f^{\mathrm{sca}}
        :=
        \mathcal T\mathbf f^{\mathrm{inc}}.
        \label{eq:supp_homogeneous_scattered_trace}
    \end{equation}
    For each body, \hyperref[supp:trace_spaces]{Supplementary Note~\ref*{supp:trace_spaces}} reconstructs from
    $\mathbf f_i^{\mathrm{inc}}$ a unique field
    $\mathbf E_i^{\mathrm{inc}}\in\mathcal M_k(B_i)$ and from
    $\mathbf f_i^{\mathrm{sca}}$ a unique outgoing field
    $\mathbf E_i^{\mathrm{sca}}\in
        \mathcal M_k(\mathbb R^3\setminus\overline B_i)$. The $i$th component of the kernel equation is
    \begin{equation}
        \mathbf f_i^{\mathrm{inc}}
        =
        \sum_{j\ne i}\mathcal W_{ij}\mathbf f_j^{\mathrm{sca}}
        =
        \gamma_i^{\mathrm C}
        \left(\sum_{j\ne i}\mathbf E_j^{\mathrm{sca}}\right).
        \label{eq:supp_homogeneous_foldy_lax_trace}
    \end{equation}
    Every field in the sum on the right is regular in $B_i$. The injectivity of the regular Cauchy-trace correspondence in Proposition~\ref{prop:supp_trace_range_characterization} therefore upgrades the trace identity to
    \begin{equation}
        \mathbf E_i^{\mathrm{inc}}
        =
        \sum_{j\ne i}\mathbf E_j^{\mathrm{sca}}
        \quad\text{in }B_i.
        \label{eq:supp_homogeneous_incident_field_balance}
    \end{equation}

    By the definition of $\mathcal T_i$ and isolated-body well-posedness, there is an outgoing physical response
    $\widehat{\mathbf E}_i^{\mathrm{sca}}$ in the background exterior of the $i$th material support such that
    $\gamma_i^{\mathrm C}\widehat{\mathbf E}_i^{\mathrm{sca}}=\mathbf f_i^{\mathrm{sca}}$ and
    $\mathbf E_i^{\mathrm{inc}}+\widehat{\mathbf E}_i^{\mathrm{sca}}$ satisfies the conditions encoded by $\Theta_i$. On
    $\mathbb R^3\setminus\overline B_i$, this response and the reconstructed field
    $\mathbf E_i^{\mathrm{sca}}$ are outgoing background solutions with the same complete Cauchy trace. Exterior trace injectivity in Proposition~\ref{prop:supp_trace_range_characterization} makes them equal there. We may therefore use the same symbol
    $\mathbf E_i^{\mathrm{sca}}$ for the physical exterior response continued from the material boundary through $\Gamma_i$ to infinity; no field inside the material support is used. Equation~\eqref{eq:supp_homogeneous_incident_field_balance} now shows that
    \begin{equation}
        \mathbf E^{h}
        :=
        \sum_{i=1}^{N}\mathbf E_i^{\mathrm{sca}}
        \label{eq:supp_assembled_homogeneous_field}
    \end{equation}
    satisfies the local conditions at every body and the radiation condition, with no external excitation. It is thus a homogeneous physical multiple-scattering solution. By the assumed collective uniqueness, $\mathbf E^{h}=\mathbf0$ in the exterior background and, in particular, in $\Omega_B$. Lemma~\ref{lem:supp_bodywise_outgoing_uniqueness} then gives
    $\mathbf E_i^{\mathrm{sca}}=\mathbf0$ for every $i$, hence
    $\mathbf f^{\mathrm{sca}}=\mathbf0$. Finally, the homogeneous Foldy--Lax balance gives
    \begin{equation}
        \mathbf f^{\mathrm{inc}}
        =
        \mathcal W\mathbf f^{\mathrm{sca}}
        =
        \mathbf0.
        \label{eq:supp_trivial_collective_kernel}
    \end{equation}
    This last step is important: no injectivity assumption on any $\mathcal T_i$ is used.

    Hence $\mathbb I-\mathcal A$ is an injective Fredholm operator of index zero and is therefore surjective. Its inverse is bounded by the bounded inverse theorem, which proves \eqref{eq:supp_bounded_collective_resolvent}; boundedness of $\mathcal K$ follows by composition.
\end{proof}

We finally spell out the exactness assertion in Proposition~1 of the main paper.

\begin{proposition}[Exactness of the continuous trace assembly]
    \label{prop:supp_exact_trace_assembly}
    Under the assumptions of Proposition~\ref{prop:supp_collective_trace_invertibility}, for every source
    $\mathbf q\in\mathcal X$, the trace solution
    \begin{equation}
        \mathbf f^{\mathrm{inc}}
        =
        \mathcal R\mathcal S\mathbf q,
        \quad
        \mathbf f^{\mathrm{sca}}
        =
        \mathcal T\mathbf f^{\mathrm{inc}}
        \label{eq:supp_solved_trace_states}
    \end{equation}
    reconstructs the physical multiple-scattering solution. Moreover,
    $\mathcal O\mathbf f^{\mathrm{sca}}$ is precisely its selected scattered-field component on $\Xi_{\mathrm{obs}}$.
\end{proposition}

\begin{proof}
    Let $\mathbf E^0=\mathcal G^0\mathbf q$. For each $i$, reconstruct the unique regular representative
    $\mathbf E_i^{\mathrm{inc}}$ of $\mathbf f_i^{\mathrm{inc}}$ and the unique outgoing representative
    $\mathbf E_i^{\mathrm{sca}}$ of $\mathbf f_i^{\mathrm{sca}}$. Since \eqref{eq:supp_solved_trace_states} solves the collective trace equation, its $i$th component gives
    \begin{equation}
        \gamma_i^{\mathrm C}\mathbf E_i^{\mathrm{inc}}
        =
        \gamma_i^{\mathrm C}
        \left(
        \mathbf E^0
        +
        \sum_{j\ne i}\mathbf E_j^{\mathrm{sca}}
        \right).
        \label{eq:supp_driven_foldy_lax_trace}
    \end{equation}
    Both sides of \eqref{eq:supp_driven_foldy_lax_trace} are traces of regular background fields in $B_i$. Proposition~\ref{prop:supp_trace_range_characterization} therefore gives the field identity
    \begin{equation}
        \mathbf E_i^{\mathrm{inc}}
        =
        \mathbf E^0
        +
        \sum_{j\ne i}\mathbf E_j^{\mathrm{sca}}
        \quad\text{in }B_i.
        \label{eq:supp_driven_incident_field_balance}
    \end{equation}
    As in the proof of Proposition~\ref{prop:supp_collective_trace_invertibility}, the local relation
    $\mathbf f_i^{\mathrm{sca}}=\mathcal T_i\mathbf f_i^{\mathrm{inc}}$ identifies the reconstructed outgoing representative with the isolated-body exterior response, for which
    $\mathbf E_i^{\mathrm{inc}}+\mathbf E_i^{\mathrm{sca}}$ satisfies the local conditions. Thus \eqref{eq:supp_driven_incident_field_balance} shows that the assembled fields
    \begin{equation}
        \mathbf E^{\mathrm{sca}}
        :=
        \sum_{i=1}^{N}\mathbf E_i^{\mathrm{sca}},
        \quad
        \mathbf E^{\mathrm{tot}}
        :=
        \mathbf E^0+\mathbf E^{\mathrm{sca}}
        \label{eq:supp_assembled_driven_fields}
    \end{equation}
    satisfy the prescribed local conditions at every body. The scattered field is outgoing because it is a finite sum of outgoing fields. Consequently, \eqref{eq:supp_assembled_driven_fields} is a physical solution of the driven multiple-scattering problem, and collective uniqueness identifies it with the physical solution.

    The trace family is unique because
    $\mathbb I-\mathcal W\mathcal T$ is invertible. Its bodywise outgoing fields are also unique: the difference of two admissible bodywise representations has zero sum in $\Omega_B$, so Lemma~\ref{lem:supp_bodywise_outgoing_uniqueness} applies. This establishes the ``exactly'' statement in Proposition~1 without assuming that any $\mathcal T_i$ is injective and without expanding a multiple-reflection Neumann series. Finally, the definition of the observation map gives
    \begin{equation}
        \mathcal O\mathbf f^{\mathrm{sca}}
        =
        \left.
        \sum_{i=1}^{N}\mathbf E_i^{\mathrm{sca}}
        \right|_{\Xi_{\mathrm{obs}}}
        =
        \left.\mathbf E^{\mathrm{sca}}\right|_{\Xi_{\mathrm{obs}}},
        \label{eq:supp_exact_observed_field}
    \end{equation}
    with the component selection prescribed by $\mathcal Y$.
\end{proof}

The argument above uses material information only through the isolated-body response encoded by the bounded operators $\mathcal T_i$ and through collective uniqueness. The trace-level assembly and its proof do not require the interior material field to be retained as part of the global state, although isolated-body well-posedness remains an input through $\mathcal T_i$. They also require no inverse construction from an abstract transition operator to material geometry. Maxwell representation and boundary-decomposition formulations for standard disjoint scatterers provide the corresponding bodywise exterior representation \cite{nedelec2001,buffa2003maxwellbem,balabane2004,claeys2012multitrace}.

\section{Fixed VSH Trace Metric and VSWF Trace Projections}
\label{supp:vswf_trace_metric}

This note supplies the details used in Section~III-B of the main paper. It first specifies, on each spherical interface, the fixed Hilbert metric that was left abstract up to norm equivalence in Section~II. It then derives the VSWF trace Gram matrix from that metric and proves the projection and strong-convergence statements. The underlying trace ranges are exactly those characterized in \hyperref[supp:trace_spaces]{Supplementary Note~\ref*{supp:trace_spaces}}; no volume field inside a material support is introduced.

Fix one enclosing sphere $\Gamma_i=\partial B_{R_i}(\mathbf c_i)$ and put $x_i:=kR_i>0$. Let
$S^2:=\{\widehat{\mathbf r}\in\mathbb R^3:|\widehat{\mathbf r}|=1\}$ denote the unit sphere, and identify each $\mathbf r\in\Gamma_i$ with its local direction
$\widehat{\mathbf r}_i:=(\mathbf r-\mathbf c_i)/R_i\in S^2$. For $l\geq1$ and $-l\leq m\leq l$, let
\begin{equation}
    Y_{lm}(\theta,\phi)
    :=
    \widehat P_l^{|m|}(\cos\theta)\mathrm e^{\mathrm jm\phi},
    \label{eq:supp_scalar_spherical_harmonic}
\end{equation}
with the normalization used in the main paper, and define the tangential VSHs on $S^2$ by
\begin{equation}
    \mathbf A_{lm}
    :=
    -\frac{\widehat{\mathbf r}_i\times\nabla_{S^2}Y_{lm}}
    {\sqrt{2l(l+1)}},
    \quad
    \mathbf B_{lm}
    :=
    \widehat{\mathbf r}_i\times\mathbf A_{lm}.
    \label{eq:supp_vsh_definition}
\end{equation}
Here $\nabla_{S^2}$ is the surface gradient on the unit sphere. Thus $\mathbf A_{lm}$ is surface-divergence-free and $\mathbf B_{lm}$ is of surface-gradient type. Orthogonality of the scalar spherical harmonics and integration by parts on $S^2$ give
\begin{align}
    \int_{S^2}\mathbf A_{lm}^{*}\mathbin{\cdot}\mathbf A_{l'm'}\,\mathrm d\Omega
     & =
    \pi\delta_{ll'}\delta_{mm'}, \\
    \int_{S^2}\mathbf B_{lm}^{*}\mathbin{\cdot}\mathbf B_{l'm'}\,\mathrm d\Omega
     & =
    \pi\delta_{ll'}\delta_{mm'}, \\
    \int_{S^2}\mathbf A_{lm}^{*}\mathbin{\cdot}\mathbf B_{l'm'}\,\mathrm d\Omega
     & =0.
    \label{eq:supp_vsh_orthogonality}
\end{align}
The same fields are used on $\Gamma_i$ by angular pullback; hence surface integration on $\Gamma_i$ contributes the factor $R_i^2$.

Every tangential distribution on the sphere has a VSH expansion. For a Cauchy pair $\mathbf f=(\mathbf f_X,\mathbf f_Y)$, let $X$ and $Y$ label its two trace slots. These are local slot labels in this note and are unrelated to the global source and observation spaces $\mathcal X$ and $\mathcal Y$. Write formally
\begin{align}
    \mathbf f_X
     & =
    \sum_{l=1}^{\infty}\sum_{m=-l}^{l}
    \left(
    a_{X,lm}\mathbf A_{lm}
    +b_{X,lm}\mathbf B_{lm}
    \right),
    \label{eq:supp_x_trace_vsh_expansion} \\
    \mathbf f_Y
     & =
    \sum_{l=1}^{\infty}\sum_{m=-l}^{l}
    \left(
    a_{Y,lm}\mathbf A_{lm}
    +b_{Y,lm}\mathbf B_{lm}
    \right).
    \label{eq:supp_y_trace_vsh_expansion}
\end{align}
Collect the two Cauchy-slot coefficients as
\begin{equation}
    \mathbf a_{lm}
    :=
    \begin{pmatrix}a_{X,lm}\\a_{Y,lm}\end{pmatrix},
    \quad
    \mathbf b_{lm}
    :=
    \begin{pmatrix}b_{X,lm}\\b_{Y,lm}\end{pmatrix}.
    \label{eq:supp_cauchy_slot_coefficient_vectors}
\end{equation}
Set $\lambda_l:=l(l+1)$ and
\begin{equation}
    w_{l,i}^{\pm}
    :=
    \left(1+\frac{\lambda_l}{x_i^2}\right)^{\pm1/2}.
    \label{eq:supp_vsh_weights}
\end{equation}
For finite VSH sums, define the product inner product, conjugate-linear in its first argument, by
\begin{equation}
    \langle\mathbf f,\mathbf g\rangle_{i,\mathrm{VSH}}
    :=
    \frac{\pi R_i^2}{2\eta}
    \sum_{l=1}^{\infty}\sum_{m=-l}^{l}
    \bigl(
    w_{l,i}^{-}\mathbf a_{lm}^{\dagger}\widetilde{\mathbf a}_{lm}
    +
    w_{l,i}^{+}\mathbf b_{lm}^{\dagger}\widetilde{\mathbf b}_{lm}
    \bigr),
    \label{eq:supp_fixed_vsh_product_metric}
\end{equation}
Here $\widetilde{\mathbf a}_{lm}$ and $\widetilde{\mathbf b}_{lm}$ are the corresponding coefficients of $\mathbf g$. The factor $1/(2\eta)$ is the fixed physical normalization selected in the main paper; it is not needed for norm equivalence.

\begin{lemma}[Equivalence of the fixed VSH metric]
    \label{lem:supp_vsh_metric_equivalence}
    The completion under \eqref{eq:supp_fixed_vsh_product_metric} is
    \begin{equation}
        \mathbf X_i
        =
        H^{-1/2}(\operatorname{div}_{\Gamma_i},\Gamma_i)^2,
        \label{eq:supp_vsh_metric_completion}
    \end{equation}
    with an equivalent Hilbert norm. The equivalence constants may depend on the fixed quantities $k$, $R_i$, and $\eta$, but are independent of $l$ and of every truncation order $L$. Consequently, restricting this metric to $\mathcal H_i^{\mathrm{reg}}$ or $\mathcal H_i^{\mathrm{out}}$ gives the same trace-space topology used in Section~II of the main paper.
\end{lemma}

\begin{proof}
    On a sphere, the standard spectral characterization of
    $H^{-1/2}(\operatorname{div}_{\Gamma_i},\Gamma_i)$ assigns, up to constants depending only on the fixed sphere, weights equivalent to
    $(1+\lambda_l)^{-1/2}$ to the surface-divergence-free coefficients and
    $(1+\lambda_l)^{1/2}$ to the surface-gradient coefficients
    \cite{buffa2002traces,kristensson2020}. For fixed $x_i>0$,
    \begin{equation}
        \min\{1,x_i^{-2}\}
        \leq
        \frac{1+\lambda_l/x_i^2}{1+\lambda_l}
        \leq
        \max\{1,x_i^{-2}\}
        \quad(l\geq1).
        \label{eq:supp_uniform_weight_comparison}
    \end{equation}
    Taking the positive and negative square roots of
    \eqref{eq:supp_uniform_weight_comparison} shows that
    $w_{l,i}^{-}$ and $w_{l,i}^{+}$ are uniformly equivalent, respectively, to the two standard spectral weights. Multiplication by the fixed positive factor $\pi R_i^2/(2\eta)$ and taking two Hilbert-product copies do not change the topology. This proves \eqref{eq:supp_vsh_metric_completion}. The trace ranges are closed by Proposition~\ref{prop:supp_trace_range_characterization}, so the restricted spaces remain Hilbert spaces under this equivalent metric.
\end{proof}

We next compute the traces of the raw VSWFs. For $s\in\{\mathrm{reg},\mathrm{out}\}$, define
\begin{equation}
    d_l^s(x)
    :=
    \frac{[xz_l^s(x)]'}{x}.
    \label{eq:supp_radial_derivative_factor}
\end{equation}
The definitions in Section~III-A of the main paper and the VSH convention \eqref{eq:supp_vsh_definition} give, on $\Gamma_i$,
\begin{equation}
    \left.\mathbf U_{lm}^{s}\right|_{\Gamma_i}
    =
    z_l^s(x_i)\mathbf A_{lm},
    \quad
    \left.(\mathbf V_{lm}^{s})_t\right|_{\Gamma_i}
    =
    d_l^s(x_i)\mathbf B_{lm}.
    \label{eq:supp_vswf_tangential_components}
\end{equation}
Moreover, $k^{-1}\nabla\times\mathbf U_{lm}^s=\mathbf V_{lm}^s$ by definition, while the source-free Maxwell equation gives
$k^{-1}\nabla\times\mathbf V_{lm}^s=\mathbf U_{lm}^s$. Since
$\widehat{\mathbf r}_i\times\mathbf A_{lm}=\mathbf B_{lm}$ and
$\widehat{\mathbf r}_i\times\mathbf B_{lm}=-\mathbf A_{lm}$, the Cauchy traces of the two polarizations are therefore
\begin{align}
    \bm\phi_{i,(1,l,m)}^s
     & =
    \begin{pmatrix}
        z_l^s(x_i)\mathbf B_{lm} \\
        -d_l^s(x_i)\mathbf A_{lm}
    \end{pmatrix}, \\
    \bm\phi_{i,(2,l,m)}^s
     & =
    \begin{pmatrix}
        -d_l^s(x_i)\mathbf A_{lm} \\
        z_l^s(x_i)\mathbf B_{lm}
    \end{pmatrix}.
    \label{eq:supp_vswf_cauchy_traces}
\end{align}
Only tangential components enter \eqref{eq:supp_vswf_cauchy_traces}; the radial component of $\mathbf V_{lm}^s$ is annihilated by the cross product in the Cauchy trace.

For $\alpha=(\tau,l,m)$ and $\beta=(\tau',l',m')$, substituting \eqref{eq:supp_vswf_cauchy_traces} into the fixed metric and using \eqref{eq:supp_vsh_orthogonality} yields
\begin{equation}
    \left\langle\bm\phi_{i,\alpha}^s,
    \bm\phi_{i,\beta}^s\right\rangle_{i,\mathrm{VSH}}
    =
    \delta_{\alpha\beta}g_{l,i}^s,
    \label{eq:supp_vswf_trace_inner_products}
\end{equation}
where
\begin{equation}
    g_{l,i}^s
    :=
    \frac{\pi R_i^2}{2\eta}
    \left[
        w_{l,i}^{+}|z_l^s(x_i)|^2
        +
        w_{l,i}^{-}|d_l^s(x_i)|^2
        \right].
    \label{eq:supp_vswf_trace_weight}
\end{equation}
The two polarizations have the same weight because they exchange the two VSH types between the two components of the Cauchy pair. Also, $g_{l,i}^s>0$: if both $z_l^s(x_i)$ and $d_l^s(x_i)$ vanished, then $z_l^s(x_i)=[z_l^s]'(x_i)=0$, contradicting uniqueness for the second-order spherical Bessel equation. Thus no nonzero raw VSWF trace has zero norm.

\begin{proposition}[Complete trace basis and stable modal projections]
    \label{prop:supp_vswf_trace_projections}
    For each $s\in\{\mathrm{reg},\mathrm{out}\}$, the family
    $\{\bm\phi_{i,\alpha}^s:\alpha\in\Lambda\}$ is a complete orthogonal family in $\mathcal H_i^s$. Let $P_{i,L}^s$ be the synthesis map defined in Section~III-B of the main paper and let
    \begin{equation}
        \mathbf G_{i,L}^s
        :=
        (P_{i,L}^s)^\dagger P_{i,L}^s,
        \quad
        Q_{i,L}^s
        :=
        (\mathbf G_{i,L}^s)^{-1}(P_{i,L}^s)^\dagger,
        \label{eq:supp_modal_analysis_definition}
    \end{equation}
    where $\dagger$ denotes the adjoint induced by the fixed trace inner product on $\mathcal H_i^s$ and the Euclidean inner product on $\mathbb C^{U_L}$.
    Then $\mathbf G_{i,L}^s$ is Hermitian positive definite and has the diagonal entries $g_{l,i}^s$ defined in \eqref{eq:supp_vswf_trace_weight}. Furthermore,
    \begin{equation}
        Q_{i,L}^sP_{i,L}^s=\mathbf I,
        \quad
        \Pi_{i,L}^s
        :=
        P_{i,L}^sQ_{i,L}^s
        \label{eq:supp_local_projection_identities}
    \end{equation}
    is the orthogonal projection onto $\operatorname{ran}P_{i,L}^s$. The ranges are nested, and
    \begin{equation}
        \|\Pi_{i,L}^s\|=1,
        \quad
        \Pi_{i,L}^s
        \xrightarrow[L\to\infty]{\mathrm{strong}}
        \mathbb I_{\mathcal H_i^s}.
        \label{eq:supp_local_projection_strong_convergence}
    \end{equation}
    The same conclusions hold for the finite direct sums over all bodies.
\end{proposition}

\begin{proof}
    Separation of variables for source-free Maxwell fields on a sphere gives a regular VSWF expansion in $B_i$ and a radiating outgoing VSWF expansion in
    $\mathbb R^3\setminus\overline B_i$, respectively
    \cite{chew1995,doicu2006,kristensson2020}. At the trace level, the two range conditions select precisely the radial branch $j_l$ or $h_l^{(1)}$ in each angular and polarization block. Equations~\eqref{eq:supp_vswf_cauchy_traces}, \eqref{eq:supp_vswf_trace_inner_products}, and \eqref{eq:supp_vswf_trace_weight} therefore identify each trace range with the weighted coefficient space generated by the corresponding raw traces. Finite angular partial sums converge in the metric \eqref{eq:supp_fixed_vsh_product_metric}; hence their union is dense. Equation~\eqref{eq:supp_vswf_trace_inner_products} gives orthogonality, while \eqref{eq:supp_vswf_trace_weight} and the preceding radial-ODE argument give strict positivity. This establishes completeness and linear independence.

    The Gram identity in \eqref{eq:supp_modal_analysis_definition} and \eqref{eq:supp_vswf_trace_inner_products} gives
    \begin{equation}
        [\mathbf G_{i,L}^s]_{u,u'}
        =
        \delta_{uu'}g_{l,i}^s,
        \quad
        u=\xi_L(\tau,l,m),
        \label{eq:supp_finite_trace_gram_entries}
    \end{equation}
    so the matrix is Hermitian positive definite. It follows immediately that
    $Q_{i,L}^sP_{i,L}^s=\mathbf I$. Moreover,
    \begin{equation}
        \Pi_{i,L}^s
        =
        P_{i,L}^s(\mathbf G_{i,L}^s)^{-1}(P_{i,L}^s)^\dagger
        \label{eq:supp_explicit_orthogonal_projection}
    \end{equation}
    is self-adjoint and idempotent, and its range is
    $\operatorname{ran}P_{i,L}^s$; it is therefore the asserted orthogonal projection.

    If
    $\mathbf f=\sum_{\alpha\in\Lambda}a_\alpha\bm\phi_{i,\alpha}^s$
    in $\mathcal H_i^s$, orthogonality gives
    \begin{equation}
        \left\|
        \mathbf f-\Pi_{i,L}^s\mathbf f
        \right\|_{i,\mathrm{VSH}}^2
        =
        \sum_{\substack{l>L\\-l\leq m\leq l}}
        \sum_{\tau=1}^{2}
        |a_{\tau lm}|^2g_{l,i}^s
        \to0.
        \label{eq:supp_projection_tail_convergence}
    \end{equation}
    This proves strong convergence, while every nonzero orthogonal projection has norm one. The sets $\Lambda_L$ are nested, so the projection ranges are nested. Finally, for finitely many bodies, the norm squared of a direct-sum error is the finite sum of the bodywise squared errors. Thus the direct-sum projections also converge strongly and have norm one.

    Finally, if $\mathbf f\in\mathcal H_i^s$ has raw expansion coefficients
    $\{a_\alpha\}_{\alpha\in\Lambda}$, the same calculation gives
    \begin{equation}
        [Q_{i,L}^s\mathbf f]_u
        =
        a_{\xi_L^{-1}(u)},
        \quad u=1,\ldots,U_L,
        \label{eq:supp_raw_coefficient_extraction}
    \end{equation}
    which is the coefficient-extraction statement used in the main paper.
\end{proof}

\section{Surface Interface Gramians and Channel Recovery}
\label{supp:surface_interface_grams}

This note supplies the Hilbert-adjoint calculations behind the source and observation interface Gramians and the explicit surface-channel recovery formulas in Section~III-E of the main paper. It uses only the tangential surface-to-surface specialization already fixed there. Throughout, Hilbert inner products are conjugate-linear in their first argument. For operators, $\dagger$ denotes the adjoint induced by the stated Hilbert metrics; for coordinate vectors and fields, it denotes conjugate transpose.

Recall that
\begin{equation}
    \mathcal S_L:\mathcal X\to\mathbb C^{N_L},
    \quad
    \mathcal O_L:\mathbb C^{N_L}\to\mathcal Y,
    \label{eq:supp_surface_interface_map_directions}
\end{equation}
and that
\begin{equation}
    \mathbf C_{S,L}
    =
    \mathcal S_L\mathcal S_L^\dagger,
    \quad
    \mathbf C_{O,L}
    =
    \mathcal O_L^\dagger\mathcal O_L.
    \label{eq:supp_surface_interface_gram_definitions}
\end{equation}
Let $\mathbf e_p$ be the $p$th Euclidean coordinate vector in $\mathbb C^{N_L}$. With
$\mathbf q=(\mathbf J_t^{\mathsf T},\mathbf M_t^{\mathsf T})^{\mathsf T}$ on $\Sigma_{\mathrm{src}}$ and tangential electric fields on $\Sigma_{\mathrm{obs}}$, the fixed surface metrics in the main paper are
\begin{align}
    \langle\mathbf q_1,\mathbf q_2\rangle_{\mathcal X}
     & =
    \frac12
    \int_{\Sigma_{\mathrm{src}}}
    \left(
    \eta\mathbf J_{1,t}^{\dagger}\mathbf J_{2,t}
    +
    \eta^{-1}\mathbf M_{1,t}^{\dagger}\mathbf M_{2,t}
    \right)\,\mathrm dS,
    \label{eq:supp_surface_source_metric} \\
    \langle\mathbf E_1,\mathbf E_2\rangle_{\mathcal Y}
     & =
    \frac{1}{2\eta}
    \int_{\Sigma_{\mathrm{obs}}}
    \mathbf E_{1,t}^{\dagger}\mathbf E_{2,t}\,\mathrm dS.
    \label{eq:supp_surface_observation_metric}
\end{align}

For the cluster label $p=\zeta_L(i,u)$, use the notation of the main paper
$\bm\Psi_p(\mathbf r)=\bm\Phi_u^{\mathrm{out}}(\mathbf r-\mathbf c_i)$ and let $\bm\Psi_{p,t}$ denote its tangential component on the relevant external surface. Split the retained source map according to the two currents as
\begin{equation}
    \mathcal S_L\mathbf q
    =
    \mathcal S_{J,L}\mathbf J_t
    +
    \mathcal S_{M,L}\mathbf M_t.
    \label{eq:supp_source_map_split}
\end{equation}
Specializing the source-pairing formula of Section~III-C in the main paper and using $\omega\mu=k\eta$ gives
\begin{align}
    [\mathcal S_{J,L}\mathbf J_t]_p
     & =
    -\frac{\eta k^2}{\pi}
    \int_{\Sigma_{\mathrm{src}}}
    \bm\Psi_{\overline p,t}^{\mathsf T}(\mathbf r')
    \mathbf J_t(\mathbf r')\,\mathrm dS',
    \label{eq:supp_surface_electric_source_pairing} \\
    [\mathcal S_{M,L}\mathbf M_t]_p
     & =
    -\frac{\mathrm j k^2}{\pi}
    \int_{\Sigma_{\mathrm{src}}}
    \bm\Psi_{\overline{p^\star},t}^{\mathsf T}(\mathbf r')
    \mathbf M_t(\mathbf r')\,\mathrm dS'.
    \label{eq:supp_surface_magnetic_source_pairing}
\end{align}
The transpose in \eqref{eq:supp_surface_electric_source_pairing}--\eqref{eq:supp_surface_magnetic_source_pairing} is essential: the source expansion contains the bilinear Euclidean dot product of the outgoing VSWF with the impressed current, not a Hermitian field inner product.

For $\mathbf h\in\mathbb C^{N_L}$, write
$\mathcal S_L^\dagger\mathbf h=(\mathbf J_{\mathbf h,t}^{\mathsf T},\mathbf M_{\mathbf h,t}^{\mathsf T})^{\mathsf T}$. The adjoint identity is
\begin{equation}
    \left\langle
    \mathcal S_L^\dagger\mathbf h,\mathbf q
    \right\rangle_{\mathcal X}
    =
    \left\langle
    \mathbf h,\mathcal S_L\mathbf q
    \right\rangle_{\mathbb C^{N_L}}.
    \label{eq:supp_source_adjoint_identity}
\end{equation}
Substitution of \eqref{eq:supp_surface_source_metric} and
\eqref{eq:supp_surface_electric_source_pairing}--\eqref{eq:supp_surface_magnetic_source_pairing} into \eqref{eq:supp_source_adjoint_identity} gives the two Riesz representatives
\begin{align}
    \mathbf J_{\mathbf h,t}(\mathbf r')
     & =
    -\frac{2k^2}{\pi}
    \sum_{p=1}^{N_L}h_p
    \bm\Psi_{\overline p,t}^{*}(\mathbf r'),
    \label{eq:supp_source_adjoint_electric_component} \\
    \mathbf M_{\mathbf h,t}(\mathbf r')
     & =
    \frac{2\mathrm j\eta k^2}{\pi}
    \sum_{p=1}^{N_L}h_p
    \bm\Psi_{\overline{p^\star},t}^{*}(\mathbf r').
    \label{eq:supp_source_adjoint_magnetic_component}
\end{align}
Indeed, the conjugate transpose of \eqref{eq:supp_source_adjoint_electric_component}, multiplied by the source-metric weight $\eta/2$, produces the coefficient $-\eta k^2/\pi$ in \eqref{eq:supp_surface_electric_source_pairing}. For the magnetic block, conjugating the factor $\mathrm j$ in \eqref{eq:supp_source_adjoint_magnetic_component} changes it to $-\mathrm j$, and multiplication by the metric weight $1/(2\eta)$ produces the coefficient $-\mathrm j k^2/\pi$ in \eqref{eq:supp_surface_magnetic_source_pairing}. This verifies the adjoint identity for every $\mathbf q$ and fixes both the complex conjugations and the magnetic-current phase.

The retained observation synthesis is
\begin{equation}
    (\mathcal O_L\mathbf a)(\mathbf r)
    =
    \sum_{p=1}^{N_L}a_p\bm\Psi_{p,t}(\mathbf r),
    \quad \mathbf r\in\Sigma_{\mathrm{obs}}.
    \label{eq:supp_surface_observation_synthesis}
\end{equation}
Using \eqref{eq:supp_surface_observation_metric} in the observation-side adjoint identity gives
\begin{equation}
    [\mathcal O_L^\dagger\mathbf E]_p
    =
    \frac{1}{2\eta}
    \int_{\Sigma_{\mathrm{obs}}}
    \bm\Psi_{p,t}^{\dagger}(\mathbf r)
    \mathbf E_t(\mathbf r)\,\mathrm dS.
    \label{eq:supp_surface_observation_adjoint}
\end{equation}

We can now evaluate both interface Gramians without any coordinate-space ambiguity. Their entries first satisfy
\begin{align}
    [\mathbf C_{S,L}]_{p,q}
     & =
    \left\langle
    \mathcal S_L^\dagger\mathbf e_p,
    \mathcal S_L^\dagger\mathbf e_q
    \right\rangle_{\mathcal X},
    \label{eq:supp_source_gram_adjoint_identity} \\
    [\mathbf C_{O,L}]_{p,q}
     & =
    \left\langle
    \mathcal O_L\mathbf e_p,
    \mathcal O_L\mathbf e_q
    \right\rangle_{\mathcal Y}.
    \label{eq:supp_observation_gram_adjoint_identity}
\end{align}
Substituting \eqref{eq:supp_source_adjoint_electric_component}--\eqref{eq:supp_source_adjoint_magnetic_component} and \eqref{eq:supp_surface_observation_synthesis} into these identities yields
\begin{align}
    [\mathbf C_{S,L}]_{p,q}
     & =
    \frac{2\eta k^4}{\pi^2}
    \int_{\Sigma_{\mathrm{src}}}
    \left(
    \bm\Psi_{\overline p,t}^{\mathsf T}
    \bm\Psi_{\overline q,t}^{*}
    +
    \bm\Psi_{\overline{p^\star},t}^{\mathsf T}
    \bm\Psi_{\overline{q^\star},t}^{*}
    \right)\,\mathrm dS,
    \label{eq:supp_surface_source_gram_entries} \\
    [\mathbf C_{O,L}]_{p,q}
     & =
    \frac{1}{2\eta}
    \int_{\Sigma_{\mathrm{obs}}}
    \bm\Psi_{p,t}^{\dagger}
    \bm\Psi_{q,t}\,\mathrm dS.
    \label{eq:supp_surface_observation_gram_entries}
\end{align}
These are exactly the two interface formulas stated in the main paper. In particular,
\begin{equation}
    \mathbf a^\dagger\mathbf C_{S,L}\mathbf a
    =
    \|\mathcal S_L^\dagger\mathbf a\|_{\mathcal X}^2,
    \quad
    \mathbf a^\dagger\mathbf C_{O,L}\mathbf a
    =
    \|\mathcal O_L\mathbf a\|_{\mathcal Y}^2,
    \label{eq:supp_interface_gram_positivity}
\end{equation}
so both matrices are Hermitian positive semidefinite; neither is assumed invertible.

It remains to connect the adjoint formulas to channel recovery. For a nonzero singular triplet $(\sigma_n,\mathbf u_n,\mathbf v_n)$ of the metric core in the main paper, define
\begin{align}
    \mathbf h_{n,S}
     & :=
    (\mathbf C_{S,L}^{1/2})^+\mathbf v_n,
    \label{eq:supp_source_unwhitened_coordinates} \\
    \mathbf h_{n,O}
     & :=
    (\mathbf C_{O,L}^{1/2})^+\mathbf u_n.
    \label{eq:supp_observation_unwhitened_coordinates}
\end{align}
The abstract recovery is
$\mathbf q_n=\mathcal S_L^\dagger\mathbf h_{n,S}$ and
$\widehat{\mathbf E}_{n,t}^{\mathrm{sca}}=\mathcal O_L\mathbf h_{n,O}$. Applying \eqref{eq:supp_source_adjoint_electric_component}--\eqref{eq:supp_source_adjoint_magnetic_component} and \eqref{eq:supp_surface_observation_synthesis} gives
\begin{align}
    \mathbf J_{n,t}(\mathbf r')
     & =
    -\frac{2k^2}{\pi}
    \sum_{p=1}^{N_L}
    [\mathbf h_{n,S}]_p
    \bm\Psi_{\overline p,t}^{*}(\mathbf r'),
    \label{eq:supp_recovered_electric_current} \\
    \mathbf M_{n,t}(\mathbf r')
     & =
    \frac{2\mathrm j\eta k^2}{\pi}
    \sum_{p=1}^{N_L}
    [\mathbf h_{n,S}]_p
    \bm\Psi_{\overline{p^\star},t}^{*}(\mathbf r'),
    \label{eq:supp_recovered_magnetic_current} \\
    \widehat{\mathbf E}_{n,t}^{\mathrm{sca}}(\mathbf r)
     & =
    \sum_{p=1}^{N_L}
    [\mathbf h_{n,O}]_p
    \bm\Psi_{p,t}(\mathbf r).
    \label{eq:supp_recovered_observation_mode}
\end{align}

For completeness, the pseudoinverses in \eqref{eq:supp_source_unwhitened_coordinates}--\eqref{eq:supp_observation_unwhitened_coordinates} act only on the accessible supports. Since $\sigma_n>0$, one has
$\mathbf v_n\in\operatorname{ran}\mathbf C_{S,L}^{1/2}$ and
$\mathbf u_n\in\operatorname{ran}\mathbf C_{O,L}^{1/2}$. Consequently,
\begin{align}
    \|\mathbf q_n\|_{\mathcal X}^2
     & =
    \mathbf h_{n,S}^\dagger
    \mathbf C_{S,L}\mathbf h_{n,S}
    =1,
    \label{eq:supp_recovered_source_normalization} \\
    \|\widehat{\mathbf E}_{n,t}^{\mathrm{sca}}\|_{\mathcal Y}^2
     & =
    \mathbf h_{n,O}^\dagger
    \mathbf C_{O,L}\mathbf h_{n,O}
    =1.
    \label{eq:supp_recovered_observation_normalization}
\end{align}
The same calculation with two distinct singular vectors gives the orthonormality assertions in the main paper.

Finally,
\begin{equation}
    \mathcal S_L\mathbf q_n
    =
    \mathbf C_{S,L}\mathbf h_{n,S}
    =
    \mathbf C_{S,L}^{1/2}\mathbf v_n.
    \label{eq:supp_recovered_source_modal_action}
\end{equation}
Let $\mathbf P_{O,L}$ be the orthogonal projector onto
$\operatorname{ran}\mathbf C_{O,L}^{1/2}$. The singular-vector equation
$\mathbf C_{O,L}^{1/2}\mathbf K_L\mathbf C_{S,L}^{1/2}\mathbf v_n
    =\sigma_n\mathbf u_n$ implies
\begin{equation}
    \mathbf P_{O,L}\mathbf K_L
    \mathbf C_{S,L}^{1/2}\mathbf v_n
    =
    \sigma_n\mathbf h_{n,O}.
    \label{eq:supp_projected_metric_core_action}
\end{equation}
Because $\ker\mathbf C_{O,L}=\ker\mathcal O_L$, one has
$\mathcal O_L=\mathcal O_L\mathbf P_{O,L}$. Combining this fact with
\eqref{eq:supp_recovered_source_modal_action}--\eqref{eq:supp_projected_metric_core_action} gives
\begin{equation}
    \mathcal G_L^{\mathrm{sca}}\mathbf q_n
    =
    \sigma_n\widehat{\mathbf E}_{n,t}^{\mathrm{sca}}.
    \label{eq:supp_recovered_channel_action}
\end{equation}
This proves the asserted channel action and completes the recovery argument.

\section{External Extraction of Nonspherical Local T-Matrices}
\label{supp:nonspherical_tmatrix_extraction}

This note supplies the reproducibility details for the three nonspherical local T-matrices used in Section~IV-A of the main paper. It concerns only their offline extraction at the fixed order $L=3$. The procedure provides finite matrices for the subsequent implementation tests; it neither establishes SCUFF-EM mesh convergence nor enters the operator-norm convergence proof of the main paper.

\subsection{Discrete bodies and fixed modal order}

All three bodies are embedded in vacuum at $f=1\,\mathrm{GHz}$, with
\begin{align}
    k_0
     & =
    20.958450219516816\,\mathrm{m}^{-1},
    \label{eq:supp_tmatrix_wavenumber} \\
    a_0:=k_0^{-1}
     & =
    0.04771345159236943\,\mathrm m.
    \label{eq:supp_tmatrix_physical_scale}
\end{align}
Their relative permeabilities are one. Under the $\exp(-\mathrm j\omega t)$ convention of the main paper, the relative permittivities of the cube, ellipsoid, and torus are $4$, $10$, and $6+0.5\mathrm j$, respectively. Each SCUFF-EM geometry uses its corresponding mesh with this constant material assignment; the exterior medium is vacuum. Every mesh is centered at the local VSWF expansion origin. The cube has nominal side length $2a_0/\sqrt3$, and the ellipsoid has semiaxes $a_0(1,0.65,0.45)$. The torus is obtained from a generating torus of major radius $1$ and minor radius $0.35$ by centering its discrete nodal bounding box and scaling all coordinates by $0.738a_0$. The exact meshes used in this work are available from the public GitHub repository at \url{https://github.com/wangzhukang/scattering-green-operator-meshes}. The repository contains the Gmsh 2.2 ASCII files \texttt{cube.msh}, \texttt{ellipsoid.msh}, and \texttt{torus.msh}. These files define the discrete bodies and their enclosing radii used below.

For each body $D$, let $R_D^{\mathrm{enc}}$ be the maximum distance from its mesh nodes to the local origin. The retained labels are
\begin{align}
    \Lambda_3
     & =
    \{(\tau,l,m):\tau=1,2,\ 1\leq l\leq3,\ -l\leq m\leq l\},
    \label{eq:supp_tmatrix_retained_labels}
    \\
    U_3
     & =
    30.
    \label{eq:supp_tmatrix_retained_dimension}
\end{align}
The regular and outgoing VSWFs $\bm\Phi_u^{\mathrm{reg}}$ and $\bm\Phi_u^{\mathrm{out}}$ use exactly the normalization and ordering fixed in Section~III-A of the main paper. Thus the extracted matrix is intended to satisfy
\begin{equation}
    \mathbf a_D^{\mathrm{sca}}
    =
    \mathbf T_{D,3}\mathbf a_D^{\mathrm{inc}},
    \quad
    \mathbf T_{D,3}\in\mathbb C^{30\times30},
    \label{eq:supp_extracted_tmatrix_action}
\end{equation}
where both coefficient vectors are in the main-paper VSWF convention.

\subsection{Batched SCUFF-EM field solves}

SCUFF-EM discretizes the isolated penetrable-body surface integral equation with RWG unknowns on the supplied triangular mesh \cite{reid2015scuff}. At the fixed body, material, and frequency, let $\mathbf Z_D$ be the assembled BEM matrix. A regular SCUFF spherical wave is applied for every retained $(l,m,P)$, where $P\in\{0,1\}$ is the SCUFF polarization label. The zero-based SCUFF column index is
\begin{equation}
    \iota_{\mathrm S}(l,m,P)
    =
    2\bigl[l(l+1)+m-1\bigr]+P.
    \label{eq:supp_scuff_modal_index}
\end{equation}
All $30$ right-hand sides are assembled into one matrix $\mathbf B_D^{\mathrm{BEM}}$, after which a single LU factorization is reused:
\begin{equation}
    \mathbf Z_D\mathbf U_D^{\mathrm{BEM}}
    =
    \mathbf B_D^{\mathrm{BEM}}.
    \label{eq:supp_batched_scuff_solve}
\end{equation}
The resulting surface-current columns are then evaluated at the same exterior points. This organization changes only the cost of the extraction: it is algebraically the collection of the $30$ isolated-body scattering solves.

The fitting points form a deterministic Fibonacci grid on the sphere of radius $3R_D^{\mathrm{enc}}$. With $M=120$ and $j=0,\ldots,M-1$, set
\begin{align}
    z_j
     & =
    1-\frac{2(j+1/2)}{M},
    \quad
    \varphi_j
    =
    \pi(1+\sqrt5)j,
    \label{eq:supp_fibonacci_parameters} \\
    \mathbf r_j
     & =
    3R_D^{\mathrm{enc}}
    \left(
    \sqrt{1-z_j^2}\cos\varphi_j,
    \sqrt{1-z_j^2}\sin\varphi_j,
    z_j
    \right).
    \label{eq:supp_fibonacci_points}
\end{align}
Because every $\mathbf r_j$ lies strictly outside the enclosing sphere, each scattered field has a convergent single-center outgoing expansion there, including for the multiply connected torus. Only the sampled electric field enters the fit below; the independently sampled magnetic field is reserved for the replay test.

\subsection{Outgoing-field fit and convention conversion}

Define the outgoing electric-field design matrix $\mathbf V_D\in\mathbb C^{3M\times U_3}$ and the SCUFF sampled-field matrix $\mathbf F_D\in\mathbb C^{3M\times U_3}$ componentwise by
\begin{align}
    [\mathbf V_D]_{3j+c,u}
     & =
    [\bm\Phi_u^{\mathrm{out}}(\mathbf r_j)]_c,
    \label{eq:supp_outgoing_fit_design_matrix} \\
    [\mathbf F_D]_{3j+c,h+1}
     & =
    [\mathbf E_{D,h}^{\mathrm{sca,SCUFF}}(\mathbf r_j)]_c,
    \label{eq:supp_scuff_sampled_field_matrix}
\end{align}
where $j=0,\ldots,M-1$, $c=1,2,3$, $u=1,\ldots,U_3$, and $h=0,\ldots,U_3-1$. All SCUFF excitation columns are fitted simultaneously against the common design matrix:
\begin{equation}
    \widetilde{\mathbf T}_D
    =
    \underset{\mathbf X\in\mathbb C^{U_3\times U_3}}{\operatorname{argmin}}
    \|\mathbf V_D\mathbf X-\mathbf F_D\|_{\mathrm F}.
    \label{eq:supp_joint_outgoing_field_fit}
\end{equation}
For each body, the $360\times30$ design matrix has full column rank; its condition number is reported in Table~\ref{tab:supp_tmatrix_fit_diagnostics}.

It remains to convert the incident-column convention. Direct source-level comparison of the SCUFF and main-paper regular and outgoing Maxwell waves gives
\begin{equation}
    \bm\Phi_{\iota_{\mathrm S}(l,m,\tau-1)}^{s,\mathrm{SCUFF}}
    =
    \gamma_{\tau m}\bm\Phi_{\xi_3(\tau,l,m)}^{s},
    \quad
    s\in\{\mathrm{reg},\mathrm{out}\},
    \label{eq:supp_scuff_msss_vswf_conversion}
\end{equation}
where the polarization labels are identified by $P=\tau-1$, and
\begin{equation}
    \gamma_{\tau m}
    =
    s_m
    \begin{cases}
        \mathrm j/\sqrt\pi, & \tau=1, \\
        -1/\sqrt\pi,        & \tau=2,
    \end{cases}.
    \label{eq:supp_scuff_msss_mode_factor}
\end{equation}
Here
\begin{equation}
    s_m
    =
    \begin{cases}
        1,      & m\leq0, \\
        (-1)^m, & m>0.
    \end{cases}.
    \label{eq:supp_scuff_positive_m_sign}
\end{equation}
The rows of $\widetilde{\mathbf T}_D$ are already coefficients of the main-paper outgoing basis because that basis was used in \eqref{eq:supp_outgoing_fit_design_matrix}. Only the incident columns require rescaling. Hence, for every retained $(\tau,l,m)$,
\begin{equation}
    [\mathbf T_{D,3}]_{:,\xi_3(\tau,l,m)}
    =
    \gamma_{\tau m}^{-1}
    [\widetilde{\mathbf T}_D]_{:,\iota_{\mathrm S}(l,m,\tau-1)+1}.
    \label{eq:supp_extracted_tmatrix_column_conversion}
\end{equation}
Equation~\eqref{eq:supp_extracted_tmatrix_column_conversion} fixes the ordering, polarization identification, and phase normalization of the matrices used by the multiple-scattering solver.

\subsection{Fit diagnostics and held-out replay}

The overall and columnwise electric-field fit residuals are
\begin{align}
    r_{\mathrm F,D}
     & :=
    \frac{\|\mathbf V_D\widetilde{\mathbf T}_D-\mathbf F_D\|_{\mathrm F}}
    {\|\mathbf F_D\|_{\mathrm F}},
    \label{eq:supp_tmatrix_global_fit_residual} \\
    r_{D,h}
     & :=
    \frac{\|[\mathbf V_D\widetilde{\mathbf T}_D-\mathbf F_D]_{:,h+1}\|_2}
    {\|[\mathbf F_D]_{:,h+1}\|_2},
    \quad
    h=0,\ldots,U_3-1.
    \label{eq:supp_tmatrix_column_fit_residual}
\end{align}
Table~\ref{tab:supp_tmatrix_fit_diagnostics} reports $r_{\mathrm F,D}$ together with the arithmetic mean and maximum of $r_{D,h}$. The global residual weights columns by their scattered-field magnitudes, whereas the mean and maximum normalize every column separately and are therefore more sensitive to weak-response columns.

\begin{table}[!t]
    \caption{Electric-field fit and reciprocity diagnostics}
    \label{tab:supp_tmatrix_fit_diagnostics}
    \centering
    \setlength{\tabcolsep}{2.8pt}
    \begin{tabular}{c c c c c c c}
        \hline
        Body      & Rank & $\kappa_2(\mathbf V_D)$ & $r_{\mathrm F,D}$ & Mean $r_{D,h}$ & Max. $r_{D,h}$ & $\epsilon_{\mathrm{rec},D}$ \\
                  &      &                         & (\%)              & (\%)           & (\%)           &                             \\
        \hline
        Cube      & 30   & 1.925                   & 0.124             & 1.288          & 5.279          & $9.74\times10^{-6}$         \\
        Ellipsoid & 30   & 1.925                   & 0.187             & 1.246          & 4.145          & $1.34\times10^{-5}$         \\
        Torus     & 30   & 1.926                   & 0.452             & 5.163          & 42.190         & $3.65\times10^{-5}$         \\
        \hline
    \end{tabular}
\end{table}

Two algebraic checks are applied after the convention conversion. Let $\mathbf R_3$ reverse $m$ within each fixed $(\tau,l)$ block. The twisted-reciprocity residual in the last column of Table~\ref{tab:supp_tmatrix_fit_diagnostics} is
\begin{equation}
    \epsilon_{\mathrm{rec},D}
    :=
    \frac{\|\mathbf T_{D,3}-\mathbf R_3\mathbf T_{D,3}^{\mathsf T}\mathbf R_3\|_{\mathrm F}}
    {\|\mathbf T_{D,3}\|_{\mathrm F}}.
    \label{eq:supp_tmatrix_reciprocity_residual}
\end{equation}
Passivity is checked through the Hermitian matrix
\begin{equation}
    \mathbf P_D^{\mathrm{pass}}
    :=
    -\frac12
    \left(\mathbf T_{D,3}+\mathbf T_{D,3}^{\dagger}\right)
    -
    \mathbf T_{D,3}^{\dagger}\mathbf T_{D,3}.
    \label{eq:supp_tmatrix_passivity_matrix}
\end{equation}
For an exact lossless full T-operator, the corresponding passive form vanishes. For its retained principal block, coupling from retained incident modes into omitted outgoing modes instead contributes a positive semidefinite remainder. The cube and ellipsoid eigenvalue ranges are $[-6.94\times10^{-7},5.09\times10^{-7}]$ and $[-1.27\times10^{-6},1.21\times10^{-6}]$, respectively; their small negative excursions are attributable to the BEM and fitting errors, while the positive scale is also compatible with modal omission. For the lossy torus, the range is $[6.25\times10^{-8},9.15\times10^{-3}]$, consistent with a positive semidefinite passive form at the reported precision.

The replay uses a TE plane wave with polar angles $(\theta,\phi)=(0.75,0.4)$ and unit electric amplitude. In the implemented convention, its propagation and electric-polarization unit vectors are
\begin{equation}
    \widehat{\mathbf k}
    =
    (\sin\theta\cos\phi,\sin\theta\sin\phi,\cos\theta),
    \quad
    \widehat{\mathbf e}
    =
    (-\sin\phi,\cos\phi,0).
    \label{eq:supp_replay_plane_wave_vectors}
\end{equation}
The incident electric field supplied to SCUFF-EM is
$\mathbf E^{\mathrm{inc}}(\mathbf r)
    =\widehat{\mathbf e}\exp(\mathrm j k_0\widehat{\mathbf k}\cdot\mathbf r)$.
This excitation is not one of the spherical-wave fitting columns, and its six evaluation points are distinct from the Fibonacci grid. Define the set of unnormalized direction vectors by
\begin{equation}
    \mathcal D^{\mathrm{rep}}
    =
    \left\{
    \begin{array}{@{}c@{\;}c@{}}
        (+0.91,+0.20,+0.36), & (-0.82,+0.28,+0.50), \\
        (+0.24,-0.91,+0.34), & (+0.18,+0.35,+0.92), \\
        (-0.62,-0.58,+0.53), & (+0.50,-0.43,-0.75)
    \end{array}
    \right\},
    \label{eq:supp_replay_directions}
\end{equation}
With $(\rho_1,\ldots,\rho_6)=(1.55,1.68,1.82,1.95,1.74,1.88)$, the replay points are
\begin{equation}
    \mathbf r_q^{\mathrm{rep}}
    =
    \rho_q R_D^{\mathrm{enc}}
    \frac{\mathbf v_q}{\|\mathbf v_q\|_2},
    \quad
    \mathbf v_q\in\mathcal D^{\mathrm{rep}}.
    \label{eq:supp_replay_points}
\end{equation}
SCUFF-EM directly solves this plane-wave problem on the same body mesh. Separately, the analytic incident-coefficient formula is truncated at order $3$ and multiplied by $\mathbf T_{D,3}$. If the result is denoted by $\mathbf a_D^{\mathrm{sca}}$, the replay fields are synthesized as
\begin{align}
    \mathbf E_{D,3}^{\mathrm{sca}}(\mathbf r)
     & =
    \sum_{u=1}^{U_3}
    [\mathbf a_D^{\mathrm{sca}}]_u
    \bm\Phi_u^{\mathrm{out}}(\mathbf r),
    \label{eq:supp_replay_electric_synthesis} \\
    \mathbf H_{D,3}^{\mathrm{sca}}(\mathbf r)
     & =
    -\frac{\mathrm j}{\eta_0}
    \sum_{u=1}^{U_3}
    [\mathbf a_D^{\mathrm{sca}}]_u
    \bm\Phi_{u^\star}^{\mathrm{out}}(\mathbf r).
    \label{eq:supp_replay_magnetic_synthesis}
\end{align}
For $F\in\{E,H\}$, the reported replay error is
\begin{equation}
    \epsilon_{F,D}^{\mathrm{rep}}
    :=
    \frac{
    \left(
    \sum_{q=1}^{6}
    \|\mathbf F_{D,3}^{\mathrm{sca}}(\mathbf r_q^{\mathrm{rep}})
    -\mathbf F_D^{\mathrm{sca,SCUFF}}(\mathbf r_q^{\mathrm{rep}})\|_2^2
    \right)^{1/2}}
    {
    \left(
    \sum_{q=1}^{6}
    \|\mathbf F_D^{\mathrm{sca,SCUFF}}(\mathbf r_q^{\mathrm{rep}})\|_2^2
    \right)^{1/2}}.
    \label{eq:supp_plane_wave_replay_error}
\end{equation}
The resulting electric- and magnetic-field replay errors are those reported in Table~I of the main paper. In particular, the magnetic field is excluded from \eqref{eq:supp_joint_outgoing_field_fit}; its replay error therefore checks, outside the fit, the outgoing-field curl relation used in \eqref{eq:supp_replay_magnetic_synthesis}.

The single-run wall times in Table~I were measured on an Intel Xeon Platinum 8268 CPU at $2.90\,\mathrm{GHz}$ with approximately $90\,\mathrm{GiB}$ of memory. They include the batched spherical-wave BEM solve, exterior-field evaluation, least-squares fit, and plane-wave replay. No GPU acceleration was used.

The residuals above have a deliberately limited interpretation. The fitting residual tests how well the SCUFF BEM fields for the retained spherical-wave excitations lie in the retained outgoing VSWF space. The replay additionally tests the incident-column conversion and the action of the extracted matrix on a new excitation. Because both tests use the same supplied body mesh, they do not constitute a BEM mesh-refinement study. The replay also remains a prescribed-excitation test, not a supremum over the continuous source unit ball. Accordingly, the extracted matrices carry fixed BEM, fitting, and order-$3$ errors into the full-wave implementation tests, while the operator-norm convergence result of the main paper concerns the exact transition operators and the limit $L\to\infty$.

\section{COMSOL Full-Wave Reference and Benchmark Protocol}
\label{supp:comsol_fullwave_protocol}

This note specifies the full-wave reference and benchmark protocol underlying Section~IV-A of the main paper. The offline extraction of the nonspherical T-matrices is given in \hyperref[supp:nonspherical_tmatrix_extraction]{Supplementary Note~\ref*{supp:nonspherical_tmatrix_extraction}}.

\subsection{Coordinates, computational domain, mesh, and field sampling}

The manuscript coordinates are denoted by $(x,y,z)$. The implementation uses the fixed coordinate relabeling
\begin{equation}
    \mathbf r_{\mathrm c}
    =
    \begin{bmatrix}
        1 & 0 & 0  \\
        0 & 0 & -1 \\
        0 & 1 & 0
    \end{bmatrix}
    \mathbf r,
    \label{eq:supp_fullwave_coordinate_map}
\end{equation}
so manuscript broadside $+y$ is computational $+z$. All dimensions below are stated in manuscript coordinates. At $f=1\,\mathrm{GHz}$, the non-PML air domain is
\begin{equation}
    [-1.75,1.75]\lambda
    \times
    [-0.45,4.30]\lambda
    \times
    [-1.75,1.75]\lambda.
    \label{eq:supp_fullwave_physical_domain}
\end{equation}
A Cartesian PML of thickness $0.5\lambda$ surrounds all six faces. Its stretching is polynomial, and the typical wavelength supplied to the PML is $\lambda$.

The production models use the Electromagnetic Waves, Frequency Domain interface in COMSOL Multiphysics 6.4.0.293. The physical domain is meshed with free tetrahedra at physics-controlled mesh level 4, corresponding to \emph{Fine}; the PML is meshed by four swept layers. Third-order electric-field edge elements are used. The point and Gaussian models use this mesh without an additional local size constraint. The isolated physical patch-array solve adds $h_{\max}=\lambda/18.7$ on the antenna boundaries and $h_{\max}=\lambda/152$ on the four lumped-port boundaries. The frozen-Huygens COMSOL models instead impose $h_{\max}=\lambda/12$ on the source-box boundary. Each frequency-domain problem is fully coupled and uses the default direct solver selected by COMSOL.

The observation plane is not inserted as a conformal geometric boundary. Both solvers are sampled at the same $21\times21$ Cartesian grid
\begin{align}
    \mathbf r_{uv}
     & =
    (x_u,4.05\lambda,z_v),
    \label{eq:supp_fullwave_observation_grid} \\
    x_u
     & =
    -1.5\lambda+0.15u\lambda,
    \quad
    u=0,\ldots,20.
    \label{eq:supp_fullwave_observation_x}    \\
    z_v
     & =
    -1.5\lambda+0.15v\lambda,
    \quad
    v=0,\ldots,20.
    \label{eq:supp_fullwave_observation_z}
\end{align}
COMSOL's \texttt{mphinterp} evaluates the three complex components of both $\mathbf E$ and $\mathbf H$ at these 441 points. The solver outputs total fields. For every source, the scattered field used in the comparison is formed by subtracting a source-matched empty-scene solve from the loaded-scene solve. The amplitude maps in Fig.~3(b) of the main paper interpolate the original $21\times21$ vector samples to a $41\times41$ display grid, whereas all reported errors use the original samples.

\subsection{Time convention and discrete comparison metrics}

The main paper uses $\exp(-\mathrm j\omega t)$, whereas the COMSOL frequency-domain interface uses $\exp(+\mathrm j\omega t)$. If $\widetilde{\mathbf F}_{\mathrm C}$ denotes a raw COMSOL phasor, the fixed field mapping used in all comparisons is
\begin{equation}
    \mathbf F_{\mathrm C}
    :=
    \overline{\widetilde{\mathbf F}_{\mathrm C}},
    \quad
    \mathbf F\in\{\mathbf E,\mathbf H\}.
    \label{eq:supp_comsol_field_conjugation}
\end{equation}
Source phasors passed from the manuscript convention to the point and Gaussian COMSOL models are conjugated before the solve. For the frozen-Huygens comparison, the currents reimposed in COMSOL remain in the raw COMSOL convention, while the currents passed to the proposed solver are obtained from the conjugated empty-array fields. This single convention map is fixed by the corresponding empty-scene comparison and is applied unchanged to every loaded scene.

The same convention conversion is applied to constitutive parameters. If $\varepsilon_{\mathrm{P}}$ and $\mu_{\mathrm{P}}$ denote the relative parameters in the manuscript convention, COMSOL receives
\begin{equation}
    \varepsilon_{\mathrm{C}}
    :=
    \overline{\varepsilon_{\mathrm{P}}},
    \quad
    \mu_{\mathrm{C}}
    :=
    \overline{\mu_{\mathrm{P}}}.
    \label{eq:supp_comsol_constitutive_conjugation}
\end{equation}
Thus the torus value $6+0.5\mathrm j$ in the manuscript convention is represented by $6-0.5\mathrm j$ in COMSOL; the real-valued cube, ellipsoid, vacuum, and PEC parameters are invariant under this conversion.

For a sampled vector field $\mathbf F$, define
\begin{equation}
    \|\mathbf F\|_{2,\mathrm{samp}}
    :=
    \left(
    \sum_{u=0}^{20}\sum_{v=0}^{20}
    \sum_{c\in\{x,y,z\}}
    |F_c(\mathbf r_{uv})|^2
    \right)^{1/2}.
    \label{eq:supp_fullwave_discrete_field_norm}
\end{equation}
This is the ordinary Euclidean norm of the stacked complex samples. No observation-plane quadrature weights are applied. Because the grid is common and uniform, a common constant area factor would cancel from the relative quantities below, but the reported values are, strictly, discrete sample norms rather than continuous surface norms.

Let the subscripts $\mathrm P$ and $\mathrm C$ denote the proposed and convention-mapped COMSOL fields, respectively, and let superscripts $0$, $\mathrm{tot}$, and $\mathrm{sca}$ denote empty, loaded total, and loaded-minus-empty fields. Thus
\begin{align}
    \mathbf F_{\mathrm P}^{\mathrm{sca}}
     & :=
    \mathbf F_{\mathrm P}^{\mathrm{tot}}-\mathbf F_{\mathrm P}^{0},
    \label{eq:supp_fullwave_proposed_scattered_field} \\
    \mathbf F_{\mathrm C}^{\mathrm{sca}}
     & :=
    \mathbf F_{\mathrm C}^{\mathrm{tot}}-\mathbf F_{\mathrm C}^{0}.
    \label{eq:supp_fullwave_comsol_scattered_field}
\end{align}
The two error curves and the scattering-strength curve are
\begin{align}
    \epsilon_F^{\mathrm{tot}}
     & :=
    \frac{\|\mathbf F_{\mathrm P}^{\mathrm{tot}}-\mathbf F_{\mathrm C}^{\mathrm{tot}}\|_{2,\mathrm{samp}}}
    {\|\mathbf F_{\mathrm C}^{\mathrm{tot}}\|_{2,\mathrm{samp}}},
    \label{eq:supp_fullwave_total_error}     \\
    \epsilon_F^{\mathrm{sca}}
     & :=
    \frac{\|\mathbf F_{\mathrm P}^{\mathrm{sca}}-\mathbf F_{\mathrm C}^{\mathrm{sca}}\|_{2,\mathrm{samp}}}
    {\|\mathbf F_{\mathrm C}^{\mathrm{sca}}\|_{2,\mathrm{samp}}},
    \label{eq:supp_fullwave_scattered_error} \\
    s_F
     & :=
    \frac{\|\mathbf F_{\mathrm P}^{\mathrm{sca}}\|_{2,\mathrm{samp}}}
    {\|\mathbf F_{\mathrm P}^{0}\|_{2,\mathrm{samp}}}.
    \label{eq:supp_fullwave_scattered_strength}
\end{align}
Fig.~3(a) of the main paper uses $F=E$. The empty-scene electric errors are $2.2814\%$, $1.5127\%$, and $1.5303\%$ for the point, Gaussian, and frozen-Huygens sources, respectively. Across the loaded nested scenes, the corresponding maxima of $\epsilon_E^{\mathrm{tot}}$ are $2.2715\%$, $1.6232\%$, and $1.5271\%$. Their offsets from the source-specific empty baselines are $-0.0099$, $+0.1105$, and $-0.0032$ percentage points, respectively; thus only the Gaussian maximum exceeds its empty-scene baseline.

\subsection{Prescribed sources and the frozen-source protocol}

The point-current case uses
\begin{equation}
    \mathbf J_p(\mathbf r)
    =
    I\ell\,\hat{\mathbf x}\,\delta(\mathbf r),
    \quad
    I\ell=10^{-3}\,\mathrm{A\,m},
    \quad
    \mathbf M_p=\mathbf0.
    \label{eq:supp_fullwave_point_source}
\end{equation}
COMSOL realizes it as an Electric Point Dipole at the origin; the proposed solver represents the same distribution by one electric-current entry of moment $10^{-3}\,\mathrm{A\,m}$.

On the square $y=0$, $|x|,|z|\leq0.5\lambda$, the Gaussian source is
\begin{align}
    \mathbf J_G(x,0,z)
     & =
    J_0\exp\!\left[-\frac{x^2+z^2}{w^2}\right]\hat{\mathbf x},
    \label{eq:supp_fullwave_gaussian_electric_current} \\
    \mathbf M_G(x,0,z)
     & =
    \eta_0J_0\exp(\mathrm j\pi/3)
    \exp\!\left[-\frac{x^2+z^2}{w^2}\right]\hat{\mathbf z},
    \label{eq:supp_fullwave_gaussian_magnetic_current}
\end{align}
where $J_0=1\,\mathrm{A/m}$ and $w=0.25\lambda$. COMSOL applies these analytic densities as one-sided electric and magnetic surface-current conditions on an internal work plane. The proposed solver evaluates the same densities with an $18\times18$ tensor-product Gauss--Legendre rule on the square, giving 324 co-located electric and magnetic quadrature samples. Thus the continuous source is common, whereas its discretization is solver-specific; the Gaussian empty-scene error includes this interface difference.

The Huygens source begins with the COMSOL 6.4 RF Module Application Library model \nolinkurl{microstrip_patch_antenna_inset.mph} from the Antenna Arrays collection. Its substrate has $\varepsilon_r=3.38$ and is lossless. All antenna dimensions in that model are multiplied by $1.575$: the substrate thickness is based on $60\,\mathrm{mil}$, and the unscaled line width, patch width, patch length, stub width, stub length, and square-substrate side are $3.2$, $53$, $52$, $7$, $15.5$, and $100\,\mathrm{mm}$, respectively. Four copies are placed on a $2\times2$ square lattice of pitch $0.64\lambda$. Each lumped port is driven by $1\,\mathrm V$ with equal phase, giving broadside radiation along manuscript $+y$.

One empty-environment physical-array solve is sampled on the boundary of
\begin{equation}
    \Omega_{\mathrm{src}}
    =
    [-0.7,0.7]\lambda
    \times
    [-0.2,0.2]\lambda
    \times
    [-0.7,0.7]\lambda.
    \label{eq:supp_fullwave_huygens_box}
\end{equation}
Each face uses a $24\times24$ tensor-product Gauss--Legendre rule, for $6(24)^2=3456$ trace samples. After applying \eqref{eq:supp_comsol_field_conjugation}, the proposed-solver currents are
\begin{equation}
    \mathbf J_H
    =
    \widehat{\mathbf n}\times\mathbf H^0,
    \quad
    \mathbf M_H
    =
    -\widehat{\mathbf n}\times\mathbf E^0.
    \label{eq:supp_fullwave_huygens_currents}
\end{equation}
The proposed solver multiplies these current samples by their face-quadrature weights. For the independent COMSOL comparison, the raw-COMSOL-convention values of the same currents are tabulated separately on each face and reimposed through two-coordinate linear interpolation. The same current tables and weighted point data are reused without modification for every $N$. A separate frozen-Huygens empty solve supplies $\mathbf F_{\mathrm C}^0$, and each loaded solve changes only the scatterer set. The physical patch, its lumped ports, and any load-induced change of antenna current are absent from all loaded cases plotted in the main paper.

\subsection{Nested scene construction, clearance, and sample selection}

The candidate geometry is generated once with NumPy random seed $20260713$. In the cube centered at $(0,2.4\lambda,0)$, the local coordinate levels are $(-0.975,-0.325,0.325,0.975)\lambda$. All eight corner sites and 42 sites drawn without replacement from the remaining 56 lattice sites are retained. Each coordinate then receives an independent uniform jitter in $[-0.055,0.055]\lambda$. A shuffled label pool assigns 13 spheres, 13 cubes, 12 ellipsoids, and 12 tori to these 50 candidate IDs. Nonspherical orientations use active ZYZ Euler angles with $\alpha$ and $\gamma$ uniform on $[0,2\pi)$ and $\cos\beta$ uniform on $[-1,1]$.

The published scenes are hand-selected nested subsets of this fixed candidate list. Table~\ref{tab:supp_fullwave_nested_scenes} gives the IDs added at each step, the resulting shape counts, and the measured minimum enclosing-sphere gap. The five-body row gives the initial IDs; each subsequent row lists only the five new IDs. Hence every smaller scene is a strict subset of every larger scene.

\begin{table}[!t]
    \caption{Deterministic nested-scene construction and enclosing-sphere clearance}
    \label{tab:supp_fullwave_nested_scenes}
    \centering
    \footnotesize
    \setlength{\tabcolsep}{2pt}
    \begin{tabular}{c c c c c c c}
        \hline
        $N$ & \shortstack{Initial or added                                \\candidate IDs} & Sphere & Cube & Ellipsoid & Torus & $g_N/\lambda$ \\
        \hline
        5   & 01, 02, 04, 05, 07           & 2  & 1  & 1  & 1  & 1.624791 \\
        10  & 14, 18, 23, 37, 43           & 3  & 3  & 2  & 2  & 0.268077 \\
        15  & 06, 08, 09, 11, 12           & 4  & 4  & 3  & 4  & 0.246676 \\
        20  & 10, 13, 16, 19, 21           & 5  & 5  & 5  & 5  & 0.246676 \\
        25  & 15, 22, 26, 27, 28           & 7  & 6  & 6  & 6  & 0.246676 \\
        30  & 24, 32, 34, 35, 40           & 8  & 8  & 7  & 7  & 0.246676 \\
        35  & 17, 20, 25, 29, 31           & 9  & 9  & 9  & 8  & 0.246676 \\
        40  & 30, 33, 38, 41, 46           & 10 & 10 & 10 & 10 & 0.246676 \\
        \hline
    \end{tabular}
\end{table}

For a selected $N$-body set, the enforced quantity is
\begin{equation}
    g_N
    :=
    \min_{1\leq i<j\leq N}
    \left(
    \|\mathbf c_i-\mathbf c_j\|_2
    -R_i^{\mathrm{enc}}-R_j^{\mathrm{enc}}
    \right).
    \label{eq:supp_fullwave_minimum_enclosing_gap}
\end{equation}
Here $R_i^{\mathrm{enc}}=\lambda/(2\pi)$ for a PEC sphere and is the maximum nodal radius of the corresponding local mesh for a nonspherical body. The scene builder terminates without writing a case if $g_N<0.2\lambda$. This check is applied after the actual body IDs have been selected, rather than inferred from the nominal lattice spacing.

Torus candidate ID 3 could not be tetrahedrally meshed in COMSOL and was replaced by candidate ID 46 before the nested production sequence was fixed. This is a meshability exclusion rather than a field-based sample rejection; the resulting sequence is one deterministic benchmark and is not used to estimate ensemble statistics.

\subsection{Timing boundaries and execution environment}

The proposed solver's wall-clock timer excludes input loading and construction of the scatterer and excitation descriptors. It includes incident-field evaluation on the observation grid, translation-table preparation, rotated local-T-block preparation, source projection, the collective solve, and total-field reconstruction. The Table~II(b) groups are
\begin{align}
    t_{\mathrm{set}}
     & :=
    t_{\mathrm{translation}}+t_{\mathrm{T\mbox{-}block}},
    \label{eq:supp_fullwave_setup_time}       \\
    t_{\mathcal O}
     & :=
    t_{\mathrm{incident}}+t_{\mathrm{field}},
    \label{eq:supp_fullwave_observation_time} \\
    t_{\mathrm{tot}}
     & :=
    t_{\mathrm{set}}+t_{\mathcal S}+t_{\mathrm{sol}}+t_{\mathcal O}.
    \label{eq:supp_fullwave_total_online_time}
\end{align}
Here $t_{\mathcal S}$ is the projection of the prescribed electric and magnetic sources onto all local regular modes, and $t_{\mathrm{sol}}$ is the collective GMRES solve.

The COMSOL timer begins immediately before the single call to the frequency-domain study and ends when that call returns. It therefore includes finite-element assembly and the direct solution, but excludes model creation, STL import, geometry construction, mesh generation, solver-sequence construction, field interpolation, file export, and MPH-file saving. In contrast, the proposed timer includes its observation-field reconstruction. The comparison is consequently a repeated-query online benchmark with each solver's geometry/source data already prepared, not a comparison of complete first-use workflows.

The local SCUFF-EM T-matrix extractions reported in Table~I and the isolated physical-array solve used to obtain \eqref{eq:supp_fullwave_huygens_currents} are offline costs and are excluded from every speedup in Table~II. The frozen-Huygens empty COMSOL solve is likewise excluded from the per-loaded-scene COMSOL time. The Table~I extraction hardware is documented in \hyperref[supp:nonspherical_tmatrix_extraction]{Supplementary Note~\ref*{supp:nonspherical_tmatrix_extraction}} and differs from the workstation used here.

Every Table~II entry is one wall-clock measurement; no repetitions, averaging, best-of-run selection, or warm-up correction is applied. The software environment was COMSOL Multiphysics 6.4.0.293 through LiveLink for MATLAB R2025b Update 5, together with Python 3.12.9, NumPy 2.2.6, and SciPy 1.15.3 on 64-bit Windows. The workstation has an Intel Core i7-14700 processor with 20 physical cores, 28 logical processors, and approximately $31.8\,\mathrm{GiB}$ of memory.

\section{Finite-Reference Operator Diagnostics}
\label{supp:finite_reference_diagnostics}

This note derives the finite-matrix diagnostics used in Section~IV-B of the main paper. It first specifies the semianalytic Gramians for the concentric closed source and observation spheres, then derives the finite-reference error, the source and observation tails, and the trace-normalized resolvent amplification.

\subsection{Semianalytic Gramians on the external spheres}

Let $L_\star$ be the retained reference order and let the cluster origin be the common center of
\begin{equation}
    \Sigma_S=\partial B_{R_S}(\mathbf0),
    \quad
    \Sigma_O=\partial B_{R_O}(\mathbf0).
    \label{eq:supp_diagnostic_external_spheres}
\end{equation}
Both spheres enclose every local trace sphere. The source carries tangential electric and magnetic currents with the metric in \eqref{eq:supp_surface_source_metric}, and the observation carries the tangential electric field with the metric in \eqref{eq:supp_surface_observation_metric}.

For a reference cluster mode $p=\zeta_{L_\star}(i,u)$, recall $\bm\Psi_p(\mathbf r)=\bm\Phi_u^{\mathrm{out}}(\mathbf r-\mathbf c_i)$. Translate every such local outgoing wave to outgoing VSWFs about the cluster origin and truncate the common-center expansion at $L_c$:
\begin{equation}
    \bm\Psi_p(\mathbf r)
    =
    \sum_{\nu\in\Lambda_{L_c}}
    [\mathbf A_c]_{\nu p}
    \bm\Phi_\nu^{\mathrm{out}}(\mathbf r)
    +
    \bm\epsilon_{p,L_c}(\mathbf r).
    \label{eq:supp_common_center_outgoing_translation}
\end{equation}
The same translation matrix $\mathbf A_c\in\mathbb C^{U_{L_c}\times N_{L_\star}}$ applies on both external spheres; only their radial weights differ. Define
\begin{equation}
    b_{\tau l}(R)
    :=
    \pi R^2
    \begin{cases}
        |z_l^{\mathrm{out}}(kR)|^2, & \tau=1, \\
        |d_l^{\mathrm{out}}(kR)|^2, & \tau=2,
    \end{cases}
    \label{eq:supp_common_center_tangential_weights}
\end{equation}
where $d_l^{\mathrm{out}}$ is defined in \eqref{eq:supp_radial_derivative_factor}, and let $\mathbf D_c(R)$ be diagonal with entry $b_{\tau l}(R)$ for every $\nu=(\tau,l,m)\in\Lambda_{L_c}$. The VSH orthogonality established in \hyperref[supp:vswf_trace_metric]{Supplementary Note~\ref*{supp:vswf_trace_metric}} gives
\begin{equation}
    \int_{\partial B_R}
    (\bm\Phi_{\nu,t}^{\mathrm{out}})^{\dagger}
    \bm\Phi_{\nu',t}^{\mathrm{out}}\,\mathrm dS
    =
    \delta_{\nu\nu'}b_{\tau l}(R).
    \label{eq:supp_common_center_tangential_orthogonality}
\end{equation}

Let $\mathbf A_J$ and $\mathbf A_M$ collect the translated columns required by the two source pairings:
\begin{equation}
    [\mathbf A_J]_{:,p}
    :=
    [\mathbf A_c]_{:,\overline p},
    \quad
    [\mathbf A_M]_{:,p}
    :=
    [\mathbf A_c]_{:,\overline{p^\star}}.
    \label{eq:supp_common_center_source_column_maps}
\end{equation}
Substitution into the surface formulas derived in \hyperref[supp:surface_interface_grams]{Supplementary Note~\ref*{supp:surface_interface_grams}} yields
\begin{align}
    \mathbf C_{O,L_\star}^{(L_c)}
     & =
    \frac{1}{2\eta}
    \mathbf A_c^{\dagger}
    \mathbf D_c(R_O)
    \mathbf A_c,
    \label{eq:supp_semianalytic_observation_gram} \\
    \mathbf C_{S,L_\star}^{(L_c)}
     & =
    \frac{2\eta k^4}{\pi^2}
    \left[
    \mathbf A_J^{\mathsf T}
    \mathbf D_c(R_S)
    \mathbf A_J^{*}
    +
    \mathbf A_M^{\mathsf T}
    \mathbf D_c(R_S)
    \mathbf A_M^{*}
    \right].
    \label{eq:supp_semianalytic_source_gram}
\end{align}
Thus the continuous surface integrals are reduced to translation coefficients and diagonal radial factors; no surface quadrature degrees of freedom enter the operator-norm calculation. All Section~IV-B quantities involving the external-surface Gramians use $L_c=44$. In the fixed-surface $L_\star=18$ closure check, increasing $L_c$ from $44$ to $48$ changed the reported reference norm and relative errors by at most $1.3\times10^{-8}$ relatively.

\subsection{Nested finite-reference error}

For $L\leq L_\star$, define the isometric coordinate injection $\mathbf Z_{L\mid L_\star}\in\mathbb C^{N_{L_\star}\times N_L}$ by
\begin{equation}
    [\mathbf Z_{L\mid L_\star}]_{(i,\alpha),(j,\beta)}
    :=
    \delta_{ij}\delta_{\alpha\beta},
    \quad
    \alpha\in\Lambda_{L_\star},
    \quad
    \beta\in\Lambda_L.
    \label{eq:supp_nested_coordinate_injection}
\end{equation}
It satisfies
\begin{equation}
    \mathbf Z_{L\mid L_\star}^{\dagger}
    \mathbf Z_{L\mid L_\star}
    =
    \mathbf I,
    \quad
    \mathbf E_{L\mid L_\star}
    :=
    \mathbf Z_{L\mid L_\star}
    \mathbf Z_{L\mid L_\star}^{\dagger},
    \label{eq:supp_nested_coordinate_projector}
\end{equation}
where $\mathbf E_{L\mid L_\star}$ selects the modes with angular order at most $L$ in the reference coordinates. Because the raw modal functions and their fixed trace weights do not depend on the ambient truncation order,
\begin{equation}
    \mathcal S_L
    =
    \mathbf Z_{L\mid L_\star}^{\dagger}\mathcal S_{L_\star},
    \quad
    \mathcal O_L
    =
    \mathcal O_{L_\star}\mathbf Z_{L\mid L_\star}.
    \label{eq:supp_nested_external_maps}
\end{equation}
Consequently, with
\begin{equation}
    \mathbf D_{L\mid L_\star}
    :=
    \mathbf K_{L_\star}
    -
    \mathbf Z_{L\mid L_\star}
    \mathbf K_L
    \mathbf Z_{L\mid L_\star}^{\dagger},
    \label{eq:supp_nested_core_difference}
\end{equation}
the two retained operators satisfy
\begin{equation}
    \mathcal G_{L_\star}^{\mathrm{sca}}
    -
    \mathcal G_L^{\mathrm{sca}}
    =
    \mathcal O_{L_\star}
    \mathbf D_{L\mid L_\star}
    \mathcal S_{L_\star}.
    \label{eq:supp_nested_operator_difference}
\end{equation}
Applying the finite metric realization from the main paper gives
\begin{equation}
    \left\|
    \mathcal G_{L_\star}^{\mathrm{sca}}
    -
    \mathcal G_L^{\mathrm{sca}}
    \right\|_{\mathcal X\to\mathcal Y}
    =
    \sigma_1
    \left(
    \mathbf C_{O,L_\star}^{1/2}
    \mathbf D_{L\mid L_\star}
    \mathbf C_{S,L_\star}^{1/2}
    \right).
    \label{eq:supp_nested_difference_metric_core}
\end{equation}
The same result with $\mathbf D_{L\mid L_\star}$ replaced by $\mathbf K_{L_\star}$ gives the reference norm. Therefore
\begin{equation}
    e_{L\mid L_\star}
    =
    \frac{
        \sigma_1
        \left(
        \mathbf C_{O,L_\star}^{1/2}
        \mathbf D_{L\mid L_\star}
        \mathbf C_{S,L_\star}^{1/2}
        \right)}{
        \sigma_1
        \left(
        \mathbf C_{O,L_\star}^{1/2}
        \mathbf K_{L_\star}
        \mathbf C_{S,L_\star}^{1/2}
        \right)},
    \label{eq:supp_finite_reference_error_formula}
\end{equation}
provided the denominator is nonzero. Setting $L_\star=L+1$ gives the adjacent-order indicator $\widehat e_L$.

\subsection{Source and observation tail ratios}

Define the complementary reference-coordinate selector
\begin{equation}
    \mathbf E_{L\mid L_\star}^{\perp}
    :=
    \mathbf I-\mathbf E_{L\mid L_\star}.
    \label{eq:supp_high_order_reference_selector}
\end{equation}
Let
\begin{equation}
    \mathbf G_{r,L_\star}
    :=
    \bigoplus_i\mathbf G_{i,L_\star}^{\mathrm{reg}},
    \quad
    \mathbf G_{o,L_\star}
    :=
    \bigoplus_i\mathbf G_{i,L_\star}^{\mathrm{out}}
    \label{eq:supp_reference_trace_grams}
\end{equation}
be the positive diagonal trace Gramians from \eqref{eq:supp_finite_trace_gram_entries}. They share the nested modal ordering with $\mathbf E_{L\mid L_\star}^{\perp}$ and therefore commute with it. Introduce
\begin{align}
    \widetilde{\mathbf C}_{S,L_\star}
     & :=
    \mathbf G_{r,L_\star}^{1/2}
    \mathbf C_{S,L_\star}
    \mathbf G_{r,L_\star}^{1/2},
    \label{eq:supp_trace_normalized_source_gram} \\
    \widetilde{\mathbf C}_{O,L_\star}
     & :=
    \mathbf G_{o,L_\star}^{-1/2}
    \mathbf C_{O,L_\star}
    \mathbf G_{o,L_\star}^{-1/2}.
    \label{eq:supp_trace_normalized_observation_gram}
\end{align}

For the source side, synthesis by $P_{L_\star}^{\mathrm{reg}}$ gives
\begin{equation}
    \left\|
    P_{L_\star}^{\mathrm{reg}}\mathbf a
    \right\|_{\mathcal H^{\mathrm{reg}}}^2
    =
    \mathbf a^{\dagger}
    \mathbf G_{r,L_\star}
    \mathbf a.
    \label{eq:supp_regular_coordinate_trace_norm}
\end{equation}
Using $\mathbf C_{S,L_\star}=\mathcal S_{L_\star}\mathcal S_{L_\star}^{\dagger}$, the squared numerator of the source-tail ratio is
\begin{equation}
    \left\|
    (\Pi_{L_\star}^{\mathrm{reg}}-\Pi_L^{\mathrm{reg}})
    \mathcal S
    \right\|^2
    =
    \lambda_{\max}
    \left(
    \mathbf E_{L\mid L_\star}^{\perp}
    \widetilde{\mathbf C}_{S,L_\star}
    \mathbf E_{L\mid L_\star}^{\perp}
    \right).
    \label{eq:supp_source_tail_numerator}
\end{equation}
The positive square roots on the source side appear because the source map produces raw regular-trace coefficients whose range norm is weighted by $\mathbf G_{r,L_\star}$.

For the observation side, an outgoing coefficient vector represents a unit trace after the whitening $\mathbf a=\mathbf G_{o,L_\star}^{-1/2}\mathbf x$ with $\|\mathbf x\|_2=1$. Since $\mathbf C_{O,L_\star}=\mathcal O_{L_\star}^{\dagger}\mathcal O_{L_\star}$,
\begin{equation}
    \left\|
    \mathcal O
    (\Pi_{L_\star}^{\mathrm{out}}-\Pi_L^{\mathrm{out}})
    \right\|^2
    =
    \lambda_{\max}
    \left(
    \mathbf E_{L\mid L_\star}^{\perp}
    \widetilde{\mathbf C}_{O,L_\star}
    \mathbf E_{L\mid L_\star}^{\perp}
    \right).
    \label{eq:supp_observation_tail_numerator}
\end{equation}
This explains the inverse square roots on the observation side: its coefficient vector belongs to the domain rather than the range of the map being measured.

Removing the two selectors in \eqref{eq:supp_source_tail_numerator} and \eqref{eq:supp_observation_tail_numerator} gives the corresponding squared denominators. Hence, with $\bullet$ standing for either $S$ or $O$,
\begin{equation}
    \kappa_{\bullet,L\mid L_\star}
    =
    \left[
        \frac{
            \lambda_{\max}
            \left(
            \mathbf E_{L\mid L_\star}^{\perp}
            \widetilde{\mathbf C}_{\bullet,L_\star}
            \mathbf E_{L\mid L_\star}^{\perp}
            \right)}{
            \lambda_{\max}
            \left(
            \widetilde{\mathbf C}_{\bullet,L_\star}
            \right)}
        \right]^{1/2}.
    \label{eq:supp_finite_reference_tail_formula}
\end{equation}
When the denominator is nonzero, positivity and compression by the orthogonal selector imply $0\leq\kappa_{\bullet,L\mid L_\star}\leq1$. The compact expression in \eqref{eq:supp_finite_reference_tail_formula} relies on the commutation of the selector and trace Gramian; it is not valid for an arbitrary non-diagonal metric without retaining the factors in their original order.

\subsection{Trace-normalized resolvent and angular content}

Let
\begin{equation}
    \mathbf R_L
    =
    (\mathbf I-\mathbf W_L\mathbf T_L)^{-1}.
    \label{eq:supp_retained_resolvent}
\end{equation}
For a regular-trace coefficient vector $\mathbf a$, the fixed trace norm is $\|\mathbf a\|_{r,L}^2=\mathbf a^{\dagger}\mathbf G_{r,L}\mathbf a$. Its induced resolvent norm is therefore
\begin{equation}
    \widetilde{\mathbf R}_L
    :=
    \mathbf G_{r,L}^{1/2}
    \mathbf R_L
    \mathbf G_{r,L}^{-1/2},
    \quad
    \sup_{\mathbf a\ne\mathbf0}
    \frac{\|\mathbf R_L\mathbf a\|_{r,L}}
    {\|\mathbf a\|_{r,L}}
    =
    \|\widetilde{\mathbf R}_L\|_2
    =
    \rho_L.
    \label{eq:supp_resolvent_induced_trace_norm}
\end{equation}
Define the metric-normalized collective matrix
\begin{equation}
    \widetilde{\mathbf M}_L
    :=
    \mathbf G_{r,L}^{1/2}
    (\mathbf I-\mathbf W_L\mathbf T_L)
    \mathbf G_{r,L}^{-1/2}.
    \label{eq:supp_metric_normalized_collective_matrix}
\end{equation}
Similarity of the inverse gives
\begin{equation}
    \mathbf G_{r,L}^{1/2}
    \mathbf R_L
    \mathbf G_{r,L}^{-1/2}
    =
    \widetilde{\mathbf M}_L^{-1},
    \quad
    \rho_L
    =
    \sigma_{\min}^{-1}(\widetilde{\mathbf M}_L).
    \label{eq:supp_resolvent_smallest_singular_value}
\end{equation}
The lifted resolvent on the full regular-trace space has the orthogonal block form
\begin{equation}
    \mathcal R_L^\sharp
    =
    \mathbb I-\Pi_L^{\mathrm{reg}}
    +
    P_L^{\mathrm{reg}}
    \mathbf R_L
    Q_L^{\mathrm{reg}}.
    \label{eq:supp_lifted_resolvent_block_form}
\end{equation}
The first block is the identity on the omitted orthogonal complement and the second is the retained resolvent. Consequently,
\begin{equation}
    \|\mathcal R_L^\sharp\|
    =
    \max\{1,\rho_L\}.
    \label{eq:supp_lifted_resolvent_norm}
\end{equation}

Finally, let
\begin{equation}
    \widetilde{\mathbf M}_L
    =
    \mathbf U_L^{\mathrm c}
    \mathbf\Sigma_L^{\mathrm c}
    (\mathbf V_L^{\mathrm c})^{\dagger}
    \label{eq:supp_collective_matrix_svd}
\end{equation}
be an SVD. Order its singular values as
$\varsigma_{1,L}\geq\cdots\geq\varsigma_{N_L,L}>0$, write
$\mathbf U_L^{\mathrm c}=[\mathbf u_{1,L}^{\mathrm c},\ldots,\mathbf u_{N_L,L}^{\mathrm c}]$, and let $d_L^{\min}$ be the multiplicity of the smallest singular value, or numerically the dimension of the isolated cluster representing it. The superscript ``$\mathrm c$'' distinguishes these collective-matrix SVD factors from the accessible-channel SVD used in \hyperref[supp:accessible_channel_diagnostics]{Supplementary Note~\ref*{supp:accessible_channel_diagnostics}}. Define
\begin{align}
    \mathbf U_{\min,L}^{\mathrm c}
     & :=
    [\mathbf u_{N_L-d_L^{\min}+1,L}^{\mathrm c},\ldots,
    \mathbf u_{N_L,L}^{\mathrm c}],
    \label{eq:supp_resolvent_maximizing_input_subspace} \\
    \mathbf P_{\min,L}^{\mathrm{in}}
     & :=
    \mathbf U_{\min,L}^{\mathrm c}
    (\mathbf U_{\min,L}^{\mathrm c})^{\dagger}.
    \label{eq:supp_resolvent_maximizing_input_projector}
\end{align}
Since $\widetilde{\mathbf M}_L^{-1}=\mathbf V_L^{\mathrm c}(\mathbf\Sigma_L^{\mathrm c})^{-1}(\mathbf U_L^{\mathrm c})^{\dagger}$, $\operatorname{ran}\mathbf U_{\min,L}^{\mathrm c}$ is the right singular subspace of $\widetilde{\mathbf R}_L$ associated with its largest singular value $\rho_L$. Although the orthonormal basis $\mathbf U_{\min,L}^{\mathrm c}$ is nonunique when $d_L^{\min}>1$, the projector $\mathbf P_{\min,L}^{\mathrm{in}}$ is basis invariant. Let $\mathcal P_l$ collect the cluster indices of angular order $l$:
\begin{equation}
    \mathcal P_l
    =
    \{\zeta_L(i,\xi_L(\tau,l,m))\mid 1\leq i\leq N, \tau =1,2, -l\leq m \leq l\}.
    \label{eq:supp_resolvent_angular_index_set}
\end{equation}
The basis-invariant trace-normalized fraction of this input subspace at angular order $l$ is
\begin{equation}
    w_{l,L}^{\mathrm{in}}
    :=
    \frac{1}{d_L^{\min}}
    \sum_{p\in\mathcal P_l}
    [\mathbf P_{\min,L}^{\mathrm{in}}]_{pp}.
    \label{eq:supp_resolvent_input_angular_fraction}
\end{equation}
For each $L$, these fractions sum to one because $\operatorname{tr}\mathbf P_{\min,L}^{\mathrm{in}}=d_L^{\min}$. In the reported near-resonant systems, the two smallest singular values agree relatively within $6\times10^{-13}$ at $L=7$ and $4\times10^{-12}$ at $L=14$, whereas the third smallest is larger than the minimum by factors $15.97$ and $95.96$, respectively. Thus the dominant input singular subspace is numerically isolated and two-dimensional at both orders. Its $l=7$ fractions are
$w_{7,7}^{\mathrm{in}}=0.994$ and $w_{7,14}^{\mathrm{in}}=0.993$, which are the two percentages reported in Section~IV-B of the main paper.

\section{Finite-Modal Realization and Diagnostics of Accessible Channels}
\label{supp:accessible_channel_diagnostics}

This note gives the numerical realization and the subspace diagnostics used in Section~IV-C of the main paper. The Hilbert-adjoint and channel-recovery formulas themselves are proved in \hyperref[supp:surface_interface_grams]{Supplementary Note~\ref*{supp:surface_interface_grams}}. Here we explain how those formulas are evaluated without first forming the Gram matrices and their square roots, how the leading doublets and their angular content are identified, and how the reconstructed channels are checked by an independent finite-order forward solution.

\subsection{QR realization of the metric core}

Use the semianalytic external-sphere construction of \hyperref[supp:finite_reference_diagnostics]{Supplementary Note~\ref*{supp:finite_reference_diagnostics}} at $L_c=44$. The metric-normalized source coordinate vector has two blocks, one for each of $\mathbf J_t$ and $\mathbf M_t$, and therefore belongs to $\mathbb C^{2U_{L_c}}$; the metric-normalized tangential-electric-field coordinate vector belongs to $\mathbb C^{U_{L_c}}$. The observation rows are indexed by the common-center mode $\nu=(\tau,l,m)\in\Lambda_{L_c}$. Each source block uses the conjugated Riesz basis fixed by \eqref{eq:supp_source_adjoint_electric_component} and \eqref{eq:supp_source_adjoint_magnetic_component}; its rows are labelled by the physical azimuthal order of the resulting current basis function. Equivalently, this physical label is the negative of the un-conjugated VSWF index before conjugation. This sign reversal leaves the symmetric $\pm m$ sectors used below unchanged.

In these orthonormal coordinates, let the coordinate-lift matrices $\mathbf L_{S,L}$ and $\mathbf L_{O,L}$ represent $\mathcal S_L^{\dagger}$ and $\mathcal O_L$, respectively. Thus $\mathbf L_{S,L}\mathbf h$ is the coordinate vector of $\mathcal S_L^{\dagger}\mathbf h$, while $\mathbf L_{O,L}\mathbf h$ is the coordinate vector of $\mathcal O_L\mathbf h$. Their columns are the coordinate representatives generated by the retained cluster modes, so
\begin{equation}
    \mathbf C_{S,L}=\mathbf L_{S,L}^{\dagger}\mathbf L_{S,L},
    \quad
    \mathbf C_{O,L}=\mathbf L_{O,L}^{\dagger}\mathbf L_{O,L}.
    \label{eq:supp_channel_factor_grams}
\end{equation}
Reduced QR factorizations give
\begin{equation}
    \mathbf L_{S,L}=\mathbf Q_{S,L}\mathbf R_{S,L},
    \quad
    \mathbf L_{O,L}=\mathbf Q_{O,L}\mathbf R_{O,L},
    \label{eq:supp_channel_factor_qr}
\end{equation}
where the columns of $\mathbf Q_{S,L}$ and $\mathbf Q_{O,L}$ are orthonormal. The external-coordinate matrix of the retained scattering operator then factors as
\begin{equation}
    \mathbf L_{O,L}\mathbf K_L\mathbf L_{S,L}^{\dagger}
    =
    \mathbf Q_{O,L}\widehat{\mathbf B}_L\mathbf Q_{S,L}^{\dagger},
    \quad
    \widehat{\mathbf B}_L
    :=
    \mathbf R_{O,L}\mathbf K_L\mathbf R_{S,L}^{\dagger}.
    \label{eq:supp_channel_reduced_core}
\end{equation}
For the reported external surfaces, both factors have full retained column rank at every order used below. Since $\mathbf R_{S,L}^{\dagger}\mathbf R_{S,L}=\mathbf C_{S,L}$ and $\mathbf R_{O,L}^{\dagger}\mathbf R_{O,L}=\mathbf C_{O,L}$, their polar decompositions give unitary matrices $\mathbf Z_{S,L}$ and $\mathbf Z_{O,L}$ such that
\begin{equation}
    \mathbf R_{S,L}
    =
    \mathbf Z_{S,L}\mathbf C_{S,L}^{1/2},
    \quad
    \mathbf R_{O,L}
    =
    \mathbf Z_{O,L}\mathbf C_{O,L}^{1/2},
    \label{eq:supp_channel_factor_polar_forms}
\end{equation}
and therefore
\begin{equation}
    \widehat{\mathbf B}_L
    =
    \mathbf Z_{O,L}\mathbf B_L\mathbf Z_{S,L}^{\dagger}.
    \label{eq:supp_channel_core_unitary_equivalence}
\end{equation}
Thus $\widehat{\mathbf B}_L$ has the same nonzero singular values as the metric core. This QR form avoids forming normal equations while preserving the Hilbert metrics exactly within the $L_c$ realization.

Let
\begin{equation}
    \widehat{\mathbf B}_L
    =
    \mathbf U_L\mathbf\Sigma_L\mathbf V_L^{\dagger}.
    \label{eq:supp_channel_reduced_core_svd}
\end{equation}
The accessible-channel corollary in Section~III-E of the main paper defines the source $\mathbf q_n$ and the unit-normalized observation field $\widehat{\mathbf E}_{n}^{\mathrm{sca}}$. Their external-coordinate representatives are
\begin{equation}
    \mathbf q_{n,L}=\mathbf Q_{S,L}\mathbf v_{n,L},
    \quad
    \widehat{\mathbf e}_{n,L}^{\mathrm{sca}}
    =\mathbf Q_{O,L}\mathbf u_{n,L}.
    \label{eq:supp_channel_qr_recovery}
\end{equation}
Here the lowercase $\widehat{\mathbf e}_{n,L}^{\mathrm{sca}}\in\mathbb C^{U_{L_c}}$ is the finite external-coordinate vector representing the uppercase field $\widehat{\mathbf E}_{n}^{\mathrm{sca}}\in\mathcal Y$; it is not a second observation field. Likewise, $\mathbf q_{n,L}\in\mathbb C^{2U_{L_c}}$ is the coordinate representation of the corresponding source $\mathbf q_n\in\mathcal X$. Within the $L_c$ realization, both coordinate identifications are isometric. Equations~\eqref{eq:supp_channel_reduced_core} and \eqref{eq:supp_channel_reduced_core_svd} therefore give unit norms, mutual orthogonality, and
\begin{equation}
    \mathbf L_{O,L}\mathbf K_L\mathbf L_{S,L}^{\dagger}
    \mathbf q_{n,L}
    =
    \sigma_{n,L}\widehat{\mathbf e}_{n,L}^{\mathrm{sca}}.
    \label{eq:supp_channel_qr_action}
\end{equation}
At $L=14$, the largest source- and observation-orthonormality residuals over the four retained diagnostic channels are $2.2\times10^{-15}$ and $1.7\times10^{-15}$, respectively. The relative Frobenius residual of \eqref{eq:supp_channel_qr_action} is $2.3\times10^{-15}$ near resonance and $1.2\times10^{-15}$ off resonance.

\subsection{Leading doublets, azimuthal sectors, and order stability}

Two identifications are needed before interpreting the leading channels. First, we must determine whether the largest singular values form a numerical doublet. Second, because singular vectors within a degenerate subspace are not unique, the azimuthal label must be assigned to the two-dimensional subspace rather than to a particular SVD basis.

For the first step, define the numerical leading cluster consecutively from the top of the spectrum. Its dimension $d_L$ is the largest integer for which
\begin{equation}
    \frac{|\sigma_{j,L}-\sigma_{1,L}|}{\sigma_{1,L}}
    \leq
    10^{-5},
    \quad 1\leq j\leq d_L,
    \label{eq:supp_channel_degeneracy_criterion}
\end{equation}
and the first singular value violating \eqref{eq:supp_channel_degeneracy_criterion} ends the cluster. At $L=14$, the relative splits between $\sigma_1$ and $\sigma_2$ are $1.23\times10^{-9}$ near resonance and $4.92\times10^{-12}$ off resonance, whereas $\sigma_3$ lies outside the cluster in both cases. Hence $d_{14}=2$ for both material states. The ratios $\sigma_2/\sigma_3=19.976$ near resonance and $1.022$ off resonance further show that only the near-resonant doublet is strongly separated from the remaining spectrum.

For the second step, let the orthogonal projectors onto the recovered leading source and observation subspaces be
\begin{equation}
    \mathbf P_{\mathrm{ch},L}^{X}
    :=
    \sum_{n=1}^{2}\mathbf q_{n,L}\mathbf q_{n,L}^{\dagger},
    \quad
    \mathbf P_{\mathrm{ch},L}^{Y}
    :=
    \sum_{n=1}^{2}
    \widehat{\mathbf e}_{n,L}^{\mathrm{sca}}
    (\widehat{\mathbf e}_{n,L}^{\mathrm{sca}})^{\dagger}.
    \label{eq:supp_channel_leading_subspace_projectors}
\end{equation}
These projectors are unchanged by any unitary mixing of the two singular vectors and therefore represent the doublet itself rather than the particular basis returned by the SVD.

To define the azimuthal selectors explicitly, let
\begin{equation}
    \Lambda_{L_c}(\mathcal M)
    :=
    \{(\tau,l,m)\in\Lambda_{L_c}:m\in\mathcal M\}.
    \label{eq:supp_channel_azimuthal_index_set}
\end{equation}
Let $\boldsymbol\delta_{b,\nu}^{X}$ be the standard unit vector of the source coordinate space at current block $b\in\{J,M\}$ and mode $\nu\in\Lambda_{L_c}$, and let $\boldsymbol\delta_{\nu}^{Y}$ be the corresponding unit vector of the observation coordinate space. The orthogonal coordinate projectors onto the selected azimuthal sectors are
\begin{align}
    \mathbf P_{\mathcal M}^{X}
     & :=
    \sum_{b\in\{J,M\}}
    \sum_{\nu\in\Lambda_{L_c}(\mathcal M)}
    \boldsymbol\delta_{b,\nu}^{X}
    (\boldsymbol\delta_{b,\nu}^{X})^{\dagger},
    \label{eq:supp_channel_source_azimuthal_projector} \\
    \mathbf P_{\mathcal M}^{Y}
     & :=
    \sum_{\nu\in\Lambda_{L_c}(\mathcal M)}
    \boldsymbol\delta_{\nu}^{Y}
    (\boldsymbol\delta_{\nu}^{Y})^{\dagger}.
    \label{eq:supp_channel_observation_azimuthal_projector}
\end{align}
Thus $\mathbf P_{\mathcal M}^{X}$ retains the same $(\tau,l,m)$ sectors in both current blocks, whereas $\mathbf P_{\mathcal M}^{Y}$ acts on the single tangential-electric-field block. For $s\in\{X,Y\}$, define the fraction of the leading doublet contained in those sectors by
\begin{equation}
    \chi_{\mathcal M,L}^{s}
    :=
    \frac12
    \operatorname{tr}
    \left(
    \mathbf P_{\mathcal M}^{s}
    \mathbf P_{\mathrm{ch},L}^{s}
    \right).
    \label{eq:supp_channel_azimuthal_fraction}
\end{equation}
The normalization by two is the dimension of the doublet, so $0\leq\chi_{\mathcal M,L}^{s}\leq1$; equality to one means that the entire two-dimensional subspace lies in the selected sectors. At $L=14$, choosing $\mathcal M=\{-2,+2\}$ in the near-resonant state and $\mathcal M=\{-1,+1\}$ in the off-resonant state gives $\chi_{\mathcal M,14}^{X}=\chi_{\mathcal M,14}^{Y}=1$ to within $8\times10^{-16}$. The leading doublets are therefore confined to $m=\pm2$ and $m=\pm1$, respectively. This establishes the azimuthal labels used in the accessible-channel spectra of the main paper without assigning physical significance to an arbitrary basis within either doublet.

We also compare the $L=14$ doublets with $L=18$ in the common $L_c=44$ external coordinates. For two $d$-dimensional subspaces with semiunitary basis matrices $\mathbf N_1$ and $\mathbf N_2$, the reported minimum principal cosine is
\begin{equation}
    c_{\min}(\mathbf N_1,\mathbf N_2)
    :=
    \sigma_{\min}(\mathbf N_1^{\dagger}\mathbf N_2).
    \label{eq:supp_channel_minimum_principal_cosine}
\end{equation}
Near resonance, the relative change of $\sigma_1$ from $L=14$ to $18$ is $5.18\times10^{-3}$, while the source and observation minimum principal cosines are $0.9999999998$ and $0.9999999994$. Off resonance, the corresponding singular-value change is $4.63\times10^{-9}$ and both cosines exceed $0.9999999999$. Thus the channel subspaces used for the $L=14$ interpretation are already stable even though the near-resonant gain retains the small finite-order change quantified in Section~IV-B.

\subsection{Basis-invariant angular trace fractions}

Fix a retained order $L$; the main-paper results use $L=14$. As in Section~IV-C of the main paper, set $\Delta\Pi_l^s:=\Pi_l^s-\Pi_{l-1}^s$ for $s\in\{\mathrm{reg},\mathrm{out}\}$, with $\Pi_0^s=0$. For each member of the leading doublet, set
\begin{align}
    \mathbf a_n^{\mathrm{exc}}
     & :=
    \mathcal S_L\mathbf q_n,
     &
    \mathbf f_n^{\mathrm{exc}}
     & :=
    P_L^{\mathrm{reg}}\mathbf a_n^{\mathrm{exc}},
    \label{eq:supp_channel_excited_trace_coordinates} \\
    \mathbf a_n^{\mathrm{sca}}
     & :=
    \mathbf K_L\mathbf a_n^{\mathrm{exc}},
     &
    \mathbf f_n^{\mathrm{sca}}
     & :=
    P_L^{\mathrm{out}}\mathbf a_n^{\mathrm{sca}}.
    \label{eq:supp_channel_scattered_trace_coordinates}
\end{align}
The two fractions defined in Section~IV-C of the main paper are therefore
\begin{align}
    w_l^{\mathrm{exc}}
     & :=
    \frac{
    \sum_{n=1}^{2}
    \|\Delta\Pi_l^{\mathrm{reg}}\mathbf f_n^{\mathrm{exc}}\|_{\mathrm{reg}}^2}
    {
    \sum_{n=1}^{2}
    \|\mathbf f_n^{\mathrm{exc}}\|_{\mathrm{reg}}^2},
    \label{eq:supp_channel_excitation_fraction_definition} \\
    w_l^{\mathrm{sca}}
     & :=
    \frac{
    \sum_{n=1}^{2}
    \|\Delta\Pi_l^{\mathrm{out}}\mathbf f_n^{\mathrm{sca}}\|_{\mathrm{out}}^2}
    {
    \sum_{n=1}^{2}
    \|\mathbf f_n^{\mathrm{sca}}\|_{\mathrm{out}}^2}.
    \label{eq:supp_channel_scattered_fraction_definition}
\end{align}

We now reduce these trace-space definitions to finite-coordinate expressions. Stack the two coefficient vectors as
\begin{equation}
    \mathbf A_L^{\mathrm{exc}}
    :=
    [\mathbf a_1^{\mathrm{exc}},\mathbf a_2^{\mathrm{exc}}],
    \quad
    \mathbf A_L^{\mathrm{sca}}
    :=
    [\mathbf a_1^{\mathrm{sca}},\mathbf a_2^{\mathrm{sca}}]
    =
    \mathbf K_L\mathbf A_L^{\mathrm{exc}}.
    \label{eq:supp_channel_internal_trace_matrices}
\end{equation}
As in \hyperref[supp:finite_reference_diagnostics]{Supplementary Note~\ref*{supp:finite_reference_diagnostics}}, the cluster indices of angular order $l$ are
\begin{equation}
    \mathcal P_l
    =
    \{\zeta_L(i,\xi_L(\tau,l,m))\mid 1\leq i\leq N, \tau=1,2, -l\leq m\leq l\}.
    \label{eq:supp_channel_trace_index_set}
\end{equation}
Define the exact-order coordinate selector $\mathbf E_{l,L}\in\mathbb C^{N_L\times N_L}$ explicitly by
\begin{equation}
    [\mathbf E_{l,L}]_{pq}
    :=
    \begin{cases}
        1, & p=q\in\mathcal P_l, \\
        0, & \text{otherwise},
    \end{cases}
    \quad 1\leq p,q\leq N_L.
    \label{eq:supp_channel_trace_order_selector}
\end{equation}
The sets $\mathcal P_l$, $1\leq l\leq L$, partition the retained cluster coordinates, and hence
\begin{equation}
    \mathbf E_{l,L}^{\dagger}
    =
    \mathbf E_{l,L},
    \quad
    \mathbf E_{l,L}\mathbf E_{l',L}
    =
    \delta_{ll'}\mathbf E_{l,L},
    \quad
    \sum_{l=1}^{L}\mathbf E_{l,L}
    =
    \mathbf I_{N_L}.
    \label{eq:supp_channel_trace_order_selector_properties}
\end{equation}
Also let
\begin{equation}
    \mathbf G_{r,L}
    :=
    (P_L^{\mathrm{reg}})^{\dagger}P_L^{\mathrm{reg}},
    \quad
    \mathbf G_{o,L}
    :=
    (P_L^{\mathrm{out}})^{\dagger}P_L^{\mathrm{out}}
    \label{eq:supp_channel_stacked_trace_grams}
\end{equation}
be the stacked fixed-trace Gramians. Because the VSWF traces are orthogonal in the fixed metric, applying the exact-order projection to a retained synthesis selects precisely the coefficient rows in $\mathcal P_l$:
\begin{equation}
    \Delta\Pi_l^{s}P_L^{s}
    =
    P_L^{s}\mathbf E_{l,L},
    \quad
    s\in\{\mathrm{reg},\mathrm{out}\}.
    \label{eq:supp_channel_projection_selector_identity}
\end{equation}
The Gramians are diagonal in $(\tau,l,m)$ and hence commute with $\mathbf E_{l,L}$. For either branch, \eqref{eq:supp_channel_projection_selector_identity} gives
\begin{equation}
    \|\Delta\Pi_l^{s}P_L^{s}\mathbf a\|_{s}^{2}
    =
    \mathbf a^{\dagger}\mathbf E_{l,L}\mathbf G_{s,L}
    \mathbf E_{l,L}\mathbf a
    =
    \|\mathbf E_{l,L}\mathbf G_{s,L}^{1/2}\mathbf a\|_2^2,
    \label{eq:supp_channel_projected_trace_norm}
\end{equation}
where $\mathbf G_{\mathrm{reg},L}:=\mathbf G_{r,L}$ and $\mathbf G_{\mathrm{out},L}:=\mathbf G_{o,L}$. Also set $\mathbf A_L^{\mathrm{reg}}:=\mathbf A_L^{\mathrm{exc}}$, $\mathbf A_L^{\mathrm{out}}:=\mathbf A_L^{\mathrm{sca}}$, $\mathbf f_n^{\mathrm{reg}}:=\mathbf f_n^{\mathrm{exc}}$, and $\mathbf f_n^{\mathrm{out}}:=\mathbf f_n^{\mathrm{sca}}$. Summing \eqref{eq:supp_channel_projected_trace_norm} over the two columns gives, for $s\in\{\mathrm{reg},\mathrm{out}\}$,
\begin{align}
    \sum_{n=1}^{2}
    \|\Delta\Pi_l^{s}\mathbf f_n^{s}\|_{s}^{2}
     & =
    \|\mathbf E_{l,L}\mathbf G_{s,L}^{1/2}
    \mathbf A_L^{s}\|_{\mathrm F}^{2},
    \label{eq:supp_channel_projected_trace_frobenius} \\
    \sum_{n=1}^{2}
    \|\mathbf f_n^{s}\|_{s}^{2}
     & =
    \|\mathbf G_{s,L}^{1/2}\mathbf A_L^{s}\|_{\mathrm F}^{2}.
    \label{eq:supp_channel_total_trace_frobenius}
\end{align}
Substitution of these two identities into \eqref{eq:supp_channel_excitation_fraction_definition} and \eqref{eq:supp_channel_scattered_fraction_definition} yields the computable formulas
\begin{align}
    w_l^{\mathrm{exc}}
     & =
    \frac{
    \|\mathbf E_{l,L}\mathbf G_{r,L}^{1/2}
    \mathbf A_L^{\mathrm{exc}}\|_{\mathrm F}^{2}}
    {\|\mathbf G_{r,L}^{1/2}\mathbf A_L^{\mathrm{exc}}\|_{\mathrm F}^{2}},
    \label{eq:supp_channel_excitation_angular_fraction} \\
    w_l^{\mathrm{sca}}
     & =
    \frac{
    \|\mathbf E_{l,L}\mathbf G_{o,L}^{1/2}
    \mathbf A_L^{\mathrm{sca}}\|_{\mathrm F}^{2}}
    {\|\mathbf G_{o,L}^{1/2}\mathbf A_L^{\mathrm{sca}}\|_{\mathrm F}^{2}}.
    \label{eq:supp_channel_scattered_angular_fraction}
\end{align}
Thus the finite formulas are exactly the coordinate evaluations of the two trace-space definitions in the main paper. Since the mutually orthogonal selectors $\mathbf E_{l,L}$ resolve the retained identity, each family of fractions sums to one.

If the orthonormal basis of the leading doublet is changed by a unitary matrix $\mathbf U$, then both internal matrices in \eqref{eq:supp_channel_internal_trace_matrices} are right-multiplied by $\mathbf U$. The identity $\|\mathbf A\mathbf U\|_{\mathrm F}=\|\mathbf A\|_{\mathrm F}$ proves the claimed basis invariance of both the numerator and denominator. The reported fractions are therefore properties of the two-dimensional accessible subspace, not of an arbitrary singular-vector basis. They are marginal trace contents on the two sides of the coherent map $\mathbf K_L$; no incoherent transfer of energy from one angular order to another is asserted.

At $L=14$, summing \eqref{eq:supp_channel_excitation_angular_fraction} over $l=5,\ldots,8$ gives $0.7960$ near resonance, while \eqref{eq:supp_channel_scattered_angular_fraction} gives $w_7^{\mathrm{sca}}=0.9925$. Off resonance, $w_1^{\mathrm{exc}}=0.9635$ and $w_1^{\mathrm{sca}}=0.9531$. These values are the percentages reported in Section~IV-C of the main paper.

\subsection{Displayed representative and independent forward check}

Within each numerically degenerate doublet, the phase of each recovered source is first fixed by making its largest-magnitude external source coordinate real and positive. Suppressing the fixed order $L=14$ in this paragraph, the displayed standing-wave source is then
\begin{equation}
    \mathbf q_{\mathrm{dom}}
    =
    \frac{\mathbf q_1+\mathbf q_2}{\sqrt2}.
    \label{eq:supp_channel_display_source}
\end{equation}
The same phases and linear combination are applied to the singular-value-weighted outputs. Thus the output paired with \eqref{eq:supp_channel_display_source} is
\begin{equation}
    \mathbf e_{\mathrm{dom}}^{\mathrm{sca}}
    =
    \frac{
    \sigma_1\widehat{\mathbf e}_1^{\mathrm{sca}}
    +
    \sigma_2\widehat{\mathbf e}_2^{\mathrm{sca}}}
    {\sqrt2},
    \quad
    \|\mathbf e_{\mathrm{dom}}^{\mathrm{sca}}\|_2
    =
    \left(\frac{\sigma_1^2+\sigma_2^2}{2}\right)^{1/2}.
    \label{eq:supp_channel_display_output}
\end{equation}
After forming \eqref{eq:supp_channel_display_source} and \eqref{eq:supp_channel_display_output}, the same final unit-modulus phase is applied to both so that the largest-magnitude coordinate of $\mathbf q_{\mathrm{dom}}$ is real and positive. This common phase leaves the norm in \eqref{eq:supp_channel_display_output} and every subspace-level conclusion unchanged. The construction selects one standing-wave representative for visualization; it is not an additional basis-invariant channel assertion.

Under the isometric coordinate identification in \hyperref[supp:accessible_channel_diagnostics]{Supplementary Note~\ref*{supp:accessible_channel_diagnostics}-A}, the same symbols $\mathbf q_n$ and $\mathbf q_{\mathrm{dom}}$ denote a source in $\mathcal X$ and its external coordinate vector; the meaning is fixed by the operator or matrix acting on it. This convention lets the following coefficient-level check be compared directly with the recovered continuous channel.

For the independent forward check, the stored action of $\mathbf K_{14}$ is replaced by the explicit retained multiple-scattering path
\begin{align}
    \mathbf a^{\mathrm{ext}}
     & =
    \mathbf L_{S,14}^{\dagger}\mathbf q_{\mathrm{dom}},
    \label{eq:supp_channel_forward_excitation} \\
    (\mathbf I-\mathbf W_{14}\mathbf T_{14})\mathbf a^{\mathrm{inc}}
     & =
    \mathbf a^{\mathrm{ext}},
    \label{eq:supp_channel_forward_solve}      \\
    \mathbf a^{\mathrm{sca}}
     & =
    \mathbf T_{14}\mathbf a^{\mathrm{inc}},
    \label{eq:supp_channel_forward_scattering} \\
    \mathbf e_{\mathrm{fwd}}^{\mathrm{sca}}
     & =
    \mathbf L_{O,14}\mathbf a^{\mathrm{sca}}.
    \label{eq:supp_channel_forward_observation}
\end{align}
Here $\mathbf a^{\mathrm{ext}}$, $\mathbf a^{\mathrm{inc}}$, and $\mathbf a^{\mathrm{sca}}$ are, respectively, the external regular, exciting regular, and scattered outgoing cluster coefficients. The collective equation is solved in the fixed regular-trace metric by GMRES with relative tolerance $10^{-12}$. Comparing $\mathbf e_{\mathrm{fwd}}^{\mathrm{sca}}$ from \eqref{eq:supp_channel_forward_excitation}--\eqref{eq:supp_channel_forward_observation} with \eqref{eq:supp_channel_display_output} gives relative observation-norm errors $9.370\times10^{-10}$ near resonance and $9.150\times10^{-12}$ off resonance. The corresponding trace-balanced equation residuals are $1.06\times10^{-13}$ and $1.16\times10^{-14}$.

Finally, Eq.~\eqref{eq:supp_channel_minimum_principal_cosine} is used to test whether the off-resonant leading low-order family remains present after the near-resonant reordering. Comparing the near-resonant channel-$3,4$ subspaces with the off-resonant channel-$1,2$ subspaces gives $c_{\min}=0.9941$ in $\mathcal X$ and $0.9967$ in $\mathcal Y$. This basis-invariant comparison supports the physical conclusion in the main paper: the collective resonance inserts a distinct high-order leading doublet while leaving the original low-order source--observation family almost unchanged beneath it.

\bibliographystyle{IEEEtran}
\bibliography{ref}